\documentclass[letterpaper,pra,,aps,twocolumn,superscriptaddress,floatfix,longbibliography,nofootinbib]{revtex4-2}

\usepackage[whole]{bxcjkjatype}
\usepackage{graphicx}
\usepackage[colorlinks,linkcolor=blue,urlcolor=blue,citecolor=blue,hypertexnames=false]{hyperref}
\usepackage{etoolbox}
\appto\UrlBreaks{\do\-}
\usepackage{amsmath, amssymb, amsthm}
\usepackage{mathtools}
\usepackage{color}
\usepackage{braket}
\usepackage{bm}
\usepackage{nicefrac}
\usepackage{comment}
\usepackage{algorithm, algpseudocode}
\floatname{algorithm}{Protocol}

\DeclareMathOperator{\Tr}{Tr}
\allowdisplaybreaks[4]

\theoremstyle{definition}
\newtheorem{theorem}{Theorem}
\newtheorem{lemma}[theorem]{Lemma}
\newtheorem{corollary}[theorem]{Corollary}
\newtheorem{proposition}[theorem]{Proposition}
\newtheorem{definition}[theorem]{Definition}

\makeatletter 
\renewcommand\onecolumngrid{
\do@columngrid{one}{\@ne}
\def\set@footnotewidth{\onecolumngrid}
\def\footnoterule{\kern-6pt\hrule width 1.5in\kern6pt}
}
\renewcommand\twocolumngrid{
        \def\footnoterule{
        \dimen@\skip\footins\divide\dimen@\thr@@
        \kern-\dimen@\hrule width.5in\kern\dimen@}
        \do@columngrid{mlt}{\tw@}
}
\makeatother

\newcommand{\abs}[1]{\left\lvert{#1}\right\rvert}

 \newcommand{\quoted}[1]{``#1''}

\newcommand{\initial}{{\rm ini}}
\newcommand{\final}{{\rm fin}}
\newcommand{\patch}{\mathrm{Sq}}

\begin{document}

\title{Continuous-Variable Fault-Tolerant Quantum Computation under General Noise}
\author{Takaya Matsuura}
\email{takaya.matsuura@riken.jp}
\affiliation{Centre for Quantum Computation \& Communication Technology, School of Science, RMIT University, Melbourne VIC 3000, Australia} 
\affiliation{RIKEN Center for Quantum Computing (RQC), Hirosawa 2-1, Wako, Saitama 351-0198, Japan}
\author{Nicolas C. Menicucci}
\affiliation{Centre for Quantum Computation \& Communication Technology, School of Science, RMIT University, Melbourne VIC 3000, Australia} 
\author{Hayata Yamasaki}
\email{hayata.yamasaki@gmail.com}
\affiliation{Department of Physics, Graduate School of Science,
The University of Tokyo, 7-3-1 Hongo, Bunkyo-ku, Tokyo 113-0033, Japan}

\begin{abstract}
    The quantum error-correcting code in the continuous-variable (CV) system attracts much attention due to its flexibility and high resistance against specific noise.  However, the theory of fault tolerance in CV systems is premature and lacks the general strategy to translate the noise in CV systems into the noise in logical qubits, leading to severe restrictions on the correctable noise models.  In this paper, we show that the Markovian-type noise in CV systems is translated into the Markovian-type noise in the logical qubits through the Gottesman-Kitaev-Preskill code with an explicit bound on the noise strength.  Combined with the established threshold theorem of the concatenated code against Markovian-type noise, we show that CV quantum computation has a fault-tolerant threshold against general Markovian-type noise, closing the existing crucial gap in CV quantum computation.  We also give a new insight into the fact that careful management of the energy of the state is required to achieve fault tolerance in the CV system.
\end{abstract}

\maketitle
\section*{Introduction}
Continuous-variable (CV) quantum optical systems, which encode quantum information to electromagnetic field quadratures, have distinct advantages for implementing quantum computation due to their affinity to quantum communication and the resulting scalability.
There are established technologies for measuring the quadratures of an optical mode in the field of optical telecommunication.  Furthermore, entangling operations are deterministic in the CV method, allowing the generation of a huge-scale entangled state even with current experimental technologies~\cite{Asavanant2019, Larsen2019}.

Fault tolerance is indispensable for the computation to be reliable, and the quantum error correction (QEC)~\cite{Shor1995} is necessary to achieve fault tolerance.
Various quantum error-correcting codes have been proposed for CV quantum computation~\cite{Chuang1997, Lloyd1998, Braunstein1998, Cochrane1999, Knill2001, Gottesman2001, Niset2008, Ralph2005, Bergmann2016, Michael2016, Grimsmo2020} (see also Ref.~\cite{Albert2018} for a review and comparison).  Among others, the Gottesman-Kitaev-Preskill (GKP) code~\cite{Gottesman2001} has the advantage of easier implementation of the universal gate set and computational-basis measurement~\cite{Gottesman2001} as well as its error correction capability \cite{Albert2018}.  In fact, only in the preparation of the GKP state does one need non-Gaussian optical operations, which are difficult to perform in experiments \cite{Baragiola2019, Yamasaki2020}.  The feasible generation of the (approximate) GKP state has been theoretically proposed in quantum optical system~\cite{Eaton2019, Tzitrin2019, Takase2022, Takase2024}, and a primitive GKP state has been experimentally demonstrated very recently \cite{Konno2024}.

Despite these experimental progresses, however, the theory of fault tolerance in CV systems has not seen full maturity yet. 
In the multi-qubit system, how to achieve fault tolerance is well known and established~\cite{Knill1996, Shor1997, Kitaev1997, Aharonov1997, Knill1998, Knill1998Science, Aharonov1999, Gottesman2009, Aliferis2011, Aliferis2013,yamasaki2022timeefficient}.
On the contrary, fault tolerance of CV quantum computation has been shown only against specific noise models such as Gaussian random displacement noise~\cite{Menicucci2014}.
Many studies~\cite{Fukui2018, Vuillot2019, Noh2019, Larsen2021, Bourassa2021, Tzitrin2021} have been claiming the existence of a fault-tolerant threshold of CV quantum computation with the GKP code following up the first study~\cite{Menicucci2014}, but all of these analyses are against very restrictive noise models such as a Gaussian-random displacement~\cite{Weedbrook2012} and a Gaussian approximation of the GKP code~\cite{Matsuura2020}.
Yet in experiments, indeed there are non-Gaussian-type errors such as the random phase rotation and experimental approximations of the GKP codeword proposed in Refs.~\cite{Eaton2019, Tzitrin2019, Takase2022, Takase2024} and demonstrated in Ref.~\cite{Konno2024}.  
These non-Gaussian errors may not be corrected perfectly by the error correction and may pile up in the existing fault-tolerance analyses, leading to a breakdown in fault tolerance.
Thus, the full fault-tolerance theory with the GKP code against general noise is indispensable for constructing a fault-tolerant optical CV quantum computer and guiding experimental efforts toward it.

In this paper, we close this crucial gap by showing the threshold theorem for optical quantum computation, regardless of the detail of the noise model.  At first glance, most of the matured techniques developed in the fault-tolerance theory of qubit quantum computation \cite{Knill1996, Shor1997, Kitaev1997, Aharonov1997, Knill1998, Knill1998Science, Aharonov1999, Gottesman2009, Aliferis2011, Aliferis2013,yamasaki2022timeefficient} are not carried over to the GKP code since
the ideal GKP codeword is nonnormalizable and thus not an element of the CV Hilbert space.  Approximate codewords are valid quantum states, but their mutual non-orthogonality makes things more complicated.
We show that by using the flexible framework to prove fault tolerance with concatenated codes~\cite{Aliferis2006, Gottesman2009, Aliferis2011, Aliferis2013,yamasaki2022timeefficient}, we can define a fault tolerance condition for the GKP code without relying on the unphysical ideal GKP codeword.
The remaining question is whether (non-Gaussian) noise models on the physical CV system are translated to correctable qubit-level noise models or not via concatenated code.  We show that under mild assumptions on the noise models in the physical CV system, the translated qubit-level noise model is actually a correctable one.  Our assumed noise model covers experimentally relevant ones such as non-Gaussian approximate GKP states, optical loss, and finite resolution of the homodyne detectors.
In this way, the threshold theorem for CV quantum computation is proved.

This paper thus provides a complete fault-tolerant digitization procedure for a quantum continuous variable, closing the aforementioned gap between the existing fault-tolerance theory and the current as well as the future experiment of CV quantum computing.  The obtained threshold theorem will guide experimental efforts on how to improve the systems to meet the fault-tolerance criteria.
This opens up the possibilities of CV fault-tolerant quantum information processing in a noisy real-world environment as well as pushes up our understanding of the CV quantum system.

\section*{Results}

\paragraph*{
Setting
}

A fault-tolerant protocol for quantum computation aims to approximate the output probability distribution of the ideal quantum computer with a quantum computer consisting of noisy devices in the real world.  The approximation error should be smaller than an arbitrary parameter $\epsilon$ in the total variation distance.
Such a fault-tolerant protocol is obtained in multi-qubit systems using concatenated codes~\cite{Gottesman2009, Aliferis2006, Aliferis2011, Aliferis2013,yamasaki2022timeefficient}, but not in the CV systems except for a very restrictive case~\cite{Menicucci2014} as discussed in the introduction. 

In this paper, we prove the threshold theorem for CV quantum computation by regarding the CV system as the physical (level-0) layer of a concatenated code and the qubit defined through the GKP code as the level-1 layer of the concatenation.  
We use the framework of a fault-tolerant protocol for the concatenated codes~\cite{Gottesman2009, Aliferis2006, Aliferis2011, Aliferis2013,yamasaki2022timeefficient} in which each state preparation, gate operation, and measurement in the level-$k$ layer of concatenation are replaced with fault-tolerant \quoted{gadgets} consisting of operations for level-$(k-1)$ layer code, which are then sandwiched by error-correction (EC) gadgets of level-$(k-1)$ layer code.  In the same way, we consider replacing each state preparation, gate operation, and measurement in the level-1 layer of a qubit-concatenated code with the GKP state-preparation gadget, the GKP gate gadget, and the GKP measurement gadget, which is then sandwiched by the GKP EC gadgets (see Fig.~\ref{fig:circuit_encoding}). 
In this way, the errors in the CV system are transformed into qubit-level errors by the GKP code.
Now, the problem is \quoted{What kind of noise model does the qubit in the level-1 layer undergo under natural assumptions on noise in the CV system and is it correctable?}

Reference~\cite{Menicucci2014} approached this question by assuming that the noise is Gaussian-random displacements, which we can handle relatively easily.  Then, that work shows that the qubit-level noise will be a stochastic Pauli noise for this particular case.  Note that even the Gaussian-approximate GKP codeword is not equivalent to the ideal GKP codeword subject to the Gaussian-random displacement noise.
Extending this result immediately faces obstacles, which we describe below.

\begin{figure}[t]
    \centering
    \includegraphics[width=0.95\linewidth]{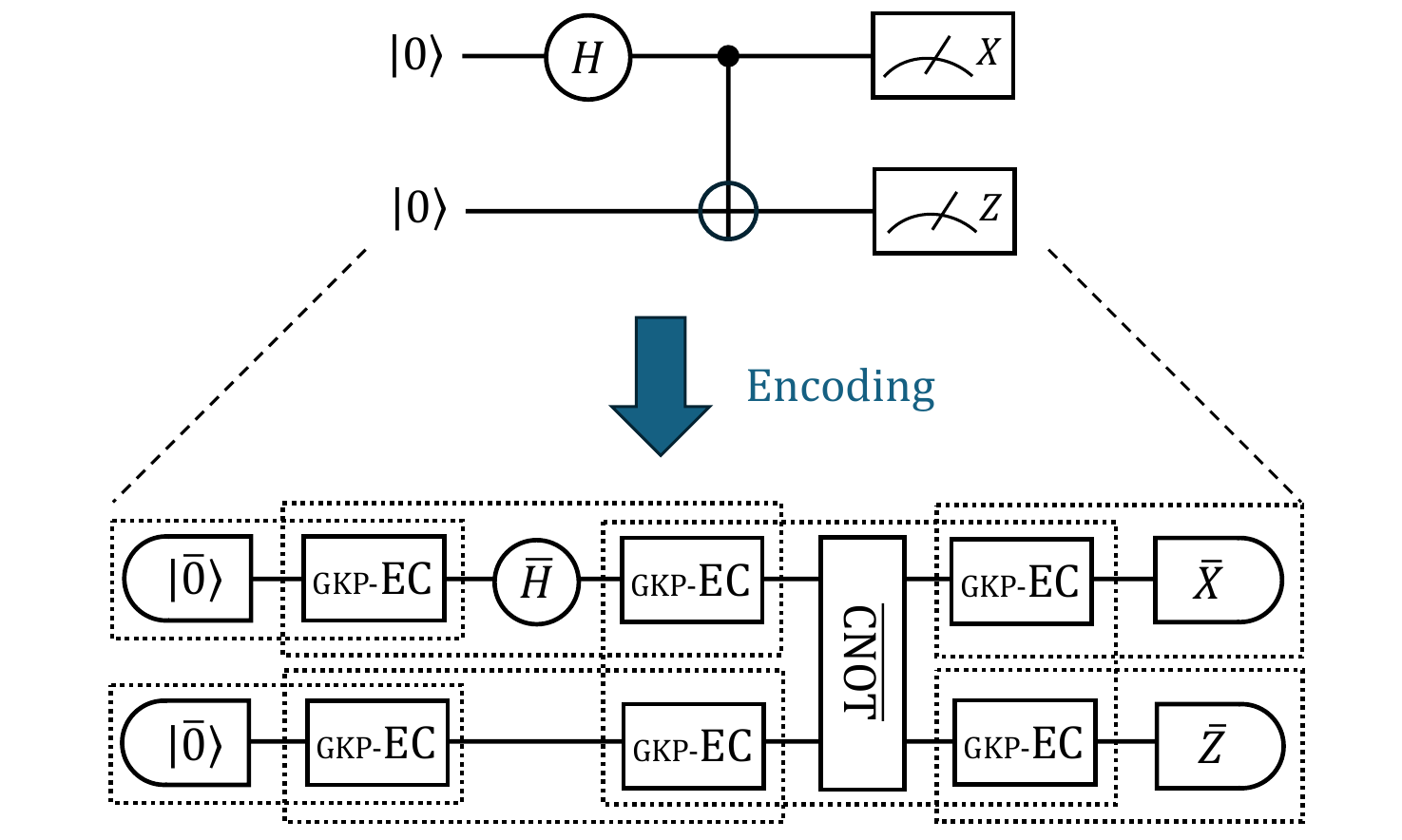}
    \caption{A qubit quantum circuit replaced with the (fault-tolerant) GKP gadgets.  In our setup, the qubit circuit should be the level-1 layer of a concatenated code.  Each dotted-lined box shows an extended rectangle (ExRec) \cite{Gottesman2009}, which is used in the level reduction theorem.}
    \label{fig:circuit_encoding}
\end{figure}

The first obstacle comes from the GKP code itself.
The GKP code we treat here is a stabilizer code with the $2\sqrt{\pi}$-shift in position and momentum quadratures in the phase space being the stabilizer generators \cite{Gottesman2001}.  For the GKP codeword, which is invariant under these shifts, the $\sqrt{\pi}$-shift in position and momentum quadratures act as logical Pauli-X and Z operators.
Despite the conceptual clarity, the GKP code has an intrinsic problem: the ideal GKP codeword is not normalizable and thus is not a physically valid state.  
If one tries to define the ideal GKP state as a limit of the sequence of its approximation~\cite{Gottesman2001, Matsuura2020}, then the energy (i.e., the average photon number) diverges in the limit.
In fact, no Hilbert subspace is invariant under $2\sqrt{\pi}$-shift in position and momentum.  This appears to be a big obstacle since we no longer have the \quoted{code space}, which is always used in the analysis of an error-correcting code and in the theory of fault tolerance.  The physically realizable state is only \quoted{approximately $2\sqrt{\pi}$-shift invariant}.  This is the reason why the GKP code has been regarded as an approximate error-correcting code.

Another obstacle is the lack of distance measures between noise quantum channels that can deal with physically relevant situations while maintaining the property that is necessary for fault-tolerance proof.  
In the finite-dimensional quantum system, the diamond norm is conventionally used as a distance measure, since it appears to be a relevant distance measure for quantum mechanics and satisfies the properties required for fault tolerance proof \cite{Kitaev1997Quantum, Aharonov1999}.
In the infinite-dimensional quantum system, however, this measure is too stringent: many physically relevant families of noise models have singular behavior with this distance measure \cite{Winter2017, Shirokov2018}. For example, the infinitesimal phase rotation is not just far from the identity map but maximally distant from it.  
The reason for this counterintuitive behavior comes from the fact that the quantum states with arbitrarily high susceptibility to the phase rotation have arbitrarily high energy (in this case, photon number).  Since what we can generate in the lab does not have unbounded energy and thus such a state cannot be generated in practice (even with perfect experimental equipment), we need to use an alternative measure to reflect this experimental constraint.
In the following, we explain how we overcome these obstacles and establish a fault tolerance proof for CV quantum computation.

\begin{figure}
    \centering
    \includegraphics[width=0.95\linewidth]{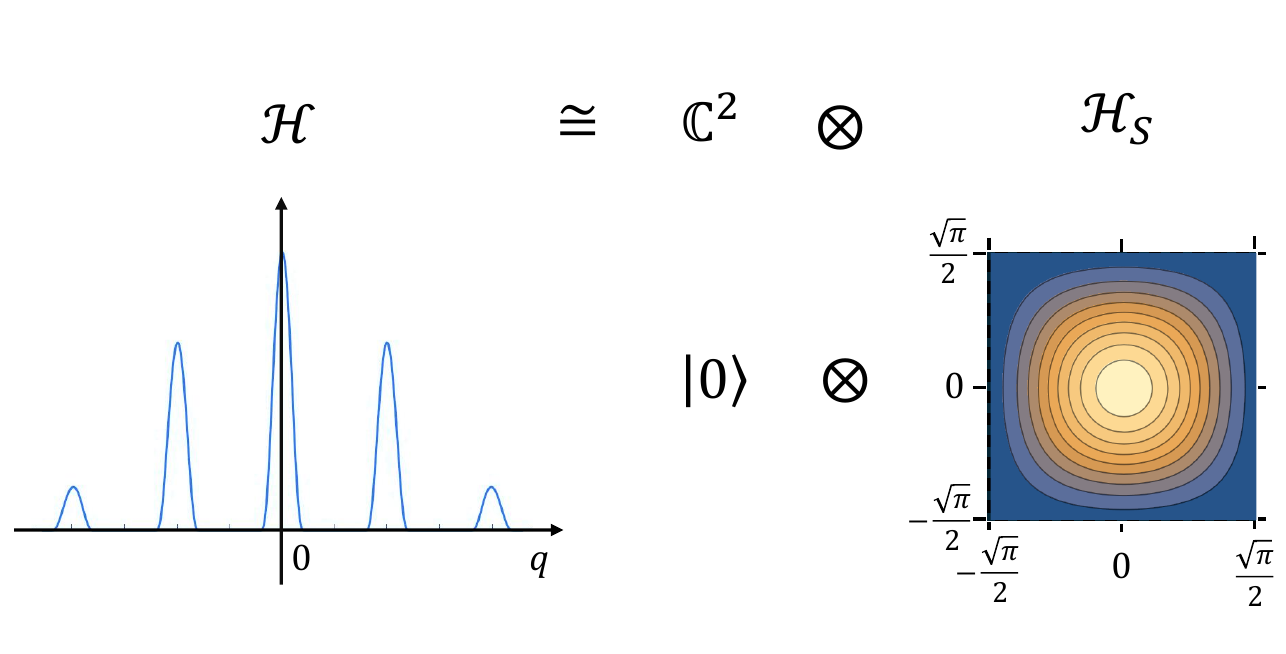}
    \caption{The pictorial representation of the subsystem decomposition studied in \cite{Pantaleoni2020,Pantaleoni2021Hidden,Pantaleoni2021Subsystem,Pantaleoni2023,Mackenzie2022}.  The above is a special case in which the state on the left-hand side can be written as a single tensor-product vector as on the right-hand side.  A general pure state can be written as a linear combination of the tensor-product vectors.}
    \label{fig:ssd}
\end{figure}

\paragraph*{
CV fault-tolerance condition from stabilizer subsystem decomposition
}

The first obstacle we face is the unphysicality of the ideal GKP codeword and the resulting difficulty in defining a fault-tolerance condition.  To circumvent this, we utilize a subsystem decomposition for the GKP code, first proposed in Ref.~\cite{Pantaleoni2020} and elaborated in Refs.~\cite{Pantaleoni2021Hidden, Pantaleoni2021Subsystem, Pantaleoni2023}. 
We will use the stabilizer subsystem decomposition~\cite{Mackenzie2022}, which remedied an undesirable asymmetry between the two quadratures that is present in the initial proposal \cite{Pantaleoni2020}. 
It splits the Hilbert space ${\cal H}$ of the quantum harmonic oscillator into a tensor product of the Hilbert space $\mathbb{C}^2$ of the logical qubit and an infinite-dimensional Hilbert space ${\cal H}_S$ representing the syndromes of stabilizer generators of the GKP code, i.e., ${\cal H}\cong\mathbb{C}^2\otimes {\cal H}_S$ (see Fig.~\ref{fig:ssd}).
In this decomposition, ${\cal H}_S$ is defined as the Hilbert space of square-integrable functions over the Cartesian square of a $\sqrt{\pi}$-sized interval~\cite{Mackenzie2022}.
Conceptually, this decomposition provides a CV generalization of the fact that a stabilizer code of $n$ physical qubits ${\cal H}=(\mathbb{C}^2)^{\otimes n}$ with a single logical qubit is decomposed into a tensor product of the Hilbert space $\mathbb{C}^2$ of the logical qubit and the Hilbert space ${\cal H}_S=(\mathbb{C}^2)^{\otimes n-1}$ of the $(n-1)$ syndrome qubits specified by the $(n-1)$ stabilizer generators, which can be represented as ${\cal H}=\mathbb{C}^2\otimes {\cal H}_S$ in the same way~\cite{Gottesman2009}.

With this decomposition, we introduce an equivalence class of errors with respect to which we define a fault-tolerance condition as in Refs.~\cite{Gottesman2009, yamasaki2022timeefficient} for the qubit case.  For this purpose, we define the stabilizer-subsystem (SSS) $r$-filter and the ideal GKP decoder, which is analogous to the $r$-filter and the ideal decoder defined for the concatenated code in Ref.~\cite{Gottesman2009}.
For the decomposition ${\cal H}=\mathbb{C}^2\otimes {\cal H}_S$, the SSS $r$-filter gives us a tool to discuss how much the wave function of a physical state of the GKP code deviates from the origin of the Cartesian square in the definition of ${\cal H}_S$ (representing no error; see Fig.~\ref{fig:ssd}) and the ideal GKP decoder tells us the state of the logical qubit $\mathbb{C}^2$ at a given location of the circuit.
The SSS $r$-filter depicted in Fig.~\ref{fig:demo_SSS_r-filter} is defined as a projection operator acting as identity on the logical qubit while acting as a projection onto the sub-region $[-r,r)\times[-r,r)$ of the Cartesian square in the definition of $\mathcal{H}_S$ with $0<r<\sqrt{\pi}/2$.
The ideal GKP decoder, on the other hand, is a completely positive trace-preserving (CPTP) map that traces out $\mathcal{H}_S$ to transform the physical CV state of $\mathcal{H}$ into the state of the logical qubit $\mathbb{C}^2$.

\begin{figure}[t]
    \centering
    \includegraphics[width=0.98\linewidth]{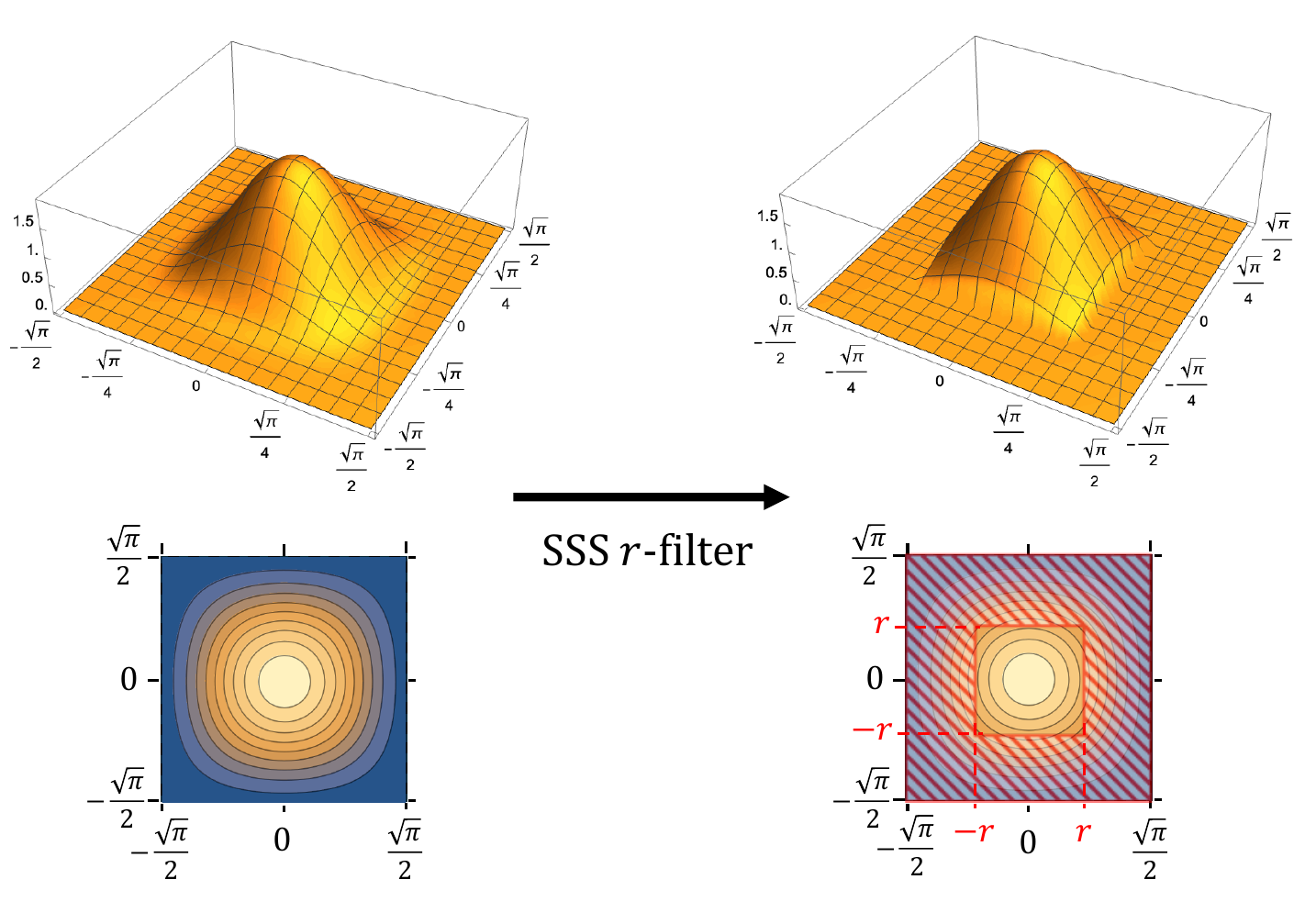}
    \caption{The pictorial representation of ${\cal H}_{S}$ in the 3D plot (top) and the contour plot (bottom) with the action of the SSS $r$-filter on it. By the action of the SSS $r$-filter, the red-shaded region in the right-hand is set to zero.}
    \label{fig:demo_SSS_r-filter}
\end{figure}

Now, we can define an equivalence class of \quoted{noisy} or \quoted{approximate} GKP states, which we call the class of $r$-parameterized GKP states. 
An $r$-parameterized GKP state for a qubit state $\ket{\psi}$ is the product state of $\ket{\psi}$ in the logical qubit and a state that is supported only on the sub-region $[-r,r)\times[-r,r)$ in the Cartesian square in the definition of $\mathcal{H}_S$  Thus, the SSS $r$-filter acts trivially on it, and the ideal GKP decoder decodes it into $\ket{\psi}$.  Unlike the unphysical ideal GKP codeword, an $r$-parameterized GKP state is well-defined for any $r>0$.
Then, a fault-tolerance condition for this CV quantum computation is given such that the state preparation gadget prepares an $r$-parameterized GKP state with sufficiently small $r$, gate and measurement gadgets do not enlarge $r$ up to small constants, and the error correction gadget refreshes the state to an $r$-parameterized GKP state while keeping the logical information.

Noticing that an $r$-parameterized GKP state can be written as a superposition of up-to-$r$ displaced ideal GKP states~\cite{Pantaleoni2020, Mackenzie2022}, it is necessary that the gate and measurement operations do not amplify displacement errors to uncorrectable ones.  This is analogous to the qubit fault-tolerance condition where it is necessary that the gate or measurement operations do not spread local errors.
We show that the conventionally used gate and measurement operations for the GKP code~\cite{Gottesman2001} satisfy the fault-tolerance conditions defined above (see Methods and Sec.~\ref{sec:FT_conditions} in Supplementary Information). 
It can also be shown that the gates followed by a small displacement error or the measurements after a small displacement error satisfy the fault-tolerance conditions.  We can thus obtain a way to describe the fault-tolerance conditions for the CV state preparation, gates, and measurements without referring to the detail of the noise and the shape of the wave function of the approximate GKP codeword.  This is in sharp contrast to the previous work~\cite{Menicucci2014} where the noise models and the shapes of the wave function were explicitly assumed.

Now, we observe that the error correction operation during the computation has to get the state back to an $r$-parameterized GKP state with $r$ sufficiently smaller than $\sqrt{\pi}/2$, i.e., the threshold of correctable displacement parameters for the GKP code.  Otherwise, noise during the computation accumulates and lets the support of the wave function in the Cartesian square of $\mathcal{H}_S$ go beyond $\sqrt{\pi}/2$, leading to a logical error.  To get the state back into an $r$-parameterized GKP state, we need GKP error correction. 
However, we postpone the discussion of GKP error correction since the second obstacle---the need to take the energy constraint into account for obtaining an appropriate distance measure---also affects the requirements of GKP error correction.

\paragraph*{
Energy-constraint condition
}

As explained so far, we have succeeded in abstracting the GKP code in a way that satisfies the requirements for fault tolerance.
Thus, we can prove that the CV fault-tolerant circuit given as in Fig.~\ref{fig:circuit_encoding} perfectly reproduces the outcomes of the original qubit circuit that we aim at simulating by the CV fault-tolerant circuit as long as all the CV circuit components perfectly satisfy the fault-tolerance conditions.  However, this idealized situation will never be realized in experiments.  We need tools to treat broader situations; the circuit components are almost perfect but may have a general class of physical errors.  This is where the distance measure for channels in CV systems is necessary.  

As mentioned earlier, the diamond-norm distance, which is conventionally used in the fault-tolerance proof for the qubit quantum computation~\cite{Kitaev1997Quantum, Aharonov1999}, is too stringent for CV systems. 
In the literature \cite{Holevo2003, Shirokov2008, Winter2016, Winter2017, Shirokov2018, Shirokov2018adaptation, Shirokov2019}, the energy-constrained version of the diamond norm has been considered, which induces a weaker and physically relevant topology of the set of quantum channels.
We also make use of this energy-constrained diamond norm.  For this, the mode-wise energy (i.e., the average photon number) of the state at every time step during the computation needs to be bounded.  Note that this is another reason why we need to avoid the reference of the \quoted{ideal} GKP codeword in the analysis, which has unbounded energy and may have an arbitrarily high susceptibility to phase rotation noise, which is numerically observed in Ref.~\cite{Mackenzie2022}.
Since CV gate operations in general increase the energy, the state during the computation should continually be refreshed to a state with a constant energy upper bound.  

\begin{figure}[t]
    \centering
    \includegraphics[width=0.98\linewidth]{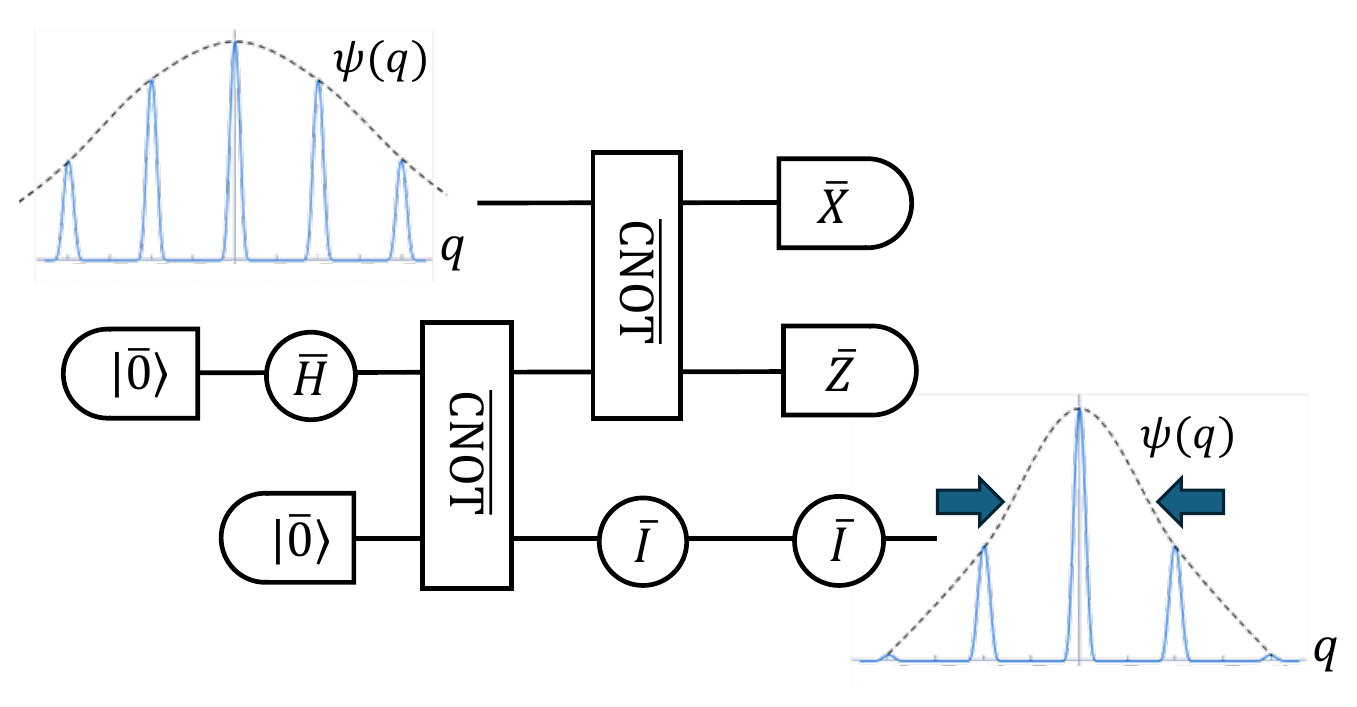}
    \caption{Teleportation-based (Knill-type) GKP error correction \cite{Knill2005} with the input and the output position wave function square.
    In the figure, the feed-forward (GKP) Pauli correction operations are omitted; they can be corrected by the Pauli frame update or by modifying the successive gates and measurements. 
    Since the quantum state in the input is teleported to the ancillary prepared mode, the energy of the state is reset, which is depicted by the change of the variance of the wave function square.  Note that the energy of the state depends also on the variance of the momentum wave function square.
    }
    \label{fig:knill_EC_energy_reset}
\end{figure}

To achieve this, we exploit a GKP EC gadget that satisfies both requirements of performing GKP error correction and resetting the energy.
This is accomplished by the teleportation-based (i.e., Knill-type) error correction~\cite{Knill2005}, as shown in Fig.~\ref{fig:knill_EC_energy_reset}, where the logical information during the computation is teleported into a newly prepared auxiliary GKP codeword with constant energy.
It can be shown that if all the circuit elements in Fig.~\ref{fig:knill_EC_energy_reset} satisfy the fault-tolerance conditions in terms of SSS $r$-filter and the ideal GKP decoder, then this EC gadget also satisfies the fault-tolerance condition as a whole (see Methods and Sec.~\ref{sec:FT_conditions} in Supplementary Information).
Furthermore, as in Fig.~\ref{fig:circuit_encoding}, states during the computation are continually reset to a newly prepared GKP state while their logical information is preserved; this ensures that each prepared state during the computation experiences only a constant number of gates before being measured in a GKP EC gadget.
Then, we can prove that the energy of a state during the computation has an upper bound (see Sec.~\ref{sec:energy_constraint_conditions} in Supplementary Information). 

This enables us to use the energy-constrained diamond-norm distance to assess the effect of noise.
Unlike Knill-type EC, the conventional Steane-type GKP EC \cite{Gottesman2001} may not reset the energy of the error-corrected state.  Then, the energy of the state during the computation may continue increasing and be more fragile against noise, which breaks down fault tolerance.  Note that there may be ways to overcome this problem for Steane-type EC by a more complicated post-processing strategy as implied in Refs.~\cite{Tzitrin2019, Wan2020, Terhal2020}, but it is difficult to evaluate the upper bound on the energy of the state during the computation.

\paragraph*{
CV FTQC under general noise
}

With the energy-constrained diamond norm, we introduce our noise model for CV systems, which covers experimentally relevant noise.  Our noise model is independent in the sense that each mode experiences noise that is independent of that suffered by any other mode in the CV quantum computer, and it is Markovian in the sense that it has no correlation in time.  Thus, we name this noise model an $(E,r,\epsilon)$-independent Markovian noise model, where $E$ is the energy constraint we take for the energy-constrained diamond norm, $r$ is a maximum amount of additional displacements to the ideal quantum operation, and the noise strength $\epsilon$ is the distance from this up-to-$r$-displaced ideal quantum operation measured in the $E$-constrained diamond norm.  For state preparation, for example, $\epsilon$ is equal to the trace distance from an $r$-parameterized GKP state.
This noise model not only covers non-Gaussian approximate GKP state preparation but also covers mode-wise optical loss, random phase rotation, and finite resolution of a homodyne detector.
(See Methods for the formal definition.)  This is in stark contrast to the previous work~\cite{Menicucci2014} where the noise model is restricted to Gaussian-random displacement.

Under this noise model, we aim to prove the threshold theorem for CV fault-tolerant quantum circuits.
Our strategy to prove fault tolerance is based on the level reduction \cite{Gottesman2009, Aliferis2011, Aliferis2013,yamasaki2022timeefficient}.  For this, we split the circuit into extended rectangles (ExRecs) in which each operation we want to implement is replaced with a fault-tolerant gadget and the error corrections are inserted between these fault-tolerant gadgets.  The ExRec is surrounded by the dotted line in Fig.~\ref{fig:circuit_encoding}.
Each ExRec determines the behavior of the logical qubit-level circuit; i.e., if the ExRec is \quoted{bad}, the corresponding logical operation is regarded as erroneous.  The detail of the level reduction is analyzed in Methods and Sec.~\ref{sec:threshold_theorem} in Supplementary Information.  Combined with the threshold theorem for qubit quantum computation in Ref.~\cite{Aliferis2011, Aliferis2013}, we reach the following theorem.  (See Methods for the sketch of the proof.)

\begin{theorem}[(Informal statement of Thm.~\ref{theo:level_reduction} and Corol.~\ref{cor:threshold_theorem}) Threshold theorem for CV quantum computation]\label{theo:threshold_main_text}
    We can construct a CV protocol using the GKP code in a way that for any target value $p>0$, there exist nonzero thresholds $\epsilon_{\rm th}$ and $r_{\rm th}$ and a finite energy bound $E_{\rm th}$ satisfying the following: even if each location of the physical CV circuit suffers from a general class of (not necessarily Gaussian) independent Markovian noise that deviates from its ideal action by (1)~a distance parameter $\epsilon<\epsilon_{\rm th}$, in terms of the energy-constraint diamond norm with energy $E_{\rm th}$, and (2)~a displacement parameter $r<r_{\rm th}$, as in the noise model described above, it is guaranteed that the logical circuit of the GKP qubits undergoes local Markovian noise with noise strength below $p$. That is, by concatenating this CV protocol with a fault-tolerant protocol for qubits with a threshold $p$ under local Markovian noise, we can achieve fault-tolerant quantum computation using CV systems at the physical level.
\end{theorem}

We remark that in the above theorem, the noise model that the logical circuit of the GKP qubits undergoes is no longer independent due to the correlation made by the error correction procedures, but the occurrence of correlated noise is as rare as that for the independent noise model.  This noise model, called local Markovian noise, is studied in detail for qubit systems in Refs.~\cite{Aliferis2011, Aliferis2013} in which the threshold theorem is established.  Thus, we can use this result to obtain the final statement of Theorem~\ref{theo:threshold_main_text}.

\section*{Discussion}
In this work, we rigorously showed the existence of a fault-tolerance threshold for quantum computation with the GKP quantum error correcting code under the general class of independent Markovian noise models, generalizing the result from Ref.~\cite{Menicucci2014}, which considered only Gaussian random displacement noise.
Our noise model covers experimentally relevant noise, such as non-Gaussian approximation of the GKP states, optical loss, and finite resolution of homodyne detection.
This can be proven using the stabilizer subsystem decomposition of the GKP code developed in Ref.~\cite{Mackenzie2022}.
The crucial difference from fault tolerance for qubit quantum computation is that the fault tolerance condition for the CV case requires (1)~that displacements not be amplified by gate operations and (2)~that the energy of a state during the computation remains bounded. In contrast with the qubit case, without the energy condition (2), a CV state may become more and more susceptible to small noise and may eventually be uncorrectable. To avoid the accumulation of energy, as well as to correct CV-level errors, we use Knill-type error correction \cite{Knill2005} to ensure these conditions are satisfied. In practice, Steane-type GKP error correction~\cite{Gottesman2001} will likely also work as long as some efforts are made to ensure bounded energy in the state~\cite{Tzitrin2019}. Performing this analysis is left to future work.

The implications obtained by the fault-tolerance framework and the threshold theorem developed here for CV quantum computation are many.  
First, the threshold theorem given here makes it clear what should be experimentally achieved in CV systems to build a CV quantum computer.  As mentioned above, the noise model we treat in this paper is broad and contains a broad class of experimentally relevant noise.  By estimating the noise strength $\epsilon$ of the actual experimental noise, one can judge that the noise level is low enough to implement a CV fault-tolerant quantum computer.
Second, there may be a crucial difference between fault-tolerant quantum information processing with DV versus CV systems.  In the CV case, the energy of the state during the information processing should be carefully managed; otherwise, the state may be more and more fragile against noise.  For the same reason, Knill-type EC and Steane-type EC may not be equivalent anymore in the CV system without careful adjustments. The former is ensured to reset the energy of the data mode, while the latter is not.
These observations open up avenues for future studies on the possibilities and limitations of CV quantum computation and information processing.

\begin{acknowledgments}
  T.~M.~acknowledges helpful discussions with Yui Kuramochi.  This work was supported by JST, CREST Gant Number JPMJCR23I3, Japan.  T.~M.~was supported by JSPS Overseas Research Fellowships.
  H.~Y.~was supported by JST PRESTO Grant Number JPMJPR201A, JPMJPR23FC, JSPS KAKENHI Grant Number JP23K19970, and MEXT Quantum Leap Flagship Program (MEXT QLEAP) JPMXS0118069605, JPMXS0120351339\@.
  This work was supported by the Australian Research Council (ARC) Centre of Excellence for Quantum Computation and Communication Technology (Project No.~CE170100012). N.~C.~M.~was supported by an ARC Future Fellowship (Project No.~FT230100571).
\end{acknowledgments}

\section*{Author contributions}

T.~M.~contributed to the conception, analysis, and interpretation of the work, and mainly wrote the manuscript.  H.~Y.~contributed to the conception, analysis, and interpretation of the work, and supervised the manuscript writing.  N.~C.~M.~contributed to the interpretation of the work and supervised the manuscript writing.

\section*{Competing interests}

The authors declare no competing interests.

\section*{Additional information}

Supplementary Information is available for this paper.
Correspondence and requests for materials should be addressed to Takaya Matsuura.

\section*{Methods}

\paragraph*{
The GKP code and the stabilizer subsystem decomposition
}

A harmonic oscillator is a prominent example of a continuous-variable (CV) system.  A single-mode harmonic oscillator can be characterized by a pair of noncommutative quadrature operators $\hat{q}$ and $\hat{p}$ that satisfy $[\hat{q},\hat{p}]=i$, acting on (a dense linear subspace of) a separable Hilbert space~${\cal H}$~\cite{Holevo2011}.  The energy level of a quantum harmonic oscillator is quantized and labeled by a non-negative number called the photon number.  The photon-number operator $\hat{n}$ is defined as $\hat{n}\coloneqq (\hat{q}^2+\hat{p}^2-1)/2$.

We consider the GKP code for encoding a qubit into a harmonic oscillator.  
The ideal GKP-encoded state (the GKP codeword) $\ket{\overline{\psi}}$ is formally defined as an infinite superposition of position eigenstates, i.e.,~$\ket{\overline{\psi}}=\alpha\ket{\overline{0}}+\beta\ket{\overline{1}}$ with $\ket{\overline{j}}\propto \sum_{s\in\mathbb{Z}}\ket{\sqrt{\pi}(2s+j)}_q$ for $j=0,1$, where $\ket{t}_q$ satisfies $\hat{q}\ket{t}_q=t\ket{t}_q$ for the operator $\hat{q}$.
The GKP code can be regarded as a stabilizer code with the stabilizer generated by $\exp(2\sqrt{\pi}i\hat{q})$ and $\exp(2\sqrt{\pi}i\hat{p})$.
Logical Clifford unitary gates on this code can be realized by Gaussian unitaries, which are generated by second-order polynomials of $\hat{q}$ and $\hat{p}$.
The correspondence between the GKP logical operations and the CV physical operations is given by \cite{Gottesman2001}
\begin{align}
 &\text{Pauli-X:}&\overline{X} &\longrightarrow \exp(-\sqrt{\pi}i\hat{p}), \label{eq:pauli-X_physical}\\
 &\text{Pauli-Z:}&\overline{Z} &\longrightarrow \exp(\sqrt{\pi}i\hat{q}), \\
 &\text{Hadamard:}&\overline{H} &\longrightarrow \hat{F}\coloneqq 
 \exp(\pi i  \hat{n}/2), \\
 &\text{CNOT:}&\overline{\text{CNOT}} &\longrightarrow \exp(-i\hat{q}_1\hat{p}_2),\label{eq:CNOT_physical} \\
 &\text{Wait: }&\overline{I}&\longrightarrow \hat{I}.
\end{align}
Other logical gates can be implemented through gate teleportation as long as one can prepare GKP magic states~\cite{Zhou2000, Baragiola2019, Yamasaki2020}.  More precisely, the logical phase gate $\overline{S}$ can be implemented through the gate teleportation with the state $\ket{\overline{Y}}\coloneqq(\ket{\overline{0}}+i\ket{\overline{1}})/\sqrt{2}$, and the logical $T$ gate $\overline{T}$ can be with the state $\ket{\frac{\pi}{8}}\coloneqq (\ket{\overline{0}}+e^{\pi i/4}\ket{\overline{1}})/\sqrt{2}$.  Note that the CV shearing operation $\exp(i\hat{q}^2/2)$ is conventionally used to implement the logical phase gate $\overline{S}$ in the GKP code \cite{Gottesman2001}, but here we do not use this for simplicity of our analysis.
The logical Pauli-X or -Z measurement in this code can be implemented by a homodyne detection in $\hat{p}$ or $\hat{q}$ quadrature, respectively, followed by the binning of the CV measurement outcome modulo $\sqrt{\pi}$.  If the outcome is binned to the even (respectively odd) multiple of $\sqrt{\pi}$, it is regarded as logical 0 (respectively 1).  This completes the description of how to perform a universal quantum computation with the GKP code.

The expression of $\ket{\overline{\psi}}$ above is only formal and does not represent a physically realizable state.  For the code to be physically realizable, we need to approximate the state in some way~\cite{Gottesman2001}, which complicates the theoretical treatment of the GKP code \cite{Matsuura2020} and makes it difficult to analyze fault tolerance.  It is also a problem of how or in what sense the state needs to be approximated, which we will resolve in this paper.
The way of approximating the GKP codewords has been discussed in the literature \cite{Gottesman2001, Menicucci2014, Matsuura2020}.  The recent advance in this direction is the technique called the subsystem decomposition \cite{Pantaleoni2020, Pantaleoni2021Hidden, Pantaleoni2021Subsystem, Pantaleoni2023, Mackenzie2022}, which decomposes the Hilbert space of the CV system into the tensor product of the logical qubit and an infinite-dimensional Hilbert space. 

The subsystem decomposition is similar to decomposing the physical Hilbert space of the stabilizer code into the tensor product of those representing the logical qubits and the syndrome qubits.
In Steane's 7-qubit code~\cite{Steane1996Multiple}, for example, the Hilbert space of the seven qubits can be represented as a tensor product of those of a logical qubit and six syndrome qubits on which the logical Pauli operators and the stabilizer operators, respectively, act nontrivially~\cite{Gottesman2009}.  A codeword can be represented, up to this isomorphism, as a tensor product of any state of the logical qubit and a particular basis state of the syndrome qubits, say $\ket{0}^{\otimes 6}$ (i.e., the origin of the name \quoted{syndrome qubits}, representing the case of no error).
A Pauli error on the physical qubits of the codeword flips some of the six $\ket{0}$s to $\ket{1}$s.  The error correction procedure for the 7-qubit code thus corresponds to detecting $\ket{1}$s in the syndrome qubits and correcting the overall state to $\ket{0}^{\otimes 6}$ of the syndrome qubits without changing the state of the logical qubit. 
One may also regard the logical qubit and the syndrome qubits as \quoted{subsystems} in terms of this tensor-product decomposition as in Refs.~\cite{Knill1997,Knill2000}, where the \quoted{subsystems} in this context do not necessarily mean directly accessible ones such as physical qubits but are specifically defined from this tensor-product decomposition. 
In this sense, we may collectively call the syndrome qubits a \textit{syndrome subsystem} for this tensor-product decomposition.
All stabilizer codes can, in principle, be decomposed in this way, and so can the GKP code, where the Hilbert space of the syndrome qubits in the case of qubit stabilizer codes is replaced with an infinite-dimensional space of the syndrome subsystem for the GKP code.
There are many different ways to define such a decomposition for the GKP code (called subsystem decomposition in the context of GKP codes~\cite{Pantaleoni2020,Pantaleoni2021Hidden, Pantaleoni2021Subsystem, Pantaleoni2023}), and each corresponds to a different decoding strategy.

Among several ways of the subsystem decomposition developed so far for the GKP code, we utilize the stabilizer subsystem decomposition developed in Ref.~\cite{Mackenzie2022} because it corresponds to the typical binning decoder for GKP error correction~\cite{Gottesman2001}.
In the stabilizer subsystem decomposition, the Hilbert space ${\cal H}$ of a quantum harmonic oscillator is decomposed as
\begin{align}
\label{eq:SSS_Hilbert_space}
{\cal H}\cong \mathbb{C}^2\otimes {\cal H}_{S}, 
\end{align}
where $\mathbb{C}^2$ denotes the logical qubit and ${\cal H}_{S}$ denotes the Hilbert space of square-integrable functions over the square area
\begin{align}
\label{eq:patch}
     \patch\coloneqq [-\sqrt{\pi}/2,\sqrt{\pi}/2)\times [-\sqrt{\pi}/2,\sqrt{\pi}/2),
\end{align}
as shown in Fig.~\ref{fig:ssd}. Thus, any state $\ket{\phi}\in{\cal H}$ can be represented as a linear combination of vectors of the form $\ket{\mu}\otimes \int_{(z_1,z_2)\in \patch}f(z_1,z_2)\ket{z_1,z_2}$, with $\ket{\mu}\in\mathbb{C}^2$ and $f\in {\cal H}_{S}$.  We call this Hilbert space $\mathbb{C}^2\otimes {\cal H}_{S}$ the \emph{stabilizer-subsystem (SSS) Hilbert space}, and we call the former subsystem the logical qubit and the latter subsystem the syndrome subsystem in terms of this decomposition.  (Note that the latter is originally called the stabilizer subsystem~\cite{Mackenzie2022}.)  The GKP stabilizer operator $\exp(2\sqrt{\pi}i\hat{q})$ and $\exp(-2\sqrt{\pi}i\hat{p})$ acting on ${\cal H}$ can be represented on the SSS Hilbert space as $\hat{I}\otimes\exp(2\sqrt{\pi}i\hat{z}_1)$ and $\hat{I}\otimes\exp(-2\sqrt{\pi}i\hat{z}_2)$, respectively, where $\hat{z}_i$ for $i=1,2$ is an operator satisfying $\hat{z}_i\ket{z_1,z_2}=z_i\ket{z_1,z_2}$.  The GKP logical Pauli-X operator $\overline{X}$ and the Z operator $\overline{Z}$ can be represented as $\hat{\sigma}_X\otimes \exp(-\sqrt{\pi}i\hat{z}_2)$ and $\hat{\sigma}_Z\otimes \exp(\sqrt{\pi}i\hat{z}_1)$, respectively, where $\hat{\sigma}_X$ and $\hat{\sigma}_Z$ denote the Pauli operators on a qubit.
The ideal GKP state has delta distribution in both $z_1$ and $z_2$ and is thus an ill-defined state.  However, it is clear from the above decomposition that the vector in the form $\ket{0}\otimes \int_{(z_1,z_2)\in \patch}f(z_1,z_2)\ket{z_1,z_2}$ has the same logical qubit state as the ideal GKP $\ket{\overline{0}}$ regardless of the form of the function $f$.

This decomposition can be derived from the fact that $\{\hat{q}\}_{\sqrt{\pi}}$ and $\{\hat{p}\}_{\sqrt{\pi}}$ commute with each other, while $\hat{q}$ and $\hat{p}$ do not \cite{Pantaleoni2023}, where $\{a\}_{\sqrt{\pi}}\coloneqq a - \lfloor a \rceil_{\sqrt{\pi}}$ and $\lfloor a \rceil_{\sqrt{\pi}}$ denote the closest integer multiple of $\sqrt{\pi}$ for the real number $a$.  Furthermore, each $\exp(\sqrt{\pi}i \lfloor \hat{q} \rceil_{\sqrt{\pi}})$ and $\exp(-\sqrt{\pi}i \lfloor \hat{p} \rceil_{\sqrt{\pi}})$ commute with both $\{\hat{q}\}_{\sqrt{\pi}}$ and $\{\hat{p}\}_{\sqrt{\pi}}$, while they two anti-commute with each other. In fact, $\exp(\sqrt{\pi}i \lfloor \hat{q} \rceil_{\sqrt{\pi}})$ and $\exp(-\sqrt{\pi}i \lfloor \hat{p} \rceil_{\sqrt{\pi}})$ acting on ${\cal H}$ correspond respectively to $\hat{\sigma}_Z\otimes\hat{I}$  and $\hat{\sigma}_X\otimes\hat{I}$ operators acting nontrivially on $\mathbb{C}^2$ of the SSS Hilbert space, while $\{\hat{q}\}_{\sqrt{\pi}}$ and $\{\hat{p}\}_{\sqrt{\pi}}$ acting on ${\cal H}$ correspond to $\hat{I}\otimes\hat{z}_1$ and $\hat{I}\otimes\hat{z}_2$ operators acting nontrivially on ${\cal H}_{S}$.
In Sec.~\ref{sec:stabilizer_ssd} in the Supplementary Information, we review the stabilizer subsystem decomposition in more detail.

\paragraph*{
Fault-tolerant quantum computation and how to achieve it
}
The goal of fault-tolerant quantum computation (FTQC) is to construct a circuit $C'$ that simulates the $W$-qubit $D$-depth original quantum circuit $C$ within any given target error $\epsilon>0$; i.e., the circuit $C'$ outputs a $W$-bit string sampled from a probability distribution close to the output probability distribution of the original circuit $C$ within error in the total variation distance at most $\epsilon$ \cite{yamasaki2022timeefficient}. 
Here, the depth means the number of time steps at which the preparation, gate, measurement, or wait operation is performed on each qubit.  We call a physical operation performed on a given physical system at a given time step a \emph{location}~$C_i$ in the quantum circuit $C$, with $i\in{\cal I}$, where the index set ${\cal I}$ is chronologically ordered for simplicity of presentation (among the locations at the same time step in $C$, the order is arbitrary).  Note that $\abs{{\cal I}}\leq WD$ holds by definition.

In qubit quantum computation, the achievability of fault tolerance has been shown~\cite{Knill1996, Shor1997, Kitaev1997, Aharonov1997, Knill1998, Knill1998Science, Aharonov1999, Gottesman2009, Aliferis2011, Aliferis2013,yamasaki2022timeefficient}.
This is made possible if the error on each circuit component is small enough and the physical circuit undergoes a reasonable noise model.
For example, let ${\cal O}_i$ be an ideal map we want to implement at location $C_i$ in the circuit $C$ and let $\tilde{\cal O}_i$ be its noisy version that we can implement in the experiment.  If the noise is independent (meaning the noise has no space-like correlation) and Markovian (meaning the noise has no time-like correlation), then the noisy output distribution $p^{\rm noisy}$ is given by 
\begin{align}
    p^{\rm noisy} &= \tilde{\cal O}_{|{\cal I}|} \circ \cdots \circ \tilde{\cal O}_{1}  \\
    &= ({\cal O}_{|{\cal I}|} + {\cal F}_{|{\cal I}|}) \circ \cdots \circ ({\cal O}_{1}+{\cal F}_{1}),
\end{align}
where the map ${\cal F}_i\coloneqq \tilde{\cal O}_i-{\cal O}_i$ is regarded as a fault.  (Here, we assume that the set $\{C_i\}_{i\in{\cal I}}$ is chronologically ordered.)  When the second expression above is expanded, each term except $O_{|{\cal I}|}\circ\cdots\circ O_{1}$ is called a \emph{fault path}.  The noise strength $\delta_i$ is then defined as the (diamond) norm of the fault ${\cal F}_i$ in the qubit case.  If each $\delta_i$ is bounded from above by a constant $\delta$, then the total variation distance between the noisy output distribution $p^{\rm noisy}$ and the ideal output distribution $p^{\rm ideal}$ is bounded from above by $\epsilon=(e-1)WD\delta$, where $e$ is Euler's number (see Sec.~\ref{sec:threshold_theorem} in Supplementary Information).  To make the distance arbitrarily small, we need to make $\delta$ arbitrarily small, and for this, the circuit component should be encoded by a quantum error-correcting code.
But again, the achievable noise strength is dominated by the noise strengths to realize this quantum error-correcting code.  As a result, one may need to use a code that can arbitrarily suppress the error by expanding the size of the code or by nesting the code over and over to achieve the target noise strength $\delta$.  The latter choice is called the concatenated code \cite{Gottesman2009, Aliferis2011, Aliferis2013,yamasaki2022timeefficient}, which we use in this paper.

When using a concatenated code, one needs to pay attention to how the noise model is translated from the physical level to the logical level. This is because the logical error may have correlations caused by the structure of the quantum error-correcting code and the fault-tolerant quantum circuit.
The translation of the noise properties from the physical level to the logical level is studied in detail in Refs.~\cite{Gottesman2009, Aliferis2011, Aliferis2013,yamasaki2022timeefficient}.
In the case of independent Markovian noise at the physical level, the translated noise at the logical level is no longer independent Markovian noise.  However, it is shown in Refs.~\cite{Aliferis2011, Aliferis2013} that, at the logical level, the sum of all the fault paths with faults in $R$ specific locations is bounded from above by $\epsilon_{\rm qubit}^{|R|}$ for a constant upper bound $\epsilon_{\rm qubit}$ on the noise strength at any location, where $|R|$ denotes the cardinality of the set $R$.  This noise model is called a \emph{local Markovian noise model}.  What is important is that if we further concatenate the qubit QEC code, then the local Markovian noise model in the lower-level qubit circuit is translated to local Markovian noise in the higher-level qubit circuit.  Thus, there exists a threshold noise strength $\epsilon^\star_{\rm qubit}$ below which we can arbitrarily suppress the error in the computational result.  This whole process of reinterpreting the fault-tolerant quantum circuit at a given level with noise following a certain model as that at a higher level with noise following a translated noise model is called \emph{level reduction}~\cite{Gottesman2009}.

To establish this level reduction for the qubit concatenated code, Refs.~\cite{Gottesman2009, Aliferis2011, Aliferis2013,yamasaki2022timeefficient} introduced the concepts of the $r$-filter and the ideal decoder.  The $r$-filter in these works is a projection operator onto a subspace ${\cal H}_r$ spanned by the vectors that are up to $r$ (Pauli) errors away from the codewords.  Thus, for an error-correcting code that can correct up to $t$ Pauli errors (i.e.,~with code distance $2t+1$), the $r$-filter defines an increasing sequence $({\cal H}_r)_{r=0,\ldots,t}$ of Hilbert subspaces with ${\cal H}_{r_1}\subseteq {\cal H}_{r_2}$ for $r_1\leq r_2$, where ${\cal H}_0$ is the code space and ${\cal H}_t$ is the total Hilbert space of physical qubits that make up the code.
The ideal decoder is a virtual quantum channel that performs a perfect decoding operation (including error correction) on the input state; that is, an error-correcting channel from physical qubits to logical qubits.
The $r$-filter gives us a way to describe how far away the states are from the code space, while the ideal decoder tells us the state of the logical qubit at each time step.  Using the $r$-filter and the ideal decoder, Ref.~\cite{Gottesman2009} established fault-tolerant conditions for gadgets. For a QEC code that can correct up to $t$ Pauli errors, the gadget that has at most $r$ faulty locations should output the state that is invariant under the action of $r$-filter and is decoded to the correct logical state via the ideal decoder if $r \leq t$.  Then, the threshold theorem can be proved by relying only on these fault-tolerance conditions and the locality of the noise, not relying on the detailed properties of the noise.
In this sense, the $r$-filter and the ideal decoder introduce an equivalence class of noise against which the fault-tolerant protocol works.  

Following the idea of Ref.~\cite{Gottesman2009}, we establish a new level-reduction technique for CV FTQC\@.  Sincethe GKP code is fundamentally different from qubit QEC codes, we need to introduce alternatives for the $r$-filter and the ideal decoder, which we name the SSS $r$-filter and the ideal GKP decoder.
As the name suggests, we use the stabilizer subsystem decomposition explained in the previous section.
For the SSS Hilbert space in Eq.~\eqref{eq:SSS_Hilbert_space}, the SSS $r$-filter is defined as the projection operator that acts trivially on the logical qubit $\mathbb{C}^2$ and projects onto the subregion $[-r,r)\times[-r,r)$ of the square area $\patch$ of the syndrome subsystem in Eq.~\eqref{eq:patch} (see Fig.~\ref{fig:demo_SSS_r-filter}).  The ideal GKP decoder is defined as the map that traces out $\mathcal{H}_S$ while leaving $\mathbb{C}^2$~\cite{Mackenzie2022}.
Thus, the SSS $r$-filter restricts the support of the wave function around the origin of the square area $\patch$ of the syndrome subsystem, in contrast to the $r$-filter of the qubit concatenated code restricting the weight of Pauli errors.
Again, the SSS $r$-filter defines an increasing net $({\cal H}_r)_{r\in(0,\sqrt{\pi}/2]}$ of Hilbert subspaces ${\cal H}_r$ under inclusion.
As $r$ gets smaller, all the states in the subspace ${\cal H}_r$ become better approximations of the ideal GKP states.  There is, however, a mathematical subtlety for this net; the projective limit of this net seen as a projective system (i.e., $\bigcap_{r}{\cal H}_r$) is not the \quoted{ideal GKP code space} in stark contrast to the qubit case.  (In fact, $\bigcap_{r}{\cal H}_r=\{0\}$.)

The underlying principle of our definition of the SSS $r$-filter is that a large displacement error occurs less frequently than a small displacement error even with the faulty operations, which is intended in the construction of the GKP code~\cite{Gottesman2001}.
See Sec.~\ref{sec:FT_conditions} in Supplementary Information for the rigorous definitions.
The SSS $r$-filter and the ideal GKP decoder are depicted as follows.
\begin{equation*}
\begin{picture}(230,40)
\thicklines

\put(0,20){\line(1,0){10}}
\put(10,10){\framebox(5,20){}}
\put(15,20){\line(1,0){10}}
\put(18,25){\makebox(0,0)[bl]{$\scriptstyle r$}}
\put(35,12){\makebox(50,16)[l]{\text{SSS $r$-filter}}}

\put(110,20){\line(1,0){10}}
\put(120,10){\line(0,1){20}}
\put(120,30){\line(1,-1){10}}
\put(130,20){\line(-1,-1){10}}
\thinlines
\put(130,20){\line(1,0){10}}
\put(150,12){\makebox(60,16)[l]{\text{Ideal GKP decoder}}}

\end{picture}
\end{equation*}
We will introduce the equivalence relation between noise through the SSS $r$-filter and the ideal GKP decoder as Ref.~\cite{Gottesman2009} does for the qubit case.
Let us define an $r$-parameterized GKP state $\hat{\rho}_{\psi}^s$ for a qubit state $\ket{\psi}$ as any state that is invariant under the action of SSS $r$-filter and outputs $\ket{\psi}\!\bra{\psi}$ under the action of the ideal GKP decoder.  This state is a product state in the SSS Hilbert space: it is $\ket{\psi}$ in the logical qubit and has support only in the region $[-r,r)\times[-r,r)$ in the square area $\patch$ of the syndrome subsystem in Eq.~\eqref{eq:patch}.  Unlike an ideal GKP state, an $s$-parameterized GKP state can be a physically realizable quantum state.

Let us further define the $s$-parameterized noise channel ${\cal N}^s$ as a noise channel whose Kraus operators are linear combinations of the elements of $\{\exp[i(u\hat{q} - v\hat{p})]:|u|,|v| < s\}$. 
Next, we define fault-tolerant GKP preparation, gate, measurement, and EC gadgets as follows.  The $s$-parameterized GKP preparation gadget ($s$-preparation in short) for a logical state $\ket{\overline{\psi}}$ is the preparation of $s$-parameterized GKP state.
The $s$-parameterized GKP gate gadget ($s$-gate in short) for a logical unitary $\overline{U}$ is the ideal CV unitary $\hat{U}$ listed in Eqs.~\eqref{eq:pauli-X_physical}--\eqref{eq:CNOT_physical} followed by the $s$-parameterized noise channel ${\cal N}^s$.  The $s$-parameterized GKP measurement gadget ($s$-measurement in short) is the $s$-parameterized noise channel ${\cal N}^s$ followed by the ideal homodyne detection and the binning of the outcome (see the main text).  The $s$-parameterized GKP EC gadget is a Knill-type GKP EC circuit with all the circuit components replaced with the $s_{\rm p}$-preparations, $s_{\rm g}$-gates, and $s_{\rm m}$-measurements, where the parameter $s$ is a function of $s_{\rm p}$, $s_{\rm g}$, and $s_{\rm m}$ (given below).
Diagrammatically, they are written as follows:
\setlength{\unitlength}{0.75pt}
\begin{align}
    \begin{picture}(60,25)
    \thicklines
        \put(30,5){\oval(50,20)[l]}
        \put(30,-5){\line(0,1){20}}
        \put(12,5){\makebox(0,0)[l]{$\ket{\overline{\psi}}$}}
        \put(30,5){\line(1,0){20}}
        \put(32,10){\makebox(0,0)[bl]{$\scriptstyle s$}}
    \end{picture}
    &\longrightarrow
    \begin{picture}(60,25)
    \thicklines
        \put(10,5){\makebox(0,0)[l]{$\hat{\rho}_{\psi}^s$}}
        \put(27,5){\line(1,0){25}}
    \end{picture},
    \label{eq:s-preparation_method}\\
    \begin{picture}(80,25)
    \thicklines
        \put(10,5){\line(1,0){20}}
        \put(40,5){\circle{20}}
        \put(40,5){\makebox(0,0){$\overline{U}$}}
        \put(50,5){\line(1,0){20}}
        \put(51,10){\makebox(0,0)[bl]{$\scriptstyle s$}}
    \end{picture}
    &\longrightarrow
    \begin{picture}(110,25)
    \thicklines
        \put(10,5){\line(1,0){20}}
        \put(30,-5){\framebox(20, 20){$\hat{U}$}}
        \put(50,5){\line(1,0){10}}
        \put(60,-5){\framebox(20, 20){${\cal N}^s$}}
        \put(80,5){\line(1,0){20}}
    \end{picture},
    \label{eq:s-gate_method} \\
    \begin{picture}(55,25)
    \thicklines
        \put(0,5){\line(1,0){20}}
        \put(20,-5){\line(0,1){20}}
        \put(20,5){\oval(50,20)[r]}
        \put(21,5){\makebox(20,0){$\overline{Z}$}}
        \put(46,10){\makebox(0,0)[bl]{$\scriptstyle s$}}
    \end{picture}
    &\longrightarrow
    \begin{picture}(105,25)
    \thicklines
        \put(10,5){\line(1,0){20}}
        \put(30,-5){\framebox(20,20){${\cal N}^s$}}
        \put(50,5){\line(1,0){10}}
        \put(60,-5){\line(0,1){20}}
        \put(60,5){\oval(70,20)[r]}
        \put(62,5){\makebox(27,0){$\hat{q}=t$}}
        \put(95,6){\line(1,0){10}}
        \put(95,4){\line(1,0){10}}
    \end{picture}
    \left\lfloor \tfrac{t}{\sqrt{\pi}}\right\rceil \text{\small mod 2},
    \label{eq:s-measurement_method} \\
    \rule{0pt}{40pt}
    \begin{picture}(75,25)
    \thicklines
        \put(0,5){\line(1,0){20}}
        \put(20,-5){\framebox(30,20){EC}}
        \put(50,5){\line(1,0){20}}
        \put(52,10){\makebox(0,0)[bl]{$\scriptstyle s'$}}
    \end{picture}
    &=
    \!\!\begin{picture}(190,40)
    \thicklines
        \put(30,5){\oval(40,20)[l]}
        \put(30,-5){\line(0,1){20}}
        \put(15,5){\makebox(0,0)[l]{$\ket{\overline{0}}$}}
        \put(30,5){\line(1,0){20}}
        \put(32,10){\makebox(0,0)[bl]{$\scriptstyle s_0$}}
        \put(90,30){\line(1,0){60}}
        \put(70,-20){\oval(40,20)[l]}
        \put(70,-30){\line(0,1){20}}
        \put(55,-20){\makebox(0,0)[l]{$\ket{\overline{0}}$}}
        \put(72,-15){\makebox(0,0)[bl]{$\scriptstyle s_0$}}
        \put(30,5){\line(1,0){20}}
        \put(70,-20){\line(1,0){40}}
        \put(60,5){\circle{20}}
        \put(60,5){\makebox(0,0){$\overline{H}$}}
        \put(71,10){\makebox(0,0)[bl]{$\scriptstyle s_H$}}
        \put(70,5){\line(1,0){80}}
        \put(90,-25){\line(0,1){30}}
        \put(90,5){\circle*{5}}
        \put(93,7){\makebox(0,0)[bl]{$\scriptstyle s_{\oplus}$}}
        \put(90,-20){\circle{10}}
        \put(120,-20){\circle{20}}
        \put(120,-20){\makebox(0,0){$\overline{I}$}}
        \put(131,-15){\makebox(0,0)[bl]{$\scriptstyle s_I$}}
        \put(130,-20){\line(1,0){20}}
        \put(123,32){\makebox(0,0)[bl]{$\scriptstyle s_{\oplus}$}}
        \put(120,30){\circle*{5}}
        \put(120,5){\circle{10}}
        \put(120,0){\line(0,1){30}}
        \put(150,20){\line(0,1){20}}
        \put(150,30){\oval(40,20)[r]}
        \put(158,30){\makebox(0,0){$\overline{X}$}}
        \put(171,35){\makebox(0,0)[bl]{$\scriptstyle s_{X}$}}
        \put(150,-5){\line(0,1){20}}
        \put(150,5){\oval(40,20)[r]}
        \put(158,5){\makebox(0,0){$\overline{Z}$}}
        \put(171,10){\makebox(0,0)[bl]{$\scriptstyle s_{Z}$}}
        \put(160,-20){\circle{20}}
        \put(160,-20){\makebox(0,0){$\overline{I}$}}
        \put(171,-15){\makebox(0,0)[bl]{$\scriptstyle s_I$}}
        \put(170,-20){\line(1,0){20}}
    \end{picture},\label{eq:s-EC_method}\\
    &\nonumber 
\end{align}
\setlength{\unitlength}{1pt} \\
where the parameter $s'$ for the EC gadget is given by 
\begin{align}
    s' \coloneqq 2s_0 + s_H + s_{\oplus} + \max\{s_{\oplus} + \max\{s_X,s_Z\}, 2s_I\}. \label{eq:s_in_EC_method}
\end{align}
Recall that we are not explicitly including the Pauli corrections since they result simply in a change of the Pauli frame that can be absorbed into future operations. 
Also note that the runtime of classical computation in the EC gadget (e.g., for binning of the measurement outcomes) can also be taken into account by running it during the wait operation in the EC gadget.

Now, with these definitions of fault-tolerant preparation, gate, measurement, and EC gadgets, we obtain the fault-tolerance (FT) conditions for the SSS $r$-filter and the ideal GKP decoder.
The fault-tolerant conditions for the $s$-preparation are that the state stays invariant under the $s$-filter and decodes to the ideal logical qubit state under the action of the ideal decoder if $s < \sqrt{\pi}/2$.  The conditions for the $s$-gate are that the output of the gate is invariant under $(r+s)$-filter with an $r$-filtered input state and it decodes to the ideal logical qubit gate under the action of the ideal decoder if $r+s < \sqrt{\pi}/2$.  The condition for the $s$-measurement is that, with $r$-filtered input state, it decodes to the ideal logical qubit measurement under the action of the ideal decoder if $r+s<\sqrt{\pi}/2$.  The conditions for the $s$-EC are that the output state is invariant under the action of $s$-filter and, with the $r$-filtered input, it decodes to the qubit identity map if $r+s<\sqrt{\pi}/2$.

It is important to stress that the satisfaction of the FT conditions for the gate gadget owes to the careful choice of the gate set in Eqs.~\eqref{eq:pauli-X_physical}--\eqref{eq:CNOT_physical}.  Here, the essential property that these gates have in common is that they do not enlarge displacement errors of the input (parameterized by $r$) up to the small additional constant $s$.  They play the same role for displacement errors as transversal gates play for local errors.  In this perspective, non-Gaussian gate operations such as $\exp(i\gamma\hat{q}^3)$ suggested in the original paper \cite{Gottesman2001} to implement the GKP non-Clifford operation will make a displacement error destructively larger and cannot satisfy these FT conditions.
Precise statements of the FT conditions and the proofs are shown in Sec.~\ref{sec:FT_conditions} in the Supplementary Information.

The FT-GKP circuit $C'$ for the qubit circuit $C$ is constructed by replacing all locations $\{C_i\}_{i\in{\cal I}}$ in $C$ with respective FT GKP gadgets, and the FT GKP EC gadgets are inserted between these gadgets.  To discuss the logical-level behavior of a given fault-tolerant circuit $C'$, Ref.~\cite{Gottesman2009} introduced the concept of an extended rectangle (ExRec), which we here generalize to the CV case.  An ExRec is a part of the CV physical circuit composed of a gadget to implement logical operations of GKP qubits and the adjacent GKP EC gadgets, as illustrated in Fig.~\ref{fig:circuit_encoding} in the main text. 
The preparation ExRec consists of a preparation gadget followed by an EC gadget; the gate ExRec consists of leading EC gadgets, a gate gadget, and trailing EC gadgets; and the measurement ExRec consists of an EC gadget followed by a measurement gadget.  
Although the meaning of the FT conditions in our setup introduced above is different from that of Ref.~\cite{Gottesman2009}, the equivalence relation of errors induced by these FT conditions is exactly the same form as that in Ref.~\cite{Gottesman2009}.  Thus, the correctness condition for the GKP ExRec, i.e., the condition that its qubit-level circuit works ideally, is the same definition as that in Ref.~\cite{Gottesman2009}.  Furthermore, we define that the GKP gate ExRec is \emph{good} if it consists of a leading $s_1$-EC, an $s_2$-gate, and a trailing $s_3$-EC such that $s_1+s_2+s_3<\sqrt{\pi}/2$.  We define the good GKP preparation ExRec and measurement ExRec in the same manner.  Then, with the same proof as that in Ref.~\cite{Gottesman2009}, we can show that the good GKP ExRec is correct.
Diagrammatically, it states that 
\begin{equation}
    \begin{split}
        &\begin{picture}(180,20)
    \thicklines
    \put(10,5){\line(1,0){15}}
    \put(25,-5){\framebox(30,20){\text{EC}}}
    \put(55,5){\line(1,0){15}}
    \put(57,10){\makebox(0,0)[bl]{$\scriptstyle s_1$}}
    \put(80,5){\circle{20}}
    \put(80,5){\makebox(0,0){$\overline{U}$}}
    \put(90,5){\line(1,0){15}}
    \put(91,10){\makebox(0,0)[bl]{$\scriptstyle s_2$}}
    \put(105,-5){\framebox(30,20){\text{EC}}}
    \put(135,5){\line(1,0){15}}
    \put(137,10){\makebox(0,0)[bl]{$\scriptstyle s_3$}}
    \put(150,-5){\line(0,1){20}}
    \put(150,15){\line(1,-1){10}}
    \put(160,5){\line(-1,-1){10}}
    \thinlines
    \put(160,5){\line(1,0){10}}
    \end{picture}\\
    \rule{0pt}{25pt}
    &=
    \begin{picture}(120, 20)
        \thicklines
    \put(10,5){\line(1,0){15}}
    \put(25,-5){\framebox(30,20){\text{EC}}}
    \put(55,5){\line(1,0){15}}
    \put(57,10){\makebox(0,0)[bl]{$\scriptstyle s_1$}}
    \put(70,-5){\line(0,1){20}}
    \put(70,15){\line(1,-1){10}}
    \put(80,5){\line(-1,-1){10}}
    \thinlines
    \put(80,5){\line(1,0){10}}
    \put(100,5){\circle{20}}
    \put(100,5){\makebox(0,0){$U$}}
    \put(110,5){\line(1,0){15}}
    \end{picture}
    \end{split}
    \label{eq:gate_correct_method}
\end{equation}
    \\
holds for the good GKP gate ExRec.  Similar diagrammatic identities hold for the good GKP preparation ExRec and measurement ExRec as well.

Unlike good GKP ExRecs, the logical qubit-level behavior of a bad GKP ExRec is not easily characterized.  The reason is that the resulting logical error may depend on the CV state of the preceding ExRec.  Thus, we cannot determine the erroneous behavior of the bad GKP ExRec solely with the given ExRec, but we need to see a larger context in the CV QC\@.

For this analysis, we introduce a new operation called the GKP $*$-decoder. This has a rather simple interpretation in the SSS representation: it is simply a unitary transformation mapping a state of a full mode to the equivalent state of the SSS Hilbert space $\mathbb{C}^2\otimes\mathcal{H}_S$ in Eq.~\eqref{eq:SSS_Hilbert_space}. This acts as a decoder because, in the SSS representation, the logical qubit is already what would result from applying ideal GKP error correction. The GKP $*$-decoder faithfully reproduces all information in the original state because it keeps the syndrome subsystem ${\cal H}_S$ (instead of tracing it out as is done in the ideal GKP decoder).
Diagrammatically, the GKP $*$-decoder is denoted as follows.
\begin{equation*}
\begin{picture}(100,30)
\thicklines
\put(0,20){\line(1,0){10}}
\put(10,10){\line(0,1){20}}
\put(10,30){\line(1,-1){10}}
\put(20,20){\line(-1,-1){10}}
\thinlines
\put(20,20){\line(1,0){10}}
\put(15,5){\line(0,1){10}}
\put(15,5){\line(1,0){15}}
\put(45,10){\makebox(55,16)[l]{\text{GKP $*$-decoder}}}
\end{picture}
\end{equation*}

\noindent 
For a good GKP-gate ExRec, the GKP $*$-decoder makes the top thin wire into the action of the ideal unitary $U$ while the bottom thin wire into the action of a channel.  Thus, the overall dynamics is decoupled between the logical qubit and the syndrome subsystem.
For a bad GKP-gate ExRec, we thus have the following.
\begin{equation}
\begin{split}
    &\begin{picture}(180,20)
    \thicklines
    \put(10,5){\line(1,0){15}}
    \put(25,-5){\framebox(30,20){\text{EC}}}
    \put(55,5){\line(1,0){15}}
    \put(57,10){\makebox(0,0)[bl]{$\scriptstyle s_l$}}
    \put(80,5){\circle{20}}
    \put(80,5){\makebox(0,0){$\overline{U}$}}
    \put(90,5){\line(1,0){15}}
    \put(91,10){\makebox(0,0)[bl]{$\scriptstyle s$}}
    \put(105,-5){\framebox(30,20){\text{EC}}}
    \put(135,5){\line(1,0){15}}
    \put(137,10){\makebox(0,0)[bl]{$\scriptstyle s_t$}}
    \put(150,-5){\line(0,1){20}}
    \put(150,15){\line(1,-1){10}}
    \put(160,5){\line(-1,-1){10}}
    \thinlines
    \put(155,-10){\line(0,1){10}}
    \put(155,-10){\line(1,0){15}}
    \put(160,5){\line(1,0){10}}
    \end{picture}\\
    \rule{0pt}{25pt} &=
    \begin{picture}(100, 20)
        \thicklines
    \put(10,5){\line(1,0){15}}
    \put(25,-5){\line(0,1){20}}
    \put(25,15){\line(1,-1){10}}
    \put(35,5){\line(-1,-1){10}}
    \thinlines
    \put(30,-10){\line(0,1){10}}
    \put(30,-10){\line(1,0){15}}
    \put(35,5){\line(1,0){10}}
    \put(55,5){\oval(20,20)[t]}
    \put(55,-10){\oval(20,20)[b]}
    \put(45,-10){\line(0,1){15}}
    \put(65,-10){\line(0,1){15}}
    \put(55,-2.5){\makebox(0,0){${\cal U}'$}}
    \put(65,5){\line(1,0){10}}
    \put(65,-10){\line(1,0){10}}
    \end{picture}
    \\
    &
    \end{split}
    \label{eq:bad_gate_exrec_method}
\end{equation}
\\
Unlike Eq.~\eqref{eq:gate_correct_method} for the good GKP ExRec, we move the GKP $*$-decoder in Eq.~\eqref{eq:bad_gate_exrec_method} entirely to the left to remove the leading GKP ExRec for a bad GKP ExRec to avoid the noise correlation at the logical level.  For details, see Sec.~\ref{sec:threshold_theorem} in Supplementary Information.
Thus, the preceding GKP ExRec loses the trailing EC gadget.  The GKP ExRec that loses the trailing EC gadget is called a truncated GKP ExRec.  We also define the correctness and goodness/badness of the truncated GKP ExRec similarly.

Now, we can obtain a qubit circuit induced by the GKP code.

\begin{definition}[Qubit circuit induced by the GKP code]\label{def:noisy_reduction_method}
    Consider an FT-GKP circuit $C'$ on CV systems for implementing the original circuit $C$ on qubits. Then, a circuit $\tilde{C}$ on qubits and environmental systems is defined by moving the GKP $*$-decoder from the end to the start of the circuit $C'$ and regarding the syndrome subsystems as the environmental systems.
\end{definition}

Recall that the GKP $*$-decoder transforms a good GKP ExRec for a location $C_i$ of $C$ to $C_i$ itself in $\tilde{C}$ while independently changing the environmental system.  Thus,
if all the GKP ExRecs are good, they are correct, and all the locations $C_i$ of $C$ are included in $\tilde{C}$ as it is; in this case, ignoring the environmental systems, we have $\tilde{C}=C$ and obtain the desired outcomes with the same probability.

However, the case where all the GKP ExRecs are good may be rare in the presence of noise.  To achieve the fault-tolerant quantum computation, we need to take into account how the noise affects each location $C_i$ and makes some of GKP ExRecs bad; even in this case, we need to prove that the fault-tolerant protocol can still correct the errors occurring in bad GKP ExRecs by concatenating the GKP code and a qubit code.
For this proof, we need to introduce an appropriate distance measure between CV quantum channels, as will be presented in the following.

\paragraph*{
Energy constraint conditions
}

In the presence of noise on CV systems, it may still be too much burden to require that the actual experiment perfectly realizes $s$-preparation, $s$-gate, and $s$-measurement for sufficiently small $s$.  Thus, here we consider noise models that cover general types of noise in CV systems and study how well the above conditions can be satisfied under the existence of such noise.  Since we are mainly focusing on the optical system, and since the optical system does not repeatedly interact with the same environmental mode in typical setups, the noise can be well modeled as Markovian; i.e., the noise map can be described as a quantum channel.  Furthermore, we assume that the modes used for CV quantum computation are well isolated so that they experience independent noise channels. 
These assumptions on the noise models are standard in the theoretical analysis of fault tolerance, and generalization to more general cases such as non-Markovian noise may be possible by a modification of our analysis, as can be seen in the qubit case~\cite{Aliferis2011, Aliferis2013}.

An immediate obstacle we face when we consider noisy physical operations on CV systems is the singular behavior of the diamond norm.  The diamond norm $\|\cdot\|_{\diamond}$ has been conventionally used in the context of fault-tolerance theory since it has several properties necessary to prove fault tolerance~\cite{Aharonov1997, Aharonov1999}.  However, it is too stringent for the noise models in CV systems as can be seen in the following example.  Consider the phase rotation channel ${\cal R}[\theta](\hat{\rho})=e^{i\theta\hat{n}}\hat{\rho}e^{-i\theta\hat{n}}$.  Then, one can show that $\|{\cal R}[\theta] - \mathrm{Id}\|_{\diamond} = 2$ holds for any $\theta\in(0,2\pi)$, where $\mathrm{Id}$ denotes the identity channel~\cite{Winter2017} since a coherent state with the amplitude $\alpha$ distinguishes these two channels better and better as $|\alpha|\to\infty$.  This means that an infinitesimally small phase rotation is maximally distant from the identity map in terms of the diamond norm even though such a small phase rotation is ubiquitous in experiments in practice.  This also implies that we do not have any fault-tolerant threshold against phase rotation noise, however small it is, as long as it is measured by the diamond norm.

To avoid this limitation, a physically meaningful distance measure has been considered, which is named the \emph{energy-constrained diamond norm}~\cite{Holevo2003, Shirokov2008, Winter2016, Winter2017, Shirokov2018, Shirokov2018adaptation, Shirokov2019}.  
The idea of the energy-constrained diamond norm is not to consider the set of all states (which may contain unphysical states) but rather to consider the subset of states that have bounded energy.  In the case of the quantum harmonic oscillator, this amounts to considering the set of states for which each one has an average photon number at most $E>0$; then, the $E$-energy-constrained diamond norm takes maximization over this energy-constrained subset of states instead of maximization over the set of all states, denoted by $\|\Phi\|_\diamond^E$ for a Hermitian-preserving map $\Phi$.  With this restriction of the set of states, we can avoid the singular behavior mentioned above~\cite{Winter2017, Shirokov2018}.

We want to utilize this distance measure, but we need to check that the energy-constrained diamond norm actually satisfies the properties required for the fault-tolerance proof and that the energy of a local state during the computation has a constant upper bound.
Whether the energy-constrained diamond norm satisfies the necessary properties for the fault-tolerance proof is nontrivial since---unlike the conventional diamond norm---a product of the energy-constrained diamond norms of Hermitian-preserving maps $\Phi$ and $\Psi$ is not an upper bound of the energy-constrained diamond norm of their composite map $\Phi\circ\Psi$.
This is because the operator $\Psi(\hat{\rho})$ may not be a positive operator for a general Hermitian-preserving map $\Psi$, and the \quoted{energy constraint} for non-positive operators does not have valid meanings.
Nevertheless, we show that the following facts hold for any $E, E'>0$ and the set $\mathfrak{S}_{E}({\cal H}_Q)$ of energy-constrained states on the mode $Q$. 
\begin{itemize}
    \item For any Hermetian-preserving linear map $\Phi_Q$ and for any $\hat{\rho}_{QR}$ with $\hat{\rho}_Q\in\mathfrak{S}_E({\cal H}_Q)$, we have $\|\Phi_Q\otimes\mathrm{Id}_R(\hat{\rho}_{QR})\|_1 \leq \|\Phi_Q\|_{\diamond}^E$.  When ${\cal H}_R\cong{\cal H}_Q$, there exists a state $\hat{\rho}_{QR}$ that achieves the equality.
    \item For any Hermetian-preserving linear maps $\Phi_{Q}$ and $\Psi_{Q'}$, we have $\|\Phi_{Q}\otimes \Psi_{Q'}\|_{\diamond}^{E,E'}\geq \|\Phi_{Q}\|_{\diamond}^{E}\|\Psi_{Q'}\|_{\diamond}^{E'}$.  The equality holds when at least one of $\Phi_{Q}$ and $\Psi_{Q'}$ is completely positive.
    \item For any CP map $\Phi:\mathfrak{S}_{E}({\cal H})\rightarrow \mathfrak{S}_{E'}({\cal H})$ and for any Hermetian-preserving linear map $\Psi$, we have $\|\Psi\circ\Phi\|_{\diamond}^E \leq \|\Psi\|_{\diamond}^{E'}\|\Phi\|_{\diamond}^E$.
    \item Let two density operators $\hat{\rho}_{QR}$ and $\hat{\sigma}_{QR}$ satisfy $\hat{\rho}_Q,\hat{\sigma}_Q\in\mathfrak{S}_E({\cal H}_Q)$ and $\|\hat{\rho}_{QR}-\hat{\sigma}_{QR}\|_1\leq 2\epsilon$ with $0<\epsilon<1$.  Then, for any Hermetian-preserving linear map $\Phi_Q$, we have $\|\Phi_Q\otimes\mathrm{Id}_R(\hat{\rho}_{QR}-\hat{\sigma}_{QR})\|_1\leq 10\epsilon\|\Phi_Q\|_{\diamond}^{E/\epsilon^2}$.
    \item Let $\Phi$ be a Hermetian-preserving linear map and $\Psi,\tilde{\Psi}:\mathfrak{S}_E({\cal H})\rightarrow \mathfrak{S}_{E'}({\cal H})$ be two CPTP maps such that $\|\Psi-\tilde{\Psi}\|_{\diamond}^E\leq 2\epsilon$.  Then, we have $\|\Phi\circ(\Psi-\tilde{\Psi})\|_{\diamond}^{E} \leq 10\epsilon\|\Phi\|_{\diamond}^{E'/\epsilon^2}$.
\end{itemize}
The last two statements are our new results, which we prove in Sec.~\ref{sec:energy_constrained_diamond_norm} in the Supplementary Information.  Although properties listed here appear to be substantially weaker than those satisfied by the conventional diamond norm, we later see that these are sufficient to prove fault tolerance.

For the problem of the existence of a constant upper bound on the energy of the local state during the computation, we impose that all the preparation gadgets generate a state with a constant energy bound $E_{\rm prep}$ and all the gate gadgets increase the energy at most finite amounts as a positive, monotonically increasing, locally bounded function $g_{\rm sup}(E)$ of the energy $E$ of the input state.  See Defs.~\ref{def:energy_constrained_prep} and \ref{def:energy_constrained_gate} in Sec.~\ref{sec:energy_constraint_conditions} in Supplementary Information for details.  
We assume that all the component gadgets in our fault-tolerant circuit satisfy the energy-constraint conditions in the sense written above, which is in principle testable in a real experiment.
Then, as explained in the main text, using quantum teleportation, Knill-type error correction continually resets the energy of the state. Due to the requirement for the CV fault-tolerant protocol presented in the main text, there must exist a constant $\ell$ such that any state prepared during the computation undergoes at most $\ell$ gates before being measured in the GKP EC gadget or at the final measurement of the CV circuit.  Thus, the reduced state on each mode during the computation has a constant energy upper bound $g_{\rm sup}^{\ell}(E_{\rm prep})$, where $g_{\rm sup}^{m}$ for $m\in\{1,\ldots,\ell\}$ is defined as 
\begin{equation}
    g_{\rm sup}^{m} \coloneqq \underbrace{g_{\rm sup}\circ g_{\rm sup} \circ \cdots\circ g_{\rm sup}}_{m}.
\end{equation}
See Sec.~\ref{sec:energy_constraint_conditions} in Supplementary Information for more details.

Having shown the upper bound on the local energy of the state during the computation, we define the noise model we will be working with: the $(s,\epsilon)$-independent Markovian noise model and $(E,s,\epsilon)$-independent Markovian noise model.
\begin{definition}[(Informal version of Def.~\ref{def:s_eps_Markovian_prep}--\ref{def:E_s_eps_Markovian_meas}) $(s,\epsilon)$-independent and $(E,s,\epsilon)$-independent Markovian noise model]\label{def:E_s_eps_Markovian_mothod}
    Let $E_{\rm prep}$ be a positive constant and $g_{\rm sup}$ be a fixed, positive, monotonically increasing, locally bounded function. Consider a collection ${\cal C}$ of physical systems that comprise a CV quantum computer, and regard other physical systems as environments.
    A noisy physical state preparation for the GKP logical state is said to obey the $(s,\epsilon)$-independent Markovian noise model if, independent of other physical systems in ${\cal C}$, it prepares a noisy state with the energy bound $E_{\rm prep}$
    that is $\epsilon$-close to an $s$-parameterized GKP state with the energy bound $E_{\rm prep}$ in the trace distance.
    Furthermore, a noisy physical gate operation is said to obey the $(E,s,\epsilon)$-independent Markovian noise model if, independent of other physical systems in ${\cal C}$, it implements the CPTP map satisfying the $g_{\rm sup}$-energy constraint that is $\epsilon$-close to the target unitary followed by the $s$-parameterized noise channel satisfying the $g_{\rm sup}$-energy constraint in the $E$-energy-constrained diamond norm.
    Finally, a noisy measurement is said to obey the $(E,s,\epsilon)$-independent Markovian noise model if, independent of other physical systems in ${\cal C}$, it implements a CPTP map that is $\epsilon$-close to the $s$-parameterized noise channel followed by the target measurement in the $E$-energy-constrained diamond norm.
\end{definition}
The formal statement is given in Defs.~\ref{def:s_eps_Markovian_prep}--\ref{def:E_s_eps_Markovian_meas} in Sec.~\ref{sec:ftqc_noise_model} of the Supplementary Information.
The above noise models cover physically relevant noise channels that we conventionally consider in CV quantum systems.  Explicit examples are also given in Sec.~\ref{sec:ftqc_noise_model} in Supplementary information.

Now, we come to the point of stating our level-reduction theorem.  The level-reduction theorem describes how the physical noise model translates to the logical ones through the fault-tolerant circuit and protocol.  Let us choose the parameters $s_{\rm p}$, $s_{\rm g}$, $s_{\rm m}$, and $s_{\rm e}$ sufficiently small so that the GKP ExRecs consisting of $s_{\rm p}$-preparation, $s_{\rm g}$-gate, $s_{\rm m}$-measurement, and $s_{\rm e}$-EC are always good and thus correct.
Define $E_{\max}^{\ell}(\epsilon)$ as
\begin{equation}
    \begin{split}
        E_{\max}^{\ell}(\epsilon) = g_{\rm sup}\left(\frac{g_{\rm sup}^{\ell-1}(E_{\rm prep})}{\epsilon^2}\right),
    \end{split}
\end{equation}
where $\ell$ is again the maximum number of gates that a prepared state during the computation undergoes.  Then, we have the following.
\begin{theorem}[(Informal statement of Thm.~\ref{theo:level_reduction}) Level reduction]\label{theo:level_reduction_method}
    Consider an FT-GKP circuit $C'$ on CV systems for implementing the original circuit $C$ on qubits.
    Suppose that all the preparation gadgets in $C'$ satisfy the $E_{\rm prep}$-energy constraint and all the gate gadgets in $C'$ satisfy the $g_{\rm sup}$-energy constraint.  Suppose further that the physical CV circuit $C'$ experiences the $(s_{\rm p},\epsilon)$-independent Markovian noise for state preparations, the $(E_{\max}^{\ell}(\epsilon),s_{\rm g},\epsilon)$-independent Markovian noise for gates, and the $(E_{\max}^{\ell}(\epsilon),s_{\rm m},\epsilon)$-independent Markovian noise for measurements, where $0<\epsilon<1$. 
    Then, the logical-qubit circuit $\tilde{C}$ implied from the FT-GKP circuit $C'$ by the procedure of Def.~\ref{def:noisy_reduction_method} undergoes a local Markovian noise with noise strength $\epsilon_{\rm qubit}$ at each location satisfying
    \begin{equation}
        \epsilon_\mathrm{qubit} = {\cal O}(\epsilon L_{\max}),
    \end{equation}
    where $L_{\max}$ denotes the maximum number of the GKP preparation, gate, and measurement gadgets in any truncated GKP ExRec of $C'$.
\end{theorem}
The proof is given in Sec.~\ref{sec:threshold_theorem} in the Supplementary Information, with the explicit coefficient for the ${\cal O}(\epsilon L_{\max})$ term.  Thus, we can immediately obtain a threshold noise strength $\epsilon_{\rm th}$ in the CV system, given the threshold noise strength $\epsilon^\star_{\rm qubit}$ against the local Markovian noise model for the qubit QEC code that is concatenated with the GKP code. 
Since a qubit concatenated code has been known to have a fault-tolerance threshold $\epsilon^\star_{\rm  qubit}$ against a local Markovian noise model \cite{Aliferis2013}, we reach Theorem~\ref{theo:threshold_main_text} in the main text (see Supplementary Information for details).

\section*{Data availability}

No data is used in this study.

\section*{Code availability}

The code used in this work is available from the corresponding author upon reasonable request.

\clearpage
\appendix
\onecolumngrid

\renewcommand{\thetheorem}{S\arabic{theorem}}
\renewcommand{\theproposition}{S\arabic{proposition}}
\renewcommand{\thelemma}{S\arabic{theorem}}
\renewcommand{\thecorollary}{S\arabic{theorem}}
\renewcommand{\thedefinition}{S\arabic{definition}}
\setcounter{equation}{0}
\setcounter{theorem}{0}

\section*{Supplementary Information}

Supplementary Information of \quoted{Continuous-Variable Fault-Tolerant Quantum Computation under General Noise} is organized as follows.  Appendix~\ref{sec:setting} summarizes preliminaries and the physical setups of continuous-variable (CV) quantum computation.
In Sec.~\ref{sec:stabilizer_ssd}, we review the stabilizer subsystem decomposition developed in Ref.~\cite{Mackenzie2022}.  Explicit representations of some CV quantum operations in this decomposition are also given.  Sec.~\ref{sec:physical_operations} defines the physical operations to implement quantum circuits with CV systems.  In Sec.~\ref{sec:energy_constrained_diamond_norm}, we define the energy-constrained diamond norm, which will be used to measure how close the noisy map implemented by the physical circuit is to the ideal map.  We also list its properties that are used in the fault-tolerance proof.
In Sec.~\ref{sec:ftqc_noise_model}, we define the noise model on the continuous-variable system against which we prove the threshold theorem in the later section.  Appendix~\ref{sec:general_FT} is the main part of constructing the fault-tolerant gadget and proving the threshold theorem.  We construct fault-tolerant Gottesman-Kitaev-Preskill (GKP) gadgets in Sec.~\ref{sec:gadget_construction} and prove their fault-tolerant conditions in Sec.~\ref{sec:FT_conditions}.  We further impose the energy-constraint condition for the fault-tolerant (FT) gadget for the later analysis and check its validity in Sec.~\ref{sec:energy_constraint_conditions}.  Finally in Sec.~\ref{sec:threshold_theorem}, we prove our main result, the threshold theorem for the CV circuit, in full detail.

\section{Setting}{\label{sec:setting}}

\subsection{GKP code and stabilizer subsystem decomposition}

\label{sec:stabilizer_ssd}

A single-mode harmonic oscillator can be characterized by a pair of noncommutative quadrature operators $\hat{q} = \tfrac 1 {\sqrt 2}(\hat{a} + \hat{a}^\dag)$ and $\hat{p} = \tfrac {-i} {\sqrt 2}(\hat{a} - \hat{a}^\dag)$, satisfying $[\hat{q},\hat{p}]=i$ (with $\hbar = 1$), acting on a separable Hilbert space ${\cal H}$ \cite{Holevo2011}.
We consider the square-lattice GKP code for encoding a qubit into a harmonic oscillator.  
The ideal square-lattice GKP-encoded state (the GKP codeword) $\ket{\overline{\psi}}$ is formally defined as an infinite superposition of position eigenstates, i.e., $\ket{\overline{\psi}}=\alpha\ket{\overline{0}}+\beta\ket{\overline{1}}$ with 
\begin{equation}
\ket{\overline{j}}\coloneqq \sum_{s\in\mathbb{Z}}\ket{\sqrt{\pi}(2s+j)}_q \label{eq:ideal_gkp_codeword}
\end{equation}
for $j=0,1$, where $\ket{x}_q$ satisfies $\hat{q}\ket{x}_q=x\ket{x}_q$ for the operator $\hat{q}$.

The GKP code can be regarded as a stabilizer code.  To explain this, let $\hat{V}(v_1, v_2)$ be defined as
\begin{equation}
    \hat{V}(v_1, v_2) \coloneqq \exp[-i (v_1 \hat{p} - v_2 \hat{q})],
\end{equation}
which leads to
\begin{equation}
    \hat{V}(u_1,u_2)\hat{V}(v_1,v_2)=e^{i(v_1u_2-u_1v_2)/2}\hat{V}(u_1+v_1,u_2+v_2).
\end{equation}
Then, one can check that the state given in Eq.~\eqref{eq:ideal_gkp_codeword} is stabilized by the two stabilizer generators $\hat{S}_1$ and $\hat{S}_2$ given by
\begin{align}
    \hat{S}_1 = \hat{V}(2c, 0), \label{eq:stabilizer_X}\\
    \hat{S}_2 = \hat{V}(0, 2c), \label{eq:stabilizer_Z}
\end{align}
where we write for simplicity
\begin{equation}
\label{eq:C}
    c \coloneqq \sqrt{\pi}.
\end{equation}
The ideal GKP codeword $\ket{\overline{j}}$ is stabilized by elements of the stabilizer obtained from these stabilizer generators~\cite{Gottesman2001}. 

Most of the logical unitary gates on the GKP code can be realized by the Gaussian unitaries, which are generated by first- or second-order polynomials of $\hat{q}$ and $\hat{p}$.
In fact, the logical Pauli-X and -Z operators $\overline{X}$ and $\overline{Z}$ for this GKP code are given by
\begin{align}
    \overline{X} \coloneqq \hat{V}(c, 0), \\
    \overline{Z} \coloneqq \hat{V}(0, c).
\end{align}
The logical Hadamard gate $\overline{H}$ can be implemented by the Fourier gate $\hat{F}$, which is defined as
\begin{align}
    \hat{F}\coloneqq \hat{R}(\pi/2),  \label{eq:fourier_gate}
\end{align}
where 
\begin{align}
    \hat{R}(\theta) & \coloneqq \exp[i \theta \hat{n}], \\
    \hat{n}&\coloneqq (\hat{q}^2+\hat{p}^2-1)/2.
\end{align}
The logical CNOT gate $\overline{\text{CNOT}}$ on the system $2$ controlled by $1$ can be implemented by the SUM gate, defined as
\begin{equation}
    \text{SUM}\coloneqq \exp(-i\hat{q}_1\hat{p}_2), \label{eq:sum_gate}
\end{equation}
where $\hat{q}_i$ and $\hat{p}_i$ denote the quadrature operators in $i^{\text{th}}$ system \cite{Gottesman2001}.  
Other logical gates can be implemented through gate teleportation as long as one can prepare their magic states \cite{Zhou2000}.  More precisely, the logical phase gate $\overline{S}$ can be implemented through the gate teleportation with the state $\ket{\overline{Y}}\coloneqq(\ket{\overline{0}}+i\ket{\overline{1}})/\sqrt{2}$, and the logical $T$ gate $\overline{T}$ can be with the state $\ket{\frac{\pi}{8}}\coloneqq (\ket{\overline{0}}+e^{\pi i/4}\ket{\overline{1}})/\sqrt{2}$.  Note that the logical phase gate $\overline{S}$ can also be implemented by the Gaussian unitary $\exp(i\hat{q}^2/2)$, which will not be used in this paper to simplify the analysis.
Measurements in the logical Pauli-$X$ or $Z$ bases can be implemented, respectively, by homodyne detection in $\hat{p}$ or $\hat{q}$ quadrature followed by classical post-processing.
The positive-operator-valued measure (POVM) of the $q$-homodyne detection is given by $\ket{x}\!\bra{x}_q dx$.
The classical post-processing is performed by the binning modulo $\sqrt{\pi}$ (the so-called GKP binning), and if binned to the even (respectively,~odd) multiple of $\sqrt{\pi}$, the outcome is regarded as logical 0 (respectively,~1).  These are a brief summary of a set of operations for implementing universal quantum computation with the GKP code.

There is a suitable representation for analyzing the GKP state and the GKP code, which we call the Zak representation \cite{Zak1967, Zak1968, Zak1972}.
This is explicitly pointed out in Ref.~\cite{Pantaleoni2020} and later used in several papers \cite{Pantaleoni2021Subsystem, Pantaleoni2021Hidden, Pantaleoni2023}, while it has already appeared in an early study \cite{Glancy2006}.
With this particular representation of Hilbert space, Ref.~\cite{Pantaleoni2020} developed the subsystem decomposition technique, which introduces a tensor-product structure between the GKP logical subsystem and the rest (called a gauge mode in Ref.~\cite{Pantaleoni2020}).
More recently, an alternative for this picture has been developed in Ref.~\cite{Mackenzie2022}.
Unlike the subsystem decomposition in Ref.~\cite{Pantaleoni2020}, the subsystem decomposition in Ref.~\cite{Mackenzie2022} has symmetry in position and momentum and has a direct connection to the ideal GKP decoding. 
For this reason, we will use the formalism of Ref.~\cite{Mackenzie2022}. For completeness and ease of reference, we now review some of the results therein.
Although the results in Ref.~\cite{Mackenzie2022} apply to any logical dimensions and any number of modes, here we restrict our attention to the single-mode, square-lattice GKP qubit.

The Zak basis for the square-lattice GKP qubit is defined as \cite{Mackenzie2022}
\begin{align}
    \ket{z_1, z_2}_Z &\coloneqq \frac{e^{iz_1 z_2/2}}{\sqrt{c}}\sum_{s\in \mathbb{Z}} e^{2c i z_2 s} \ket{z_1 + 2c s}_{q}
\\
&
=
    \hat{V}(z_1, z_2) \ket{0, 0}_Z
    , \label{eq:Zak_in_position}
\end{align}
for $z_1\in [-\frac{c}{2}, \frac{3c}{2})$ and $z_2\in [-\frac{c}{2}, \frac{c}{2})$.
The expression above defines the unitary transformation between the Zak and position representations.  
(Mathematically, the Hilbert space for the former representation is $L^2\bigl([-\frac{c}{2}, \frac{3c}{2})\times [-\frac{c}{2}, \frac{c}{2})\bigr)$.) 
As discussed in Ref.~\cite{Pantaleoni2023}, we can enlarge the domain of $(z_1, z_2)$ to $\mathbb{R}\times \mathbb{R}$ with the following quasi-periodic condition:
\begin{align}
    \ket{z_1+2c,z_2}_Z &= e^{-ci z_2} \ket{z_1, z_2}_Z, \label{eq:quasi_period_1}\\
    \ket{z_1, z_2 + c}_Z &= e^{ci z_1/2}\ket{z_1, z_2}_Z.\label{eq:quasi_period_2}
\end{align}
Let $(\bar{z}_1, \bar{z}_2) \in \mathbb{R}\times \mathbb{R}$.  Let $\lfloor\bar{z}\rceil_{a}$ be rounding $\bar{z}$ to integer multiple of $a$ and $\{\bar{z}\}_{a}$ be its residual, i.e.,
\begin{align}
    \lfloor\bar{z}\rceil_{a}&\coloneqq a\,\left\lfloor \frac{\bar{z}}{a} + \frac{1}{2}\right\rfloor, \label{eq:rounding}\\
    \{\bar{z}\}_{a}&\coloneqq \bar{z} - \lfloor\bar{z}\rceil_{a},
\end{align}
where $\lfloor{}\cdot{}\rfloor$ is the floor function.
Then from Eqs.~\eqref{eq:quasi_period_1} and \eqref{eq:quasi_period_2}, we should interpret $\ket{\bar{z}_1, \bar{z}_2}_Z$ for $(\bar{z}_1, \bar{z}_2)\in \mathbb{R}\times\mathbb{R}$ as
\begin{equation}
    \ket{\bar{z}_1, \bar{z}_2}_Z = e^{\frac{i}{2}(\lfloor\bar{z}_1\rceil_{2c}\lfloor\bar{z}_2\rceil_{c}-\lfloor\bar{z}_1\rceil_{2c}\{\bar{z}_2\}_{c}+\{\bar{z}_1\}_{2c}\lfloor\bar{z}_2\rceil_{c})}\ket{\{\bar{z}_1\}_{2c}, \{\bar{z}_2\}_{c}}_Z.\label{eq:reinterpretation}
\end{equation}
Note that $\exp[\frac{i}{2}\lfloor\bar{z}_1\rceil_{2c}\lfloor\bar{z}_2\rceil_{c}]=\exp[-\frac{i}{2}\lfloor\bar{z}_1\rceil_{2c}\lfloor\bar{z}_2\rceil_{c}]$. 
The Zak basis satisfies the following orthogonal property:
\begin{equation}
    \braket{z_1, z_2|z'_1, z'_2}_Z = \delta^{(2c)}(z_1 - z'_1)\,\delta^{(c)}(z_2 - z'_2), \label{eq:orthogonality}
\end{equation}
where $\delta^{(a)}(x)$ denotes the delta distribution on the interval $(-\frac{a}{2},\frac{a}{2})$; i.e., whenever the indices $z_1-z'_1$ or $z_2-z'_2$ are outside the respective intervals, then it should be pushed back to the intervals by using the quasi-periodicity Eqs.~\eqref{eq:quasi_period_1} and \eqref{eq:quasi_period_2}.
The Zak basis also satisfies the following completeness relation:
\begin{equation}
    \int_{-\frac{c}{2}}^{\frac{3c}{2}}dz_1\int_{-\frac{c}{2}}^{\frac{c}{2}}dz_2 \ket{z_1,z_2}\!\bra{z_1,z_2}_Z = \hat{I}, \label{eq:completeness}
\end{equation}
where $\hat{I}$ denotes the identity operator on the Hilbert space.

The ideal GKP state given in Eq.~\eqref{eq:ideal_gkp_codeword} is now given in the Zak basis as $\ket{\overline{0}}\propto\ket{0,0}_Z$ and $\ket{\overline{1}} \propto \ket{c, 0}_Z$.  
This defines a qubit in an oscillator.
However, these are not states in the usual sense; i.e., these are not elements of the Hilbert space.
Instead, we can define a qubit in the following way.  
We write $z_1 \in [-\frac{c}{2},\frac{3c}{2})$ as $z_1=\lfloor z_1 \rceil_c + \{z_1\}_c$ and regard $\lfloor z_1 \rceil_c/c$ as a bit.  Then, we define
\begin{align}
    \ket{0;z_1, z_2}_{LS} &\coloneqq \ket{z_1, z_2}_{Z} = \hat{V}(z_1,z_2)\ket{\overline{0}}, \label{eq:shifted_zero}\\
    \ket{1;z_1, z_2}_{LS} &\coloneqq e^{ciz_2/2}\ket{z_1+c, z_2}_{Z} = \hat{V}(z_1,z_2)\ket{\overline{1}}, \label{eq:shifted_one}
\end{align}
for $z_1, z_2 \in[-\frac{c}{2},\frac{c}{2}) $.  In this way, we can decompose the Hilbert space ${\cal H}$ into the tensor product of a qubit Hilbert space $\mathbb{C}^2$ and an infinite-dimensional Hilbert space ${\cal H}_{S}$, i.e., $\mathcal{H}=\mathbb{C}^2\otimes \mathcal{H}_S$, with
\begin{equation}
    \ket{\mu}_{L}\otimes\ket{z_1, z_2}_S = \ket{\mu;z_1, z_2}_{LS}. \label{eq:stabilizer_ssd}
\end{equation}
This decomposition is named the \emph{stabilizer subsystem (SSS) decomposition} in Ref.~\cite{Mackenzie2022}.  This decomposition is analogous to decomposing the Hilbert space of physical qubits for a multi-qubit stabilizer code into the tensor product of the logical qubit and the syndrome subsystem, as discussed in Methods.  Thus, we call the system $L$ with the Hilbert space $\mathbb{C}^2$ in the decomposition the logical qubit and the system $S$ with the Hilbert space ${\cal H}_S$ the syndrome subsystem. Note that this is a sub\emph{system} decomposition of a single mode~\cite{Pantaleoni2020,Mackenzie2022}, with the logical qubit being in a tensor product with the syndrome subsystem. This differs from typical approaches to error correction, which focus on (`logical' and `error') sub\emph{spaces} of a multi-qubit Hilbert space. We emphasize our approach by the use of the term `subsystem' to identify the Hilbert spaces involved.

In the following, we will use the notation on both sides of Eq.~\eqref{eq:stabilizer_ssd} interchangeably.
The quasi-periodic conditions~\eqref{eq:quasi_period_1} and \eqref{eq:quasi_period_2} lead us to obtain
\begin{align}
    \ket{\mu;z_1+c, z_2}_{LS} &= e^{-ciz_2/2}\ket{\mu\oplus 1;z_1, z_2}_{LS}, \label{eq:quasi_period_zak_position}\\
    \ket{\mu;z_1, z_2+c}_{LS} &= e^{ciz_1/2} (-1)^{\mu} \ket{\mu;z_1, z_2}_{LS}, \label{eq:quasi_period_zak_momentum}
\end{align}
where $\oplus$ denotes the summation modulo 2.
The orthogonality relation~\eqref{eq:orthogonality} and the completeness relation~\eqref{eq:completeness} now become
\begin{align}
    \braket{\mu;z_1,z_2|\nu;z^\prime_1,z^\prime_2}_{LS} &= \delta_{\mu \nu}\, \delta^{(c)}(z_1-z^\prime_1)\,\delta^{(c)}(z_2-z^\prime_2),\label{eq:orthogonality_SSS}\\
    \sum_{\mu=0,1}\int_{-\frac{c}{2}}^{\frac{c}{2}}dz_1\int_{-\frac{c}{2}}^{\frac{c}{2}} dz_2 \ket{\mu;z_1,z_2} \! \bra{\mu;z_1,z_2}_{LS} &= \hat{I}.\label{eq:completeness_SSS}
\end{align}
In the following, we sometimes use a notation $\ket{\psi;z_1, z_2}$ with a two-dimensional complex vector state $\ket{\psi}=\alpha\ket{0}+\beta\ket{1}$ to denote 
\begin{equation}
    \ket{\psi;z_1, z_2}_{LS} \coloneqq \alpha\ket{0;z_1, z_2}_{LS} + \beta\ket{1; z_1, z_2}_{LS} = \ket{\psi}_L\otimes \ket{z_1,z_2}_S = \hat{V}(z_1, z_2) \ket{\overline{\psi}},
\end{equation}
where $\ket{\overline{\psi}}$ denotes an ideal GKP state.

This representation developed in Ref.~\cite{Mackenzie2022} has favorable properties.
In particular, important Gaussian operations are concisely represented in this representation.
Let $\hat{z}_1$ and $\hat{z}_2$ be operators satisfying 
\begin{align}
    \hat{z}_1 \ket{z_1, z_2}_{S} &= z_1 \ket{z_1, z_2}_{S}, \\*
    \hat{z}_2 \ket{z_1, z_2}_{S} &= z_2 \ket{z_1, z_2}_{S},
\end{align}
which implies $[\hat{z}_1, \hat{z}_2]=0$.  Formally, $\hat{z}_1$ and $\hat{z}_2$ correspond to $\{\hat{q}\}_{c}$ and $\{\hat{p}\}_{c}$, respectively, in the ordinary Hilbert space. 
Then, the stabilizer generators $\hat{S}_1$ and $\hat{S}_2$ defined in Eqs.~\eqref{eq:stabilizer_X} and \eqref{eq:stabilizer_Z} are represented in the SSS decomposition as
\begin{align}
    \hat{S}_1 &= \hat{I}_L\otimes e^{-2ci\hat{z}_2}_S, \\*
    \hat{S}_2 &= \hat{I}_L\otimes e^{2ci\hat{z}_1}_S.
\end{align}
Similarly, for the GKP logical operators $\overline{X}$ and $\overline{Z}$, we have
\begin{align}
    \overline{X} &= \hat{X}_L \otimes e^{-ci\hat{z}_2}_S, \label{eq:logical_X} \\*
    \overline{Z} &= \hat{Z}_L \otimes e^{ci\hat{z}_1}_S, \label{eq:logical_Z}
\end{align}
where $\hat{X}$ and $\hat{Z}$ are qubit Pauli operators. For Fourier gate $\hat{F}$ defined in Eq.~\eqref{eq:fourier_gate}, we have \cite{Mackenzie2022}
\begin{equation}
    \hat{F}=\hat{H}_L\otimes \hat{F}'_S,  \label{eq:fourier_subsystem}
\end{equation}
where $\hat{H}$ is a qubit Hadamard gate and $\hat{F}'$ is defined to satisfy
\begin{equation}
    \hat{F}'\ket{z_1, z_2}_S = \ket{-z_2, z_1}_S. \label{eq:rotation_SSS}
\end{equation}

General Gaussian unitaries nontrivially couple the logical qubit and the syndrome subsystem.  
For example, the unitary $\exp(i\hat{q}^2/2)$ that works as a phase gate $\overline{S}$ for the GKP-encoded state is represented as
\begin{align}
    \exp(i\hat{q}^2/2)\ket{\mu;z_1,z_2}_{LS} &= i^{\mu}\ket{\mu;z_1,z_1+z_2}_{LS}\\ & = \begin{cases} (-i)^{\mu}e^{-ciz_1/2}\ket{\mu;z_1,z_1+z_2+c}_{LS}& z_1 + z_2 < -c/2, \\
    i^{\mu}\ket{\mu;z_1,z_1+z_2}_{LS} & -c/2 \leq z_1+z_2 < c/2, \\
    (-i)^{\mu}e^{ciz_1/2}\ket{\mu;z_1,z_1+z_2-c}_{LS}& c/2 \leq z_1 + z_2, 
    \end{cases}  \label{eq:GKP_phase}
\end{align}
where we used Eq.~\eqref{eq:quasi_period_2}.
Another example is the SUM gate defined in Eq.~\eqref{eq:sum_gate}, which is represented as
\begin{equation}
    \exp(-i\hat{q}_j\hat{p}_k)\ket{\mu;z_1, z_2}_{LS,j}\ket{\nu;z'_1, z'_2}_{LS,k}= \ket{\mu;z_1,z_2-z'_2}_{LS,j}\ket{\mu\oplus\nu;z'_1+z_1,z'_2}_{LS,k},
    \label{eq:CX}
\end{equation}
where we forego writing down each case explicitly this time and leave interested readers to derive them using Eqs.~\eqref{eq:quasi_period_1} and \eqref{eq:quasi_period_2} (see also Ref.~\cite{Mackenzie2022}).
In this way, general Gaussian unitaries do not act as a tensor product of operators on the logical qubit and the syndrome subsystem but rather entangle the two subsystems.

Let us finally observe how the POVM $\ket{x}\!\bra{x}_q dx$ of the homodyne detection can be represented in the Zak basis and the stabilizer subsystem decomposition.  First, we observe that
\begin{align}
    \ket{x}_q &= \int_{-\frac{c}{2}}^{\frac{3c}{2}}dz_1\int_{-\frac{c}{2}}^{\frac{c}{2}} dz_2 \ket{z_1,z_2}\!\bra{z_1,z_2}_Z \ket{x}_q \\
    &= \int_{-\frac{c}{2}}^{\frac{3c}{2}}dz_1\int_{-\frac{c}{2}}^{\frac{c}{2}} dz_2 \ket{z_1,z_2}_Z  \frac{e^{-iz_1 z_2/2}}{\sqrt{c}} \sum_{s\in\mathbb{Z}} e^{-2ciz_2 s}\delta(x-z_1-2cs) \\
    &= \int_{-\frac{c}{2}}^{\frac{c}{2}} dz_2\, \frac{e^{-iz_2(x+2c\bar{s}(x))/2}}{\sqrt{c}} \ket{x-2c\bar{s}(x),z_2}_Z ,\label{eq:position_to_Zak}
\end{align}
where $\bar{s}(x)$ is a unique integer that satisfies $-\frac{c}{2}\leq x - 2c\bar{s}(x) < \frac{3c}{2}$.  Thus, the homodyne POVM can be written in the Zak basis as
\begin{equation}
    \ket{x}\!\bra{x}_q dx = \iint_{-\frac{c}{2}}^{\frac{c}{2}} \frac{e^{-i(z_2-z'_2)(x+2c\bar{s}(x))/2}}{c} \ket{x-2c\bar{s}(x),z_2}\!\bra{x-2c\bar{s}(x),z'_2}_Z dz_2 dz'_2 dx.
\end{equation}
For example, a homodyne POVM element for the interval $\bigl[2cm + \Delta_1,2cm + \Delta_2\bigr)$ with an integer $m$ and $-\frac{c}{2} \leq \Delta_1 < \Delta_2\leq\frac{3c}{2}$ is given by
\begin{equation}
    \int_{2cm+\Delta_1}^{2cm + \Delta_2} \ket{x}\!\bra{x}_q dx = \int_{\Delta_1}^{\Delta_2} dz_1 \iint_{-\frac{c}{2}}^{\frac{c}{2}} dz_2 dz'_2 \frac{e^{-i(z_2-z'_2)(z_1 + 4c m)/2}}{c} \ket{z_1,z_2}\!\bra{z_1, z'_2}_Z,
\end{equation}
where $z_1 = x - 2cm$ and we used $\bar{s}(x)=m$.
If we further sum up the above over all the integers $m$, we have
\begin{align}
    \sum_{m=-\infty}^{\infty} \int_{2cm+\Delta_1}^{2cm+\Delta_2} \ket{x}\!\bra{x}_q dx &= \int_{\Delta_1}^{\Delta_2} dz_1 \iint_{-\frac{c}{2}}^{\frac{c}{2}}  dz_2dz'_2 \frac{e^{-i(z_2-z'_2)z_1/2}}{c}\sum_{n=-\infty}^{\infty} \delta\bigl((z_2-z'_2)/c-n\bigr) \ket{z_1,z_2}\!\bra{z_1,z'_2}_Z \\
    &= \int_{\Delta_1}^{\Delta_2}dz_1 \int_{-\frac{c}{2}}^{\frac{c}{2}} dz_2 \ket{z_1,z_2}\!\bra{z_1,z_2}_Z
    \label{eq:GKP_Z_meas}
\end{align}
where we used the Poisson summation formula $\sum_{m=-\infty}^{\infty}e^{-2\pi i x m}=\sum_{n=-\infty}^{\infty}\delta(x-n)$ in the first equality and $-c<z_2-z'_2<c$ in the second equality.
When $\Delta_1=-\frac{c}{2}$ and $\Delta_2=\frac{c}{2}$, the measurement operator in the left-hand side of Eq.~\eqref{eq:GKP_Z_meas} corresponds to the event in which homodyne detection is performed, the outcome is binned to an integer multiple of $\sqrt{\pi}$, and the resulting integer is even.  This is nothing but the logical $\ket{\overline{0}}$ event of the $Z$ measurement on the GKP code.  In fact, Eq.~\eqref{eq:GKP_Z_meas} can be rewritten in this case as
\begin{equation}
    \sum_{m=-\infty}^{\infty} \int_{2cm-\frac{c}{2}}^{2cm+\frac{c}{2}} \ket{x}\!\bra{x}_q dx = \ket{\overline{0}}\!\bra{\overline{0}}_L \otimes \iint_{-\frac{c}{2}}^{\frac{c}{2}}dz_1 dz_2\ket{z_1,z_2}\!\bra{z_1,z_2}_S = \ket{\overline{0}}\!\bra{\overline{0}}_L\otimes \hat{I}_S.
\end{equation}
The same applies to the case $\Delta_1=\frac{c}{2}$ and $\Delta_2=\frac{3c}{2}$, which corresponds to the GKP logical $\ket{\overline{1}}$ measurement.

\subsection{Physical operations to implement the continuous-variable quantum computation with the GKP code}

\label{sec:physical_operations}

A fault-tolerant protocol aims to simulate the original circuit $C$ on qubits (representing quantum computation) within any given target error $\epsilon>0$, i.e., to output a bit string sampled from a probability distribution close to the output probability distribution of the original circuit within error in the total variation distance at most $\epsilon$~\cite{Gottesman2009,yamasaki2022timeefficient}.
In this paper, we consider a CV fault-tolerant protocol achieving this simulation using a fault-tolerant circuit $C'$ on CV systems with noise occurring according to a noise model specified later in Sec.~\ref{sec:ftqc_noise_model}.
We say that we can achieve fault-tolerant quantum computation if, for any $\epsilon>0$ and any original circuit $C$, we can construct a fault-tolerant circuit $C'$ achieving this simulation under the noise model.

Here, we assume that the original circuit $C$ is composed of the computational-basis qubit-state preparations, single- and two-qubit gate operations including waits, and the computational-basis measurements, where the (universal) set of gates are the Pauli $X$ and $Z$, the Hadamard $H$, the phase $S$, the Controlled-NOT CNOT, and the $T$ gates that are explicitly given as follows:
\begin{equation}
    X=\begin{pmatrix}
        0 & 1 \\ 1 & 0
    \end{pmatrix}, \quad Z=\begin{pmatrix}
        1 & 0 \\ 0 & -1 
    \end{pmatrix}, \quad H = \frac{1}{\sqrt{2}}\begin{pmatrix}
        1 & 1 \\ 1 & -1
    \end{pmatrix}, \quad S= \begin{pmatrix}
        1 & 0 \\ 0 & i
    \end{pmatrix}, \quad \text{CNOT} = \begin{pmatrix}
        1 & 0 & 0 & 0 \\ 0 & 1 & 0 & 0 \\ 0 & 0 & 0 & 1 \\ 0 & 0 & 1 & 0
    \end{pmatrix}, \quad T=\begin{pmatrix}
        1 & 0 \\ 0 & e^{\frac{\pi i}{4}}
    \end{pmatrix}.
    \label{eq:universal_gate_set}
\end{equation}
Note that Pauli $X$ and $Z$ gates are auxiliary since they can be generated by the combination of $H$ and $S$ gates, and the $S$ gate as well due to $S=T^2$, but we add them to the list for convenience.
The wait operation is to perform the identity gate
\begin{align}
    I=\begin{pmatrix}
        1 & 0 \\ 0 & 1
    \end{pmatrix}.
\end{align}
We call such circuit constituents (state preparations, gates, waits, and measurements) as locations.

To simulate the original (qubit) circuit with the CV system, 
we use CV state preparations, gates, and measurements described in the previous section.  However, we need to avoid the use of the unphysical state such as the ideal GKP codeword in Eq.~\eqref{eq:ideal_gkp_codeword}.  Based on the subsystem decomposition explained in the previous section, we will introduce a family of physically realizable \quoted{approximate} GKP states as follows.
\begin{definition}[$s$-parameterized GKP state]\label{def:s-parameterized_state}
    For any qubit state $\ket{\psi}=\alpha\ket{0}+\beta\ket{1}$, we define an $s$-parameterized GKP state $\hat{\rho}^s_{\psi}$ as
\begin{align}
    \hat{\rho}^s_{\psi} &\coloneqq \sum_{i}p_i \ket{\rho_i^s[\psi]}\!\bra{\rho_i^s[\psi]}, \label{eq:s-parameterized_mixed}\\
    \ket{\rho_i^s[\psi]} &\coloneqq \int_{-\infty}^{\infty}dz_1\int_{-\infty}^{\infty}dz_2\; f_i^s(z_1,z_2) \hat{V}(z_1,z_2)\ket{\overline{\psi}},\label{eq:s-parameterized_pure}
\end{align}
where $p_i$ ($i=1,2,\ldots$) denotes the probabilities summed up to one, $\ket{\overline{\psi}}=\alpha\ket{\overline{0}}+\beta\ket{\overline{1}}$ denotes the ideal GKP state, and $f_i^s(x,y)$ denotes a function of $(x,y)$ that has support only in $[-s,s)\times [-s,s)$ and satisfies $\braket{\rho_i^s[\psi]|\rho_i^s[\psi]}=1$. Notice that the notation~$\hat{\rho}^s_{\psi}$ leaves implicit the mixing probabilities~$\{ p_i \}$ and the $s$-supported functions~$\{ f_i^s \}$.
\end{definition}
\noindent From this definition and Eqs.~\eqref{eq:shifted_zero} and \eqref{eq:shifted_one}, a state $\ket{\rho_i^s[\psi]}$ has the form
\begin{equation}
    \ket{\rho_i^s[\psi]} = \ket{\psi}_L\otimes \int_{-\frac{c}{2}}^{\frac{c}{2}}dz_1\int_{-\frac{c}{2}}^{\frac{c}{2}}dz_2\; f_i^s(z_1,z_2) \ket{z_1,z_2}_S
\end{equation}
when $s<c/2$.  Thus, it can be regarded as an \quoted{approximate GKP codeword}.  In contrast to the typical definition of approximate GKP states~\cite{Gottesman2001, Matsuura2020}, the state $\hat{\rho}_{\psi}^s$ has \emph{exactly} the same logical content as the ideal one $\ket{\overline{\psi}}$ as long as $s<c/2$.

Having defined the class of physically realizable approximate GKP states that we will consider, we can now list the physical operations to implement our CV quantum circuit:
\begin{itemize}
    \item preparation of a state from a finite set $\mathcal{S}$ including $\hat{\rho}^s_0$, $\hat{\rho}^s_Y$, and $\hat{\rho}^s_{\frac{\pi}{8}}$ with $s$ chosen to be sufficiently smaller than $c/2$, which are defined through Eq.~\eqref{eq:s-parameterized_mixed} with $\ket{Y}=(\ket{0}+i\ket{1})/\sqrt{2}$ and $\ket{\frac{\pi}{8}}=(\ket{0}+e^{\pi i/4}\ket{1})/\sqrt{2}$,
    \item displacement $\hat{V}(v_1, v_2)$ by an amount drawn from a finite set $V$ of possible pairs $(v_1,v_2)$ of values,
    including $(c,0)$ and $(0,c)$,
    \item phase rotation $\hat{R}(\theta)$ by an angle drawn from a finite set $\Theta$ for possible values of $\theta$, including $\pi/2$,
    \item SUM gate $\exp(-i\hat{q}_1\hat{p}_2)$,
    \item $q$- and $p$-homodyne detection,
    \item waiting (or using delay lines).
\end{itemize}
Later, it will be clear that to avoid a logical error, the parameter $s$ in the prepared state should be smaller than $c/(2N_{\text{2-mode}})$ with $N_{\text{2-mode}}$ denoting the maximum number of two-mode gates that a prepared state goes through before the measurement.  Note that if we adopt Knill-type error correction \cite{Knill2005}, which we will do, $N_{\text{2-mode}}$ is a constant number.

\subsection{Energy-constrained diamond norm} \label{sec:energy_constrained_diamond_norm}
In this section, we introduce a distance measure between CV quantum channels that we use in the next section.  The distance measure is necessary to define the noise strength for a given noise model; i.e., small noise strength implies closeness to the identity channel.
However, there are subtleties in the distance measure for channels in infinite-dimensional systems, as discussed in the main text with an example of the phase-rotation channel  ${\cal R}[\theta]$.
Our physical intuition that ${\cal R}[\theta]$ is close to the identity map when $\theta>0$ is very close to $0$ comes from the fact that quantum states we usually generate and treat in the lab are not very susceptible to small phase rotation, whereas the diamond-norm distance tells us it is far from the identity map since a coherent state with the amplitude $\alpha$ distinguishes these two channels better and better as $|\alpha|\to\infty$.
This raises the following two issues to discuss.  On one hand, if the quantum state during a quantum computation is \quoted{reasonably good}, then the noise such as an infinitesimal rotation should be well approximated by the identity map.  On the other hand, for this to be true, we need to take care that the quantum state during a computation is ensured to be kept \quoted{reasonably good}.  We postpone the discussion of the second issue later and focus on the first issue here.

In this paper, quantum states with the energy constraint are regarded as the \quoted{reasonably good} quantum states, which seems reasonable in terms of the actual experiment.  More precisely, we regard the following set of quantum states as those realizable in the experiment and thus in the CV quantum circuit.
\begin{definition}[The set of energy-constrained quantum states]
    Given an energy bound $E>0$ on a system (or equivalently, a mode) $Q$,  the set $\mathfrak{S}_{E}({\cal H}_Q)$ of energy-constrained density operators on the system $Q$ is defined as
    \begin{equation}
        \mathfrak{S}_{E}({\cal H}_Q) = \{\hat{\rho} \in {\cal B}({\cal H}_Q): \hat{\rho}\geq 0, \Tr (\hat{\rho})=1, \Tr (\hat{n}\hat{\rho})\leq E\}.
    \end{equation}
    The set of all density operators without energy constraint is simply denoted by $\mathfrak{S}({\cal H}_Q)$.
\end{definition}
\noindent In the definition above, as well as in the following discussions in this section, we identify the number operator~$\hat{n}$ as the Hamiltonian of the system, with the optical system in mind, but the same argument holds for any Hamiltonian operator $\hat{H}$ that satisfies the Gibbs hypothesis, i.e., $\Tr [\exp(-\beta\hat{H})]<\infty$ for all $\beta>0$~\cite{Winter2017, Shirokov2019}.  Note that the number operator $\hat{n}$ clearly satisfies this condition.
A particularly nice property of the set $\mathfrak{S}_{E}({\cal H}_Q)$ for finite $E$ is that it is compact with respect to the trace norm~\cite{Holevo2003}, while $\mathfrak{S}({\cal H}_Q)$ is not.
For this set of energy-constrained quantum states, we can define the energy-constrained diamond norm considered in Refs.~\cite{Holevo2003, Shirokov2008, Winter2016, Winter2017, Shirokov2018, Shirokov2018adaptation, Shirokov2019}.
\begin{definition}[Energy-constrained diamond norm \cite{Shirokov2018}]
    Let $\Phi$ be a Hermetian-preserving linear map (not necessarily completely positive or even positive) acting on operators on ${\cal H}_Q$. For $E>0$, the energy-constrained diamond norm $\|\Phi\|_{\diamond}^E$ of $\Phi$ is defined as
    \begin{equation}
        \|\Phi\|_{\diamond}^E\coloneqq \sup_{\hat{\rho}_{RQ}\in\mathfrak{S}({\cal H}_{RQ}):\hat{\rho}_Q\in\mathfrak{S}_{E}({\cal H}_Q)}\|\mathrm{Id}_{R}\otimes \Phi(\hat{\rho}_{RQ})\|_1, \label{eq:def_energy_constrained_diamond}
    \end{equation}
    where $\mathrm{Id}_R$ denotes the identity map of operators on ${\cal H}_R\cong{\cal H}_Q$, $\hat{\rho}_Q\coloneqq \Tr _R[\hat{\rho}_{RQ}]$, and $\|\hat{T}\|_1\coloneqq \Tr \sqrt{\hat{T}^{\dagger}\hat{T}}$.  We simply let $\|\Phi\|_{\diamond}$ denote the diamond norm without energy constraint, i.e.,
    \begin{equation}
        \|\Phi\|_{\diamond} \coloneqq \sup_{\hat{\rho}_{RQ}\in\mathfrak{S}({\cal H}_{RQ})}\|\mathrm{Id}_{R}\otimes \Phi(\hat{\rho}_{RQ})\|_1=\sup_{\hat{T}_{RQ}\in {\cal B}({\cal H}_{RQ}):\|\hat{T}_{RQ}\|_1\leq 1} \|\mathrm{Id}_R\otimes \Phi (\hat{T}_{RQ})\|_1. \label{eq:def_conventional_diamond}
    \end{equation}
    For a composite system $Q_1Q_2$ with the energy constraint $E_1$ on $Q_1$ and $E_2$ on $Q_2$, we define $\|\Phi_{Q_1Q_2}\|^{E_1,E_2}_{\diamond}$ as 
    \begin{equation}
        \|\Phi_{Q_1Q_2}\|^{E_1,E_2}_{\diamond}\coloneqq \sup_{\hat{\rho}_{RQ_1Q_2}\in\mathfrak{S}({\cal H}_{RQ_1Q_2}):\hat{\rho}_{Q_1}\in\mathfrak{S}_{E_1}({\cal H}_{Q_1}), \hat{\rho}_{Q_2}\in\mathfrak{S}_{E_2}({\cal H}_{Q_2})}\|\mathrm{Id}_{R}\otimes \Phi_{Q_1Q_2}(\hat{\rho}_{RQ_1Q_2})\|_1,
    \end{equation}
    where ${\cal H}_R\cong {\cal H}_{Q_1Q_2} $.
\end{definition}
\noindent This norm defines a distance between quantum channels called the energy-constrained diamond-norm distance.
Many of the favorable properties of the diamond norm in the finite-dimensional case are recovered in the energy-constrained diamond norm \cite{Winter2017, Shirokov2018}.  We cite some of the properties that will be used later.

\begin{itemize}
    \item (Monotonicity under energy increase) For $0<E_1\leq E_2$, the energy-constrained diamond norm satisfies 
    \begin{equation}\|\Phi\|_{\diamond}^{E_1}\leq\|\Phi\|_{\diamond}^{E_2} . \label{eq:monotonicity}
    \end{equation}   
    In particular, $\|\Phi\|_{\diamond}^{E}\leq\|\Phi\|_{\diamond}$ holds for any $E>0$.
    \item (Submultiplicativity under composition) Let $\Phi$ be a Hermetian-preserving linear map and $\Psi:\mathfrak{S}_{E_1}({\cal H}_Q)\to \mathfrak{S}_{E_2}({\cal H}_Q)$ be a completely positive map.  Then, we have
    \begin{equation}
    \|\Phi\circ\Psi\|_{\diamond}^{E_1}\leq \|\Phi\|_{\diamond}^{E_2}\|\Psi\|^{E_1}_{\diamond}. \label{eq:submultiplicativity_1}
    \end{equation}
    Furthermore, for another Hermitian-preserving linear map $\Psi':\mathfrak{S}_{E_1}({\cal H}_Q)\to \mathcal{B}({\cal H}_{Q})$, we have 
    \begin{equation}
        \|\Phi\circ\Psi'\|_{\diamond}^{E_1} \leq \|\Phi\|_{\diamond}\|\Psi'\|_{\diamond}^{E_1}. \label{eq:submultiplicativity_2}
    \end{equation}
    \item (Supermultiplicativity under tensor product) For the Hermitian-preserving linear maps $\Phi_{Q_1}$ and $\Phi_{Q_2}$ of operators on ${\cal H}_{Q_1}$ and ${\cal H}_{Q_1}$, respectively, we have \begin{equation}\|\Phi_{Q_1}\otimes\Phi_{Q_2}\|_{\diamond}^{E_1, E_2} \geq \|\Phi_{Q_1}\|^{E_1}_{\diamond}\|\Phi_{Q_2}\|^{E_2}_{\diamond}.
    \end{equation}
    The equality holds when at least one of $\Phi_{Q_1}$ and $\Phi_{Q_2}$ is completely positive.
    This property can straightforwardly be derived from the facts that $\|\hat{A}\otimes\hat{B}\|_1=\|\hat{A}\|_1\|\hat{B}\|_1$ holds and that there exists a state $\hat{\rho}_{RQ_1Q_2}$ on the composite system $RQ_1Q_2$ such that $\hat{\rho}_{Q_1Q_2}$ is an entangled state and $\hat{\rho}_{Q_i}\in\mathfrak{S}_{E_i}({\cal H}_{Q_i})$ for $i=1,2$.
    \item (Achievability) From the compactness of $\mathfrak{S}_E({\cal H}_{Q})$, there always exists $\hat{\rho}_{QR}$ with $\hat{\rho}_Q\in\mathfrak{S}_E({\cal H}_{Q})$ such that $\|\Phi_Q\|_{\diamond}^E=\|\mathrm{Id}_R \otimes \Phi_Q (\hat{\rho}_{RQ})\|_1$.
\end{itemize}
One shortcoming of this norm is that one needs extra care when one applies it to composite maps, which can already be seen in the preconditions above.
It is particularly because $(\hat{\rho}_A-\hat{\sigma}_A)_\pm\notin \mathfrak{S}_E({\cal H}_A)$ even though $\hat{\rho}_A,\hat{\sigma}_A\in\mathfrak{S}_E({\cal H}_A)$, where $(\hat{X})_\pm \coloneqq (\hat{X} \pm |\hat{X}|)/\Tr(\hat{X} \pm |\hat{X}|)$ is the (normalized) density operator corresponding to the positive part (and the negative part, respectively) of $\hat{X}$.
However, by appropriately enlarging the energy-constrained set, we have a bound on the energy-constrained diamond norm of composite maps with respective norms, as we show in the following proposition by generalizing the arguments used for Lemma 1 and Theorem 1 in Ref.~\cite{Shirokov2020}.

\begin{proposition}[Norm of a map acting on the difference of energy-constrained density operators] \label{prop:difference_energy_constrained}
    Given $E>0$, let $\hat{\rho}_{AR}$ and $\hat{\sigma}_{AR}$ be density operators with $\hat{\rho}_A,\hat{\sigma}_{A}\in\mathfrak{S}_E({\cal H}_A)$ and $\|\hat{\rho}_{AR}-\hat{\sigma}_{AR}\|_1 \leq 2\epsilon$ with $0<\epsilon<1$.  Then, for any Hermetian-preserving linear map $\Phi_A$ on the system $A$, we have
    \begin{equation}
        \|\Phi_A\otimes \mathrm{Id}_R(\hat{\rho}_{AR} - \hat{\sigma}_{AR})\|_1 \leq 10\epsilon\|\Phi_A\|_{\diamond}^{E/\epsilon^2}.
    \end{equation}
\end{proposition}

To prove Proposition~\ref{prop:difference_energy_constrained}, we use the following auxiliary lemma, which is basically a reformulation of what is proved in Lemma 1 in Ref.~\cite{Shirokov2020}.
We here provide the proof of this auxiliary lemma for completeness of our analysis.

\begin{lemma}[Finite-dimensional approximation of energy-constrained density operators: Footnote~5 in the proof of Lemma~1 in Ref.~\cite{Shirokov2020}, Lemma~3 of Ref.~\cite{Shirokov2019}] \label{lem:difference_energy_constrained}
    Given $E>0$, let $\hat{\rho}_{AR}$ be a density operator with $\hat{\rho}_A\in\mathfrak{S}_E({\cal H}_A)$.  Then, for any purification $\ket{\phi}_{ARS}$ of $\hat{\rho}_{AR}$, there exist a pure state $\ket{\tilde{\phi}}_{ARS}$ and an error parameter $t\in(0,1)$ such that
    \begin{align}
    &\mathrm{supp}(\tilde{\phi}_A)\subset\mathrm{span}\{\ket{0}_A,\ldots,\ket{\lfloor E/t^2\rfloor-1}_A\},\\
    &\frac{1}{2}\|\ket{\phi}\!\bra{\phi}-\ket{\tilde{\phi}}\!\bra{\tilde{\phi}}\|_1 \leq t,\\
    \label{eq:orthogonal_decomp_lemma}
    &\ket{\phi}\!\bra{\phi}-\ket{\tilde{\phi}}\!\bra{\tilde{\phi}} = t\ket{\gamma_+}\!\bra{\gamma_+} - t\ket{\gamma_-}\!\bra{\gamma_-},\\
    & \braket{\gamma_{\pm}|\gamma_{\pm}}=1, \quad \braket{\gamma_+|\gamma_-}=0, \\
    \label{eq:energy_bound_gamma_lemma}
    &\Tr[\hat{n}_A(\Tr_{RS}[\ket{\gamma_\pm}\!\bra{\gamma_\pm}_{ARS}])]\leq E/t^2,
    \end{align}
    where $\tilde{\phi}_A\coloneqq\Tr_{RS}[\ket{\tilde{\phi}}\!\bra{\tilde{\phi}}_{ARS}]$, $\mathrm{supp}(\tilde{\phi}_A)$ is the support of $\tilde{\phi}_A$, $\{\ket{n}_A\}_{n=0}^{\infty}$ is the Fock basis of $\mathcal{H}_A$, and $\hat{n}_A=\sum_{n=0}^\infty n \ket{n}\bra{n}_A$ is the number operator (i.e. the Hamiltonian of the system $A$).
\end{lemma}

\begin{proof}
    Let $\ket{\phi}_{ARS}$ be a purification of $\hat{\rho}_{AR}$.  We have, without loss of generality,
    \begin{equation}
        \ket{\phi}_{ARS} = \sum_{j,k=0}^{\infty} \sqrt{p_{jk}}e^{i\theta_{jk}} \ket{j}_A\ket{\alpha_k}_{RS},
    \end{equation}
    where $\ket{j}_A$ denotes the Fock state, $\{\ket{\alpha_k}\}_{k=0}^{\infty}$ is an orthonormal system in ${\cal H}_{R}\otimes{\cal H}_{S}$, and $\{p_{jk}\}_{j,k=0}^{\infty}$ are probabilities satisfying $\sum_{j,k=0}^{\infty}p_{jk}=1$.  
    Let $d > E$ be any integer and define $\delta_d$ as $\delta_d\coloneqq \sum_{j\geq d}\tilde{p}_j$, where $\tilde{p}_j\coloneqq \sum_{k=0}^{\infty}p_{jk}$.  We further define $\ket{\tilde{\phi}}_{ARS}$ as
    \begin{equation}
        \ket{\tilde{\phi}}_{ARS} \coloneqq (1-\delta_d)^{-1/2}\sum_{j=0}^{d-1}\sum_{k=0}^{\infty} \sqrt{p_{jk}}e^{i\theta_{jk}}\ket{j}_A\ket{\alpha_k}_{RS}.
    \end{equation}
    From the energy constraint, we have
    \begin{equation}
        \delta_d d \leq \sum_{j\geq d}\sum_{k=0}^{\infty} p_{jk} \bra{j}_A\bra{\alpha_k}_{RS}\hat{n}_A\otimes\hat{I}_{RS}\ket{j}_A\ket{\alpha_k}_{RS} \leq E,
    \end{equation}
    and thus we have $\delta_d \leq E/d <1$.  Since
    $\braket{\phi|\tilde{\phi}}=(1-\delta_d)^{1/2}$, we have
\begin{align}
    \|\ket{\phi}\!\bra{\phi}-\ket{\tilde{\phi}}\!\bra{\tilde{\phi}}\|_1 = 2\sqrt{1-|\braket{\phi|\tilde{\phi}}|^2} = 2\sqrt{\delta_d}, \label{eq:delta_d-approx}
\end{align}
where we used the well-known relation between trace distance and fidelity for pure states in the first equality.
    From the diagonalization of the $2\times 2$ matrix $\ket{\phi}\!\bra{\phi}_{ARS}-\ket{\tilde{\phi}}\!\bra{\tilde{\phi}}_{ARS}$, we have
    \begin{equation}
        \ket{\phi}\!\bra{\phi}_{ARS}-\ket{\tilde{\phi}}\!\bra{\tilde{\phi}}_{ARS} = \sqrt{\delta_d}\ket{\gamma_+}\!\bra{\gamma_+}_{ARS} - \sqrt{\delta_d}\ket{\gamma_-}\!\bra{\gamma_-}_{ARS}, \label{eq:orthogonal_decomp}
    \end{equation}
    where $\ket{\gamma_{\pm}}$ are normalized eigenvectors of positive and negative parts, respectively, with $\braket{\gamma_+|\gamma_-}=0$.  Written explicitly, these eigenvectors are
    \begin{equation}
        \ket{\gamma_{\pm}}_{ARS} = p_{\pm}\ket{\phi}_{ARS} + q_{\pm}\ket{\tilde{\phi}}_{ARS},
    \end{equation}
    with $p_{\pm}=\sqrt{(1\pm\sqrt{\delta_d})/(2\delta_d)}$ and $q_{\pm}=-\sqrt{(1\mp\sqrt{\delta_d})/(2\delta_d)}$.
    
    Now, we would like to bound the average energy of $\ket{\gamma_{\pm}}_{ARS}$ in the system $A$.
    We have
    \begin{align}
        \bra{\gamma_{\pm}} \hat{n}_A\otimes \hat{I}_{RS} \ket{\gamma_{\pm}} &= p_{\pm}^2 \bra{\phi}\hat{n}_A\otimes \hat{I}_{RS}\ket{\phi} + q_{\pm}^2 \bra{\tilde{\phi}}\hat{n}_A\otimes \hat{I}_{RS}\ket{\tilde{\phi}} + 2p_{\pm}q_{\pm} \Re{(\bra{\phi}\hat{n}_A\otimes \hat{I}_{RS}\ket{\tilde{\phi}})}\\
        &= p_{\pm}^2 \sum_{j=0}^{\infty}\tilde{p}_j j + q_{\pm}^2 (1-\delta_d)^{-1}\sum_{k=0}^{d-1} \tilde{p}_k k + 2p_{\pm}q_{\pm}(1-\delta_d)^{-1/2} \sum_{l=0}^{d-1}\tilde{p}_l l \\
        & = [p_{\pm}^2 + 2p_{\pm}q_{\pm}(1-\delta_d)^{-1/2}+q_{\pm}^2(1-\delta_d)^{-1}]T + p_{\pm}^2U, \label{eq:terms_T_U}
    \end{align}
    where $T\coloneqq \sum_{j=0}^{d-1}\tilde{p}_j j $ and $U\coloneqq \sum_{k\geq d} \tilde{p}_k k$.  
    Using the explicit form of $p_{\pm}$ and $q_{\pm}$, we have
    \begin{align}
        \eqref{eq:terms_T_U} &= [p_{\pm} + q_{\pm}(1-\delta_d)^{-1/2}]^2\, T + p_{\pm}^2 U \\
        & = \frac{1}{2\delta_d}\left[\left(\sqrt{1\pm\sqrt{\delta_d}} - \frac{1}{\sqrt{1\pm\sqrt{\delta_d}}}\right)^2 T + (1\pm\sqrt{\delta_d})U \right] \\
        &= \frac{1}{2\delta_d}\left[\frac{\delta_d}{1\pm\sqrt{\delta_d}}T + (1\pm\sqrt{\delta_d})U\right],
    \end{align}
    Applying the inequalities $T=\sum_{k=0}^{\infty}\tilde{p}_k\sum_{j=0}^{d-1}\tilde{p}_j j \leq \sum_{j=0}^{d-1}\sum_{k=0}^{\infty}\tilde{p}_j\tilde{p}_k k = (1-\delta_d)E$ and $U\leq E$ to the above, we have
    \begin{align}
        \bra{\gamma_{\pm}} \hat{n}_A\otimes \hat{I}_{RS} \ket{\gamma_{\pm}}
    &\leq
        \frac{1}{2\delta_d} [\delta_d(1\mp\sqrt{\delta_d})E + (1\pm\sqrt{\delta_d})E]
    \\
    &\leq
        \frac{1}{2\delta_d} [(1\mp\sqrt{\delta_d})E + (1\pm\sqrt{\delta_d})E] 
    = \frac{E}{\delta_d}, \label{eq:energy_bound_gamma}
    \end{align}
    where we used $\delta_d \leq 1$.
    Let us replace $\delta_d\mapsto t^2$ with $t\in(0,1)$.  Then, from Eq.~\eqref{eq:energy_bound_gamma} and $t^2 d\leq E$, we have
    \begin{align}
    \Tr _{RS}[\ket{\gamma_{\pm}}\!\bra{\gamma_{\pm}}_{ARS}] &\in
    \mathfrak{S}_{E/t^2}({\cal H}_A),
    \label{eq:conclusion_1}
    \\
    \ket{\tilde{\phi}}_{ARS}
    &\in
    (\hat{P}^{\lfloor E/t^2 \rfloor}_A\otimes \hat{I}_{RS}){\cal H}_{ARS},
    \label{eq:conclusion_2}
    \intertext{where} \hat{P}^d
    &\coloneqq
    \sum_{n=0}^{d-1}\ket{n}\!\bra{n} \qquad \text{for integer $d$}.  
    \end{align}
    Therefore, from Eqs.~\eqref{eq:delta_d-approx},~\eqref{eq:orthogonal_decomp},~\eqref{eq:conclusion_1}, and~\eqref{eq:conclusion_2}, we obtain the conclusion.
\end{proof}

Using Lemma~\ref{lem:difference_energy_constrained}, we prove Proposition~\ref{prop:difference_energy_constrained} as follows.

\begin{proof}[Proof of Proposition~\ref{prop:difference_energy_constrained}]
    Let $\ket{\phi}_{ARS}$ be a purification of $\hat{\rho}_{AR}$.  
    Using Lemma~\ref{lem:difference_energy_constrained}, we obtain a $t$-approximation $\ket{\tilde{\phi}}_{ARS}$ of $\ket{\phi}_{ARS}$.
    Then, for any Hermitian-preserving linear map $\Phi_A$, we have
    \begin{align}
        \|\Phi_A\otimes \mathrm{Id}_{RS}(\ket{\phi}\!\bra{\phi}_{ARS} - \ket{\tilde{\phi}}\!\bra{\tilde{\phi}}_{ARS})\|_1 &= t\|\Phi_A\otimes \mathrm{Id}_{RS}(\ket{\gamma_+}\!\bra{\gamma_+}_{ARS} - \ket{\gamma_-}\!\bra{\gamma_-}_{ARS})\|_1 \\
        & \leq t \|\Phi_A\otimes \mathrm{Id}_{RS}(\ket{\gamma_+}\!\bra{\gamma_+}_{ARS})\|_1 + t\|\Phi_A\otimes \mathrm{Id}_{RS}(\ket{\gamma_-}\!\bra{\gamma_-}_{ARS})\|_1 \\
        & \leq 2t \|\Phi_A \|_{\diamond}^{E/t^2}, \label{eq:renormalized_energy}
    \end{align}
    where we used Eq.~\eqref{eq:orthogonal_decomp_lemma} of  Lemma~\ref{lem:difference_energy_constrained} in the first equality, the triangle inequality in the first inequality, and Eq.~\eqref{eq:energy_bound_gamma_lemma} of  Lemma~\ref{lem:difference_energy_constrained} with Eq.~\eqref{eq:def_energy_constrained_diamond} in the last inequality.
    
    An analogous derivation can be carried out for the purification $\ket{\psi}_{ARS}$ of $\hat{\sigma}_{AR}$.
    With a $t$-approximation $\ket{\tilde{\psi}}_{ARS}$ of $\ket{\psi}_{ARS}$,
    we obtain from  Lemma~\ref{lem:difference_energy_constrained}:
    \begin{align}
        \|\Phi_A\otimes \mathrm{Id}_{RS}(\ket{\psi}\!\bra{\psi}_{ARS} - \ket{\tilde{\psi}}\!\bra{\tilde{\psi}}_{ARS})\|_1 \leq 2t \|\Phi_A \|_{\diamond}^{E/t^2},
        \label{eq:renormalized_energy_2}
    \end{align}
    as in Eq.~\eqref{eq:renormalized_energy}.

    In the following, we will analyze an upper bound of $\|\Phi_A\otimes \mathrm{Id}_R(\hat{\rho}_{AR} - \hat{\sigma}_{AR})\|_1$.
    For this analysis, we define $\hat{\rho}'_{AR}\coloneqq \Tr_S[\ket{\tilde{\phi}}\!\bra{\tilde{\phi}}_{ARS}]$ and $\hat{\sigma}'_{AR}\coloneqq \Tr_S[\ket{\tilde{\psi}}\!\bra{\tilde{\psi}}_{ARS}]$.
    Importantly, the operator $\Tr_R[\hat{\rho}'_{AR}-\hat{\sigma}'_{AR}]$ has a finite support in $\mathrm{span}\{\ket{0}_A,\ldots,\ket{\lfloor E/t^2\rfloor-1}_A\}$ on $A$, which makes our analysis possible; to analyze bounds for this operator, we introduce a set of (Hermitian and unit-norm) operators on $\mathcal{H}_A\otimes\mathcal{H}_R$ with the finite support on $A$, i.e.,
    \begin{align}
        \mathcal{T}_{E/t^2}\coloneqq\left\{\hat{X}_{AR} \in \mathcal{B}\left((\hat{P}^{\lfloor E/t^2\rfloor}_A\otimes \hat{I}_{R}){\cal H}_{AR}\right) : \hat{X}_{AR}=\hat{X}_{AR}^{\dagger}, \|\hat{X}_{AR}\|_1 = 1\right\},
    \end{align}
    where $\mathcal{B}\left((\hat{P}^{\lfloor E/t^2\rfloor}_A\otimes \hat{I}_{R}){\cal H}_{AR}\right)$ is the set of bounded operators on $(\hat{P}^{\lfloor E/t^2\rfloor}_A\otimes \hat{I}_{R}){\cal H}_{AR}$.
    Then, we have the following (with further explanation given below):
    \begin{align}
        &\mathrel{\phantom{=}}\|\Phi_A\otimes \mathrm{Id}_R(\hat{\rho}_{AR} - \hat{\sigma}_{AR})\|_1\\
        &\leq \|\Phi_A\otimes \mathrm{Id}_{R}(\hat{\rho}_{AR} - \hat{\rho}'_{AR})\|_1 + \|\Phi_A\otimes \mathrm{Id}_{R}(\hat{\sigma}_{AR} - \hat{\sigma}'_{AR})\|_1 +\|\Phi_A\otimes \mathrm{Id}_{R}(\hat{\rho}'_{AR} - \hat{\sigma}'_{AR})\|_1
        \\ 
        \begin{split}
        &\leq\|\Phi_A\otimes \mathrm{Id}_{RS}(\ket{\phi}\!\bra{\phi} - \ket{\tilde{\phi}}\!\bra{\tilde{\phi}})\|_1 + \|\Phi_A\otimes \mathrm{Id}_{R}(\ket{\psi}\!\bra{\psi} - \ket{\tilde{\psi}}\!\bra{\tilde{\psi}})\|_1
        +\sup_{\hat{X}_{AR} \in \mathcal{T}_{E/t^2}}\|\Phi_{A}\otimes \mathrm{Id}_{R}(\|\hat{\rho}'_{AR} - \hat{\sigma}'_{AR}\|_1\hat{X}_{AR})\|_1 
         \end{split}
        \\
        &\leq 4t\|\Phi_A\|_{\diamond}^{E/t^2}+ \sup_{\hat{X}_{AR} \in \mathcal{T}_{E/t^2}}\|\Phi_{A}\otimes \mathrm{Id}_{R}(\hat{X}_{AR})\|_1 \|\hat{\rho}'_{AR} - \hat{\sigma}'_{AR}\|_1, \label{eq:diamond_norm_like}
    \end{align}
    where the triangle inequality is used in the first inequality, the monotonicity of the trace distance under tracing out and the fact $\hat{\rho}'_{AR},\hat{\sigma}'_{AR}\in\mathcal{T}_{E/t^2}$ are used in the second inequality, and Eqs.~\eqref{eq:renormalized_energy} and~\eqref{eq:renormalized_energy_2} are used in the last inequality.
    Since $\Phi_A$ is a Hermetian-preserving map, due to the convexity of the norm, we can replace maximization over $\hat{X}_{AR}$ in Eq.~\eqref{eq:diamond_norm_like} with that over rank-1 projections $\{\ket{\tau}\!\bra{\tau}_{AR}\}$, where $\ket{\tau}_{AR}\in (\hat{P}^{\lfloor E/t^2\rfloor}_A\otimes \hat{I}_{R}){\cal H}_{AR}$.
    In other words, the optimization over the Hermitian operators $\hat{X}_{AR} \in \mathcal{T}_{E/t^2}$ (not necessarily quantum states) can be replaced with that over the (pure) quantum states $\ket{\tau}_{AR}\in (\hat{P}^{\lfloor E/t^2\rfloor}_A\otimes \hat{I}_{R}){\cal H}_{AR}$ with the finite support on $A$, for which the energy of $A$ is well-defined (unlike $\hat{X}_{AR}$) and bounded by
    \begin{align}
        \bra{\tau}\hat{n}_A\otimes\hat{I}_{R}\ket{\tau} \leq E/t^2.
        \label{eq:energy_bound_tau}
    \end{align}
    Therefore, we obtain
    \begin{align}
        \sup_{\hat{X}_{AR} \in \mathcal{T}_{E/t^2}}\|\Phi_{A}\otimes \mathrm{Id}_{R}(\hat{X}_{AR})\|_1 
        &\leq \sup_{\ket{\tau}\in (\hat{P}^{\lfloor E/t^2\rfloor}_A \otimes \hat{I}_{R}){\cal H}_{AR}:\braket{\tau|\tau}=1}\|\Phi_{A}\otimes \mathrm{Id}_{R}(\ket{\tau}\!\bra{\tau}_{AR})\|_1 \\
        &\leq \|\Phi_A\|_{\diamond}^{E/t^2},\label{eq:finite-dim_diamond}
    \end{align}
    where the first inequality is this replacement, and the second inequality follows from the energy bound~\eqref{eq:energy_bound_tau} and the definition~\eqref{eq:def_energy_constrained_diamond} of the energy-constrained diamond norm.
    We also have
    \begin{equation}
        \|\hat{\rho}'_{AR}- \hat{\sigma}'_{AR}\|_1 \leq \|\hat{\rho}_{AR} - \hat{\sigma}_{AR}\|_1 + \|\hat{\rho}_{AR} - \hat{\rho}'_{AR}\|_1 + \|\hat{\sigma}_{AR}- \hat{\sigma}'_{AR}\|_1\leq 2\epsilon + 4t, \label{eq:distance_approx_states}
    \end{equation}
    where the first term of the right-most expression comes from the assumption of this proposition, and the second term comes from Eq.~\eqref{eq:delta_d-approx} and the monotonicity of the trace distance under tracing out applied to both $\|\hat{\rho}_{AR} - \hat{\rho}'_{AR}\|_1$ and $\|\hat{\sigma}_{AR}- \hat{\sigma}'_{AR}\|_1$.
    From Eqs.~\eqref{eq:diamond_norm_like}, \eqref{eq:finite-dim_diamond}, and \eqref{eq:distance_approx_states}, we prove the statement of the proposition by substituting $t=\epsilon$, so that the overall coefficient is $4t + 2 \epsilon + 4t \bigr\rvert_{t\to\epsilon} = 10 \epsilon$.
\end{proof}
As an immediate corollary of the above theorem, we have the following generalization of the submultiplicativity property cited above.
\begin{corollary} \label{cor:composition}
    Let $\Phi_A$ be a Hermetian-preserving map and $\Psi_A,\tilde{\Psi}_A:\mathfrak{S}_{E_1}({\cal H}_A)\to \mathfrak{S}_{E_2}({\cal H}_A)$ be two CPTP maps such that $\|\Psi_A - \tilde{\Psi}_A\|_{\diamond}^{E_1}\leq 2\epsilon$.  Then, we have 
    \begin{equation}
        \|\Phi_A \circ (\Psi_A - \tilde{\Psi}_A)\|_{\diamond}^{E_1} \leq 10\epsilon\|\Phi_A\|_{\diamond}^{E_2/\epsilon^2}.
    \end{equation}
\end{corollary}
\begin{proof}
    From the achievability mentioned above, we have a state $\hat{\tau}_{AR}$ with $\hat{\tau}_A\in\mathfrak{S}_{E_1}({\cal H}_A)$ such that $\|\Phi_A \circ (\Psi_A - \tilde{\Psi}_A)\|_{\diamond}^{E_1}=\|[\Phi_A \circ (\Psi_A - \tilde{\Psi}_A)]\otimes \mathrm{Id}_R(\hat{\tau}_{AR})\|_1$.  Then, we can apply Proposition~\ref{prop:difference_energy_constrained} by substituting $\hat{\rho}_{AR}=\Psi_A\otimes \mathrm{Id}_R(\hat{\tau}_{AR})$ and $\hat{\sigma}_{AR}=\tilde{\Psi}_A\otimes \mathrm{Id}_R(\hat{\tau}_{AR})$ and using $\|\hat{\rho}_{AR} - \hat{\sigma}_{AR}\|_1 = \|(\Psi_A - \tilde{\Psi}_A)\otimes \mathrm{Id}_R(\hat{\tau}_{AR})\|_1 \leq \|\Psi_A - \tilde{\Psi}_A\|_{\diamond}^{E_1} \leq 2\epsilon$.
\end{proof}

\subsection{Noise model} \label{sec:ftqc_noise_model}
In the previous sections, we described the physical operations that implement a CV quantum circuit and introduced a distance measure between noise channels to evaluate the noise strength.  In this section, we consider a noise model for these physical operations, which then leads to the noise model for each location of the CV quantum computation.
The noise model in this paper is an independent Markovian noise model, meaning that the noise at a given location can be described as a quantum channel that is independent of any other location in the quantum circuit.  This noise model is general as long as the computational modes do not strongly interact with each other (so that there is no correlated noise in space) and do not repeatedly interact with the same environmental mode (so that there is no correlated noise in time).  These assumptions may be reasonable in propagating optical systems where the interaction between the modes is weak and they move fast.

Before moving on to the formal definition of the noise model and its noise strength, we define the following channel.
\begin{definition}[$s$-parameterized noise channel]\label{def:def_N_s}
    An $s$-parameterized noise channel ${\cal N}^s$ is defined as a noise channel whose Kraus operators are linear combinations of elements in $\{\hat{V}(z_1,z_2):(z_1,z_2)\in [-s,s)\times [-s,s)\}$ for a single-mode gate, i.e.,
\begin{equation}
    {\cal N}^s(\hat{\rho}) = \int_{\mathbb{R}^4} d^2z\, d^2z'\; {\cal E}_s(z_1,z_2,z'_1,z'_2)\hat{V}(z_1,z_2) \hat{\rho} \hat{V}(z'_1,z'_2)^{\dagger},
    \label{eq:def_N_s}
\end{equation}
where the function ${\cal E}_s(z_1,z_2,z'_1,z'_2)$ is supported only on the Cartesian product $\prod_{i=1}^4[-s,s)$.
Likewise, for a two-mode gate, the Kraus operators are linear combinations of elements in $\{\hat{V}(z_1,z_2)\otimes \hat{V}(z'_1,z'_2):(z_1,z_2)\in [-s,s)\times [-s,s), (z'_1,z'_2)\in [-s,s)\times [-s,s)\}$. 
\end{definition}
Now, we define the noise models against which we prove fault tolerance in this paper.

\begin{definition}[$(s,\epsilon)$-independent Markovian noise model for preparation]\label{def:s_eps_Markovian_prep}
    Given a constant  $E_{\rm prep}>0$, consider a collection ${\cal C}'$ of physical systems that comprise a CV quantum computer, and regard other physical systems as environments.
    A noisy physical state preparation for the GKP logical state $\ket{\overline{\psi}}\!\bra{\overline{\psi}}$ is said to obey the $(s,\epsilon)$-independent Markovian noise model if it prepares a noisy state $\hat{\rho}_{\psi}^{\rm noisy}\in\mathfrak{S}_{E_{\rm prep}}({\cal H})$, 
    independent of other physical systems in ${\cal C}'$, such that there exists a state $\hat{\rho}_{\psi}^s$ defined in Def.~\ref{def:s-parameterized_state} that satisfies $\hat{\rho}_{\psi}^s\in\mathfrak{S}_{E_{\rm prep}}({\cal H})$ and
    \begin{equation}
        \frac{1}{2}\|\hat{\rho}_{\psi}^{\rm noisy} - \hat{\rho}_{\psi}^s\|_1\leq \epsilon.
        \label{eq:noise_model_for_preparation}
    \end{equation}
    Thus, as a state preparation channel ${\cal O}[\hat{\rho}]:1\mapsto\hat{\rho}$, it should satisfy
    \begin{equation}
        \frac{1}{2}\|{\cal O}[\hat{\rho}_{\psi}^{\rm noisy}] - {\cal O}[\hat{\rho}_{\psi}^s]\|_\diamond\leq \epsilon.
    \end{equation}

\end{definition}
\begin{definition}[$(E,s,\epsilon)$-independent Markovian noise model for gate] \label{def:E_s_eps_Markovian_gate}
    Given a fixed, positive, monotonically increasing, locally bounded function $g_{\rm sup}$, consider a collection ${\cal C}'$ of physical systems that comprise a CV quantum computer, and regard other physical systems as environments.
    A noisy physical gate operation for a target physical gate ${\cal U}(\hat{\rho}) = \hat{U}\hat{\rho}\hat{U}^{\dagger}$ is said to obey the $(E,s,\epsilon)$-independent Markovian noise model if it implements the CPTP map ${\cal U}^{\rm noisy}:\mathfrak{S}_{E}({\cal H}_Q)\to\mathfrak{S}_{g_{\rm sup}(E)}({\cal H}_Q)$, independent of other physical systems in ${\cal C}'$, such that there exists a CPTP map ${\cal N}^s$ defined in Def.~\ref{def:def_N_s} that satisfies ${\cal N}^s\circ {\cal U}(\hat{\rho})\in\mathfrak{S}_{g_{\rm sup}(E)}({\cal H}_Q)$ for any $\hat{\rho}\in\mathfrak{S}_{E}({\cal H}_Q)$ and
    \begin{equation}
        \frac{1}{2}\left\|{\cal U}^{\rm noisy} - {\cal N}^s \circ {\cal U} \right\|_{\diamond}^{E} \leq \epsilon. \label{eq:noise_model_for_gate}
    \end{equation} 
\end{definition}
\begin{definition}[$(E,s,\epsilon)$-independent Markovian noise model for measurement] \label{def:E_s_eps_Markovian_meas}
    Consider a collection ${\cal C}'$ of physical systems that comprise a CV quantum computer, and regard other physical systems as environments.
    A noisy measurement for a target measurement channel ${\cal M}:\mathfrak{S}_{E}({\cal H}_Q)\to {\cal P}$, where ${\cal P}$ denotes the set of probability distribution over the binary outcome $\{0,1\}$, is said to obey the $(E,s,\epsilon)$-independent Markovian noise model if it implements a CPTP map ${\cal M}^{\rm noisy}$, independent of other physical systems in ${\cal C}'$, such that there exists a CPTP map ${\cal N}^s$ defined in Def.~\ref{def:def_N_s} that satisfies
    \begin{equation}
        \frac{1}{2}\|{\cal M}^{\rm noisy} - {\cal M} \circ {\cal N}^s\|_{\diamond}^{E} \leq \epsilon. \label{eq:noise_model_for_measurement}
    \end{equation}
\end{definition}
The meaning of a positive constant $E_{\rm prep}$ and a fixed, positive, monotonically increasing, locally bounded function $g_{\rm sup}$ in the FT-GKP circuit will be clarified in the next section.
One may consider that the above definitions of the noise strength through Eqs.~\eqref{eq:noise_model_for_preparation}, \eqref{eq:noise_model_for_gate}, and \eqref{eq:noise_model_for_measurement} are not conventional.  In fact, if we assume that our target state, gate, and measurement are $\hat{\rho}^s_{\psi}$, ${\cal N}^s\circ{\cal U}$, and ${\cal M}\circ{\cal N}^s$ rather than $\ket{\overline{\psi}}\!\bra{\overline{\psi}}$, ${\cal U}$, and ${\cal M}$, respectively, then the definition reduces to a more conventional Markovian-type noise model with the noise strength measured by the energy-constrained diamond-norm distance.  However, for later use, we add a parameter $s$ in the definition here.  The reader can always reproduce a more conventional picture by assuming that $\hat{\rho}^s_{\psi}$, ${\cal N}^s\circ{\cal U}$, and ${\cal M}\circ{\cal N}^s$ are the target state, gate, and measurement.

In the following, we give some explicit examples of the $(s,\epsilon)$-independent and $(E,s,\epsilon)$-independent Markovian noise models that frequently appear in quantum optical experiments and in theoretical analyses of optical quantum computation.  
For state preparation, we analyze how the conventional approximate GKP codes studied in Ref.~\cite{Matsuura2020} can be interpreted as $(s,\epsilon)$-independent Markovian noise. 
We consider the following standard form of the approximate GKP codeword studied in Ref.~\cite{Matsuura2020}, which is symmetric under the Fourier transform and parameterized by only one parameter $\sigma^2$:
\begin{equation}
    \ket{\overline{0}_{\sigma^2}^{\rm app}}\coloneqq  \frac{1}{\sqrt{c \sigma^2 N_{\sigma^2}}}\sum_{m\in\mathbb{Z}} e^{-4c^2\sigma^2m^2}\int dt\, e^{-\frac{t^2}{4\sigma^2}}\ket{t+2cm\sqrt{1-4\sigma^4}}_q, \label{eq:conventional_approx}
\end{equation}
where the normalization factor $N_{\sigma^2}$ is explicitly given in Ref.~\cite{Matsuura2020}.  
From Eq.~\eqref{eq:position_to_Zak}, we obtain the Zak-basis representation of $\ket{\overline{0}_{\sigma^2}^{\rm app}}$ state as
\begin{equation}
     \ket{\overline{0}_{\sigma^2}^{\rm app}} = \frac{e^{iz_1z_2/2}}{\sqrt{c^2\sigma^2 N_{\sigma^2}}} \Theta\! \begin{bmatrix}(z_1/2c \ \;0)^{\top} \\ (0\ -z_2/c)^{\top}\end{bmatrix}\!(\bm{0}, i\Omega) \ket{z_1,z_2}_Z,
\end{equation}
where $\Theta\begin{bmatrix} \vec{a} \\ \vec{b} \end{bmatrix}(\vec{z},\bm{\tau})\coloneqq \sum_{\vec{s}\in\mathbb{Z}^n}\exp[\pi i (\vec{s}+\vec{a})^{\top}\bm{\tau}(\vec{s}+\vec{a}) + 2\pi i(\vec{z}+\vec{b})^{\top}(\vec{s}+\vec{a})]$ denotes the Riemann theta function \cite{Matsuura2020}, and the ${2\times 2}$ matrix $\Omega$ is given by
\begin{equation}
    \Omega \coloneqq \frac{1}{\sigma^2}\begin{pmatrix}
        1 & -\sqrt{1-4\sigma^4} \\ 
        -\sqrt{1-4\sigma^4} & 1
    \end{pmatrix}.
\end{equation}
This state has a support all over the region $(z_1,z_2)\in[-\frac{c}{2},\frac{c}{2})\times[-\frac{c}{2},\frac{3c}{2})$, and thus has non-zero trace distance from the $s$-parameterized GKP state in Def.~\ref{def:s-parameterized_state} for any $s<c/2$.  Thus, the preparation of this type of approximate GKP state $\ket{\overline{0}_{\sigma^2}^{\rm app}}$ can be regarded as an example of $(s,\epsilon)$-independent Markovian noise for preparation.
One may expect that $\hat{\Pi}_s\ket{\overline{0}_{\sigma^2}^{\rm app}}$ (properly normalized) would be a valid $s$-parameterized GKP state and thus $(1-\bra{\overline{0}_{\sigma^2}^{\rm app}}\hat{\Pi}_s\ket{\overline{0}_{\sigma^2}^{\rm app}})^{1/2}$ would be the noise strength $\epsilon$ of the $(s,\epsilon)$-independent Markovian noise.  However, the state $\hat{\Pi}_s\ket{\overline{0}_{\sigma^2}^{\rm app}}$ (after normalization) unfortunately has infinite energy.  This is because a state whose wave function has discontinuity has an infinite energy in general. We thus need to consider a \quoted{smoothed version} of $\hat{\Pi}_s\ket{\overline{0}_{\sigma^2}^{\rm app}}$ that has a finite energy $E_{\rm prep}$, but such a finite-energy state may have larger trace distance from $\ket{\overline{0}_{\sigma^2}^{\rm app}}$ than that between $\hat{\Pi}_s\ket{\overline{0}_{\sigma^2}^{\rm app}}$ and $\ket{\overline{0}_{\sigma^2}^{\rm app}}$.  In general, there should be a trade-off relation between the trace distance and the energy $E_{\rm prep}$ of the $s$-parameterized GKP state that approximates $\ket{\overline{0}_{\sigma^2}^{\rm app}}$.  
Deriving this explicit trade-off between the noise strength $\epsilon$ and the energy $E_{\rm prep}$ is left to future work.

For the CV gate operations we listed in Sec.~\ref{sec:physical_operations}, the bound on the energy-constrained diamond-norm distance between an ideal CV gate and its conventional experimental approximation has been studied in detail \cite{Sharma2020}.  The results in Ref.~\cite{Sharma2020} thus have the direct interpretations as $(E,0,\epsilon)$-independent Markovian noise.  Importantly, when the identity gate (simply waiting or employing a delay line) is subject to independent photon loss on each mode (which is a dominant noise source in an optical CV system), the energy-constrained diamond-norm distance between the identity operation  $\mathrm{Id}$ and its noisy version, i.e., the loss channel ${\cal L}^{\eta}$ with the transmissivity (transmission probability) $\eta$ is given for an integer $E$ by~\cite{Nair2018, Sharma2020}
\begin{equation}
    \frac{1}{2}\|\mathrm{Id} - {\cal L}^{\eta}\|^E_{\diamond} \leq \sqrt{1-\eta^{E}},
\end{equation}
and thus it can be interpreted as $(E,0,\sqrt{1-\eta^E})$-independent Markovian noise.  The nonideality caused by the conventional experimental approximations of displacement and SUM gates can also be reduced to the distance between the identity and the loss channel given above~\cite{Sharma2020}.  
Another ubiquitous noise in an optical CV system is a random phase rotation.
The energy-constrained diamond-norm distance between a random phase-rotation channel and an identity channel can be computed as \cite{Winter2017}
\begin{equation}
    \frac{1}{2}\Bigl\| \mathrm{Id} - \int_{-\pi}^{\pi} d\theta\, f(\theta) {\cal R}[\theta] \Bigr\|_{\diamond}^{E} \leq \int_{-\pi}^{\pi} d\theta\, f(\theta) \sqrt[3]{4|\theta|E},
\end{equation}
where $f(\theta)$ is a probability distribution over $[-\pi,\pi)$, typically well peaked at $\theta=0$.
We have only considered the case $s=0$ up here, but there may in general be a trade-off between $s$ and $\epsilon$.  Deriving an explicit trade-off between $s$ and $\epsilon$ for the given channel may not be straightforward in general.

For measurement, experimental homodyne detectors generally have a finite resolution~$b$; i.e., the measurement outcome is discretized with a bin size $b$.  It also has a finite range~$\Gamma$; i.e., the value of $\hat{q}$ larger (resp.~smaller) than $\Gamma$ (resp.~$-\Gamma$) is rounded to $\Gamma$ (resp.~$-\Gamma$).  The binned homodyne detection for $\hat{q}$ with the bin size $b$ can be interpreted as the ideal homodyne detection followed by the binning with the size $b$, but then, this is equivalent to $\{\hat{q}\}_b$-dependent displacement before the ideal homodyne detection.  Since it is up-to-$b$ displacement noise in the $\hat{q}$ quadrature, it can be regarded as the $(E,b,0)$-independent Markovian noise for any $E$.  On the other hand, the effect of the finite range $\Gamma$ should depend on the energy $E$ of the input state since the state with large $E$ has more probability to be observed in $|\hat{q}|\geq \Gamma$.  Using the fact that $\braket{\hat{q}^2}_{\hat{\rho}}\leq 2\braket{\hat{n}}_{\hat{\rho}}+1\leq 2E+1$ for any $\hat{\rho}\in \mathfrak{S}_E({\cal H})$, we have 
\begin{equation}
    \mathrm{Pr}\bigl[|\hat{q}|> \Gamma \mid \hat{\rho}\in\mathfrak{S}_E({\cal H})\bigr] = \mathrm{Pr}\bigl[\hat{q}^2> \Gamma^2 \mid \hat{\rho}\in\mathfrak{S}_E({\cal H})\bigr] \leq \frac{\braket{\hat{q}^2}_{\hat{\rho}}}{\Gamma^2} \leq \frac{2E+1}{\Gamma^2},
\end{equation}
where we used Markov's inequality in the first inequality.  Thus, for the experimental approximation ${\cal M}_{\rm hom}^{\rm noisy}[\Gamma]$ with the range $\Gamma$ of the ideal homodyne detector ${\cal M}^{\rm ideal}_{\rm hom}$ satisfies
\begin{equation}
    \frac{1}{2}\|{\cal M}^{\rm ideal}_{\rm hom} - {\cal M}_{\rm hom}^{\rm noisy}[\Gamma]\|^{E}_{\diamond} \leq \frac{2E+1}{\Gamma^2}, \label{eq:finite_range}
\end{equation}
and can thus be regarded as belonging to an $(E,0,(2E+1)/\Gamma^2)$-independent Markovian noise model for measurement.  Since Eq.~\eqref{eq:finite_range} still holds even if we perform post-processing after both ${\cal M}^{\rm ideal}_{\rm hom}$ and ${\cal M}_{\rm hom}^{\rm noisy}[\Gamma]$, the experimentally realizable homodyne detector with resolution $b$ and a range $\Gamma$ can be regarded as $(E,b,(2E+1)/\Gamma^2)$-independent Markovian noise.

\section{General fault tolerance for the GKP code}\label{sec:general_FT}

\subsection{Construction of GKP gadgets and a fault-tolerant GKP circuit} \label{sec:gadget_construction}

In a fault-tolerant circuit, each location $C_i$ of the original circuit $C$ is replaced with the \quoted{gadget} to implement such operations in a fault-tolerant manner \cite{Gottesman2009}.
In a concatenated code, for example, each location of the circuit $C$ in the $n^{\text{th}}$ level of concatenation consists of a circuit with and $(n-1)^{\text{th}}$-level code and fault-tolerant operations with it.  Since the GKP code encodes a qubit in a CV system, each location in the original qubit circuit $C$ is, in this case, replaced with the physical CV operations listed in Sec.~\ref{sec:physical_operations}.
Explicitly, we define the fault-tolerant GKP gadget for each qubit circuit element as follows.
\begin{align}
    \text{Preparation:}\hspace{3cm}& \nonumber\\
    \begin{picture}(60,20)
    \thicklines
        \put(30,5){\oval(40,20)[l]}
        \put(30,-5){\line(0,1){20}}
        \put(15,5){\makebox(0,0)[l]{$\ket{\overline{0}}$}}
        \put(30,5){\line(1,0){20}}
    \end{picture}
    &\longrightarrow
    \begin{picture}(60,20)
    \thicklines
        \put(10,5){\makebox(0,0)[l]{$\hat{\rho}^s_0$}}
        \put(27,5){\line(1,0){20}}
    \end{picture} \label{eq:computational-basis_preparation_gadget}\\
    \begin{picture}(60,25)
    \thicklines
        \put(30,5){\oval(40,20)[l]}
        \put(30,-5){\line(0,1){20}}
        \put(13,5){\makebox(0,0)[l]{$\ket{\overline{Y}}$}}
        \put(30,5){\line(1,0){20}}
    \end{picture}
    &\longrightarrow
    \begin{picture}(60,25)
    \thicklines
        \put(10,5){\makebox(0,0)[l]{$\hat{\rho}^s_Y$}}
        \put(27,5){\line(1,0){20}}
    \end{picture} \\
    \begin{picture}(60,25)
    \thicklines
        \put(30,5){\oval(40,20)[l]}
        \put(30,-5){\line(0,1){20}}
        \put(13,5){\makebox(0,0)[l]{$\ket{\overline{\frac{\pi}{8}}}$}}
        \put(30,5){\line(1,0){20}}
    \end{picture}
    &\longrightarrow
    \begin{picture}(60,25)
    \thicklines
        \put(10,5){\makebox(0,0)[l]{$\hat{\rho}^s_{\frac{\pi}{8}}$}}
        \put(27,5){\line(1,0){20}}
    \end{picture} 
    \\
    &\nonumber\\
    \text{Gate:}\hspace{3cm}& \nonumber\\
    \begin{picture}(80,20)
    \thicklines
        \put(10,5){\line(1,0){20}}
        \put(40,5){\circle{20}}
        \put(40,5){\makebox(0,0){$\overline{Z}$}}
        \put(50,5){\line(1,0){20}}
    \end{picture}
    &\longrightarrow
    \begin{picture}(80,20)
    \thicklines
        \put(10,5){\line(1,0){20}}
        \put(30,-5){\framebox(30, 20){$\hat{V}(0,c)$}}
        \put(60,5){\line(1,0){20}}
    \end{picture}
    \\
    \begin{picture}(80,25)
    \thicklines
        \put(10,5){\line(1,0){20}}
        \put(40,5){\circle{20}}
        \put(40,5){\makebox(0,0){$\overline{X}$}}
        \put(50,5){\line(1,0){20}}
    \end{picture}
    &\longrightarrow
    \begin{picture}(80,25)
    \thicklines
        \put(10,5){\line(1,0){20}}
        \put(30,-5){\framebox(30, 20){$\hat{V}(c,0)$}}
        \put(60,5){\line(1,0){20}}
    \end{picture}
    \\
    \begin{picture}(80,25)
    \thicklines
        \put(10,5){\line(1,0){20}}
        \put(40,5){\circle{20}}
        \put(40,5){\makebox(0,0){$\overline{H}$}}
        \put(50,5){\line(1,0){20}}
    \end{picture}
    &\longrightarrow
    \begin{picture}(80,25)
    \thicklines
        \put(10,5){\line(1,0){20}}
        \put(30,-5){\framebox(20, 20){$\hat{F}$}}
        \put(50,5){\line(1,0){20}}
    \end{picture} \\
    \begin{picture}(70,30)
    \thicklines
        \put(10,15){\line(1,0){50}}
        \put(10,-5){\line(1,0){50}}
        \put(35,15){\circle*{5}}
        \put(35,-10){\line(0,1){25}}
        \put(35,-5){\circle{10}}
    \end{picture}
    &\longrightarrow
    \begin{picture}(70,30)
    \thicklines
        \put(10,15){\line(1,0){64}}
        \put(10,-5){\line(1,0){20}}
        \put(42,15){\circle*{5}}
        \put(42,1){\line(0,1){14}}
        \put(30,-11){\framebox(24,12){SUM}}
        \put(54,-5){\line(1,0){20}}
    \end{picture}
    \\
    \rule{0pt}{33pt}
    \begin{picture}(80,25)
    \thicklines
        \put(10,5){\line(1,0){20}}
        \put(40,5){\circle{20}}
        \put(40,5){\makebox(0,0){$\overline{I}$}}
        \put(50,5){\line(1,0){20}}
    \end{picture}
    &\longrightarrow
    \begin{picture}(80,25)
    \thicklines
        \put(10,5){\line(1,0){20}}
        \put(30,-5){\framebox(20, 20){$\hat{I}$}}
        \put(50,5){\line(1,0){20}}
    \end{picture} \\
    &\nonumber\\
    \text{Measurement:}\hspace{3cm}&\nonumber\\
    \begin{picture}(60,20)
    \thicklines
        \put(10,5){\line(1,0){20}}
        \put(30,-5){\line(0,1){20}}
        \put(30,5){\oval(40,20)[r]}
        \put(37,5){\makebox(0,0){$\overline{Z}$}}
    \end{picture}
    &\longrightarrow
    \begin{picture}(80,20)
    \thicklines
        \put(10,5){\line(1,0){20}}
        \put(30,-5){\line(0,1){20}}
        \put(30,5){\oval(70,20)[r]}
        \put(35,5){\makebox(20,0){$\hat{q}=s$}}
        \put(65,6){\line(1,0){10}}
        \put(65,4){\line(1,0){10}}
    \end{picture}
    \lfloor s/c \rceil \text{ mod }2
    \\
    \begin{picture}(60,25)
    \thicklines
        \put(10,5){\line(1,0){20}}
        \put(30,-5){\line(0,1){20}}
        \put(30,5){\oval(40,20)[r]}
        \put(37,5){\makebox(0,0){$\overline{X}$}}
    \end{picture}
    &\longrightarrow
    \begin{picture}(80,25)
    \thicklines
        \put(10,5){\line(1,0){20}}
        \put(30,-5){\line(0,1){20}}
        \put(30,5){\oval(70,20)[r]}
        \put(35,5){\makebox(20,0){$\hat{p}=t$}}
        \put(65,6){\line(1,0){10}}
        \put(65,4){\line(1,0){10}}
    \end{picture}
    \lfloor t/c \rceil \text{ mod }2\\
    &\nonumber
\end{align}
In the above, $\lfloor a \rceil$ denotes the rounding of $a$ defined in Eq.~\eqref{eq:rounding}.  In the following, we may omit writing the wait operations (identity gates) explicitly in the circuits unless necessary.
With these primitive gadgets, we can construct a universal gate set in Eq.~\eqref{eq:universal_gate_set}.
The logical phase gate $\overline{S}$ can be realized by the following catalytic circuit \cite{Takagi2017}.
\begin{align}
    &\begin{picture}(80,30)
        \thicklines
        \put(10,5){\line(1,0){20}}
        \put(40,5){\circle{20}}
        \put(40,5){\makebox(0,0){$\overline{S}$}}
        \put(50,5){\line(1,0){20}}
    \end{picture}
    =
    \begin{picture}(200,30)
    \thicklines
        \put(30,-5){\oval(40,20)[l]}
        \put(30,-15){\line(0,1){20}}
        \put(13,-5){\makebox(0,0)[l]{$\ket{\overline{Y}}$}}
        \put(30,15){\line(1,0){120}}
        \put(30,-5){\line(1,0){50}}
        \put(55,15){\circle*{5}}
        \put(55,-10){\line(0,1){25}}
        \put(55,-5){\circle{10}}
        \put(90,-5){\circle{20}}
        \put(90,-5){\makebox(0,0){$\overline{H}$}}
        \put(100,-5){\line(1,0){50}}
        \put(125,15){\circle*{5}}
        \put(125,-10){\line(0,1){25}}
        \put(125,-5){\circle{10}}
        \put(160,-5){\circle{20}}
        \put(160,-5){\makebox(0,0){$\overline{H}$}}
        \put(170,-5){\line(1,0){20}}
    \end{picture}\\
    &\nonumber
\end{align}
In the right-hand side above, the phase gate applies to the first qubit.  The last $\overline{H}$ gate on the second wire resets the state to $\ket{\overline{Y}}$, which can replace a $\ket{\overline{Y}}$-state preparation for later $\overline{S}$ gates.
Conventionally, the phase gate for the GKP code is realized by the shear $\exp(i\hat{q}^2/2)$, but to simplify the fault-tolerance condition stated in this section, we assume the $\ket{\overline{Y}}$-state preparation and the above catalytic circuit instead.  As mentioned later, we can use the conventional gate $\exp(i\hat{q}^2/2)$ instead of $\ket{\overline{Y}}$ preparation at the cost of a slight modification of the fault-tolerance condition.
Since all the Clifford gates are generated by the phase gate, the Hadamard gate $\overline{H}$, and the CNOT gate, we conclude that the above gadgets are sufficient to realize the Clifford operations.  Furthermore, when we combine the $\ket{\overline{\frac{\pi}{8}}}$ state with the Clifford gates, we can realize the $\overline{T}$ gate through the gate teleportation as follows \cite{Zhou2000}.
\begin{align}
    &\begin{picture}(80,30)
        \thicklines
        \put(10,5){\line(1,0){20}}
        \put(40,5){\circle{20}}
        \put(40,5){\makebox(0,0){$\overline{T}$}}
        \put(50,5){\line(1,0){20}}
    \end{picture}
    =
    \begin{picture}(200,30)
    \thicklines
        \put(30,-5){\oval(40,20)[l]}
        \put(30,-15){\line(0,1){20}}
        \put(13,-5){\makebox(0,0)[l]{$\ket{\overline{\frac{\pi}{8}}}$}}
        \put(30,15){\line(1,0){75}}
        \put(30,-5){\line(1,0){50}}
        \put(55,15){\circle*{5}}
        \put(55,-10){\line(0,1){25}}
        \put(55,-5){\circle{10}}
        \put(80,-5){\oval(40,20)[r]}
        \put(80,-15){\line(0,1){20}}
        \put(87,-5){\makebox(0,0){$\overline{Z}$}}
        \put(100,-4){\line(1,0){15}}
        \put(100,-6){\line(1,0){15}}
        \put(115,-5){\circle*{5}}
        \put(114,-5){\line(0,1){10}}
        \put(116,-5){\line(0,1){10}}
        \put(115,15){\circle{20}}
        \put(115,15){\makebox(0,0){$\overline{S}$}}
        \put(125,15){\line(1,0){20}}
    \end{picture}\label{eq:t-gate_gadget}\\
    &\nonumber
\end{align}
The gadgets in Eqs.~\eqref{eq:computational-basis_preparation_gadget}--\eqref{eq:t-gate_gadget} implement the universal gate set in Eq.~\eqref{eq:universal_gate_set}, as well as the computational-basis state preparation and measurement.

For the theory of fault tolerance, we also need to construct a GKP error-correction (EC) gadget.  There are several options for this, but here we use the Knill-type EC gadget~\cite{Knill2005, Walshe2020} for several reasons, which will be made clear later.
\begin{align}
    \begin{picture}(90,20)
    \thicklines
        \put(10,5){\line(1,0){20}}
        \put(30,-5){\framebox(30,20){EC}}
        \put(60,5){\line(1,0){20}}
    \end{picture}
    &=
    \begin{picture}(180,40)
    \thicklines
        \put(30,5){\oval(40,20)[l]}
        \put(30,-5){\line(0,1){20}}
        \put(15,5){\makebox(0,0)[l]{$\ket{\overline{0}}$}}
        \put(30,5){\line(1,0){20}}
        \put(90,30){\line(1,0){50}}
        \put(70,-20){\oval(40,20)[l]}
        \put(70,-30){\line(0,1){20}}
        \put(55,-20){\makebox(0,0)[l]{$\ket{\overline{0}}$}}
        \put(30,5){\line(1,0){20}}
        \put(70,-20){\line(1,0){35}}
        \put(60,5){\circle{20}}
        \put(60,5){\makebox(0,0){$\overline{H}$}}
        \put(70,5){\line(1,0){70}}
        \put(90,-25){\line(0,1){30}}
        \put(90,5){\circle*{5}}
        \put(90,-20){\circle{10}}
        \put(115,30){\circle*{5}}
        \put(115,5){\circle{10}}
        \put(115,0){\line(0,1){30}}
        \put(115,-20){\circle{20}}
        \put(115,-20){\makebox(0,0){$\overline{I}$}}
        \put(125,-20){\line(1,0){15}}
        \put(140,20){\line(0,1){20}}
        \put(140,30){\oval(40,20)[r]}
        \put(148,30){\makebox(0,0){$\overline{X}$}}
        \put(140,-5){\line(0,1){20}}
        \put(140,5){\oval(40,20)[r]}
        \put(148,5){\makebox(0,0){$\overline{Z}$}}
        \put(150,-20){\circle{20}}
        \put(150,-20){\makebox(0,0){$\overline{I}$}}
        \put(160,-20){\line(1,0){15}}
    \end{picture}\label{eq:knill_EC}\\
    &\nonumber \\
    &\nonumber
\end{align}
In the above, we omit the correction operations.  The correction operation can instead be taken into account by updating the Pauli frame~\cite{Chamberland2018} or by changing the successive gate or measurement.

Using these GKP gadgets, we make a precise definition of our fault-tolerant protocol.  Our protocol is constructed based on the protocol for concatenated codes in Ref.~\cite{Gottesman2009}, but we only aim at describing how the qubit-level circuit can be constructed by a fault-tolerant CV circuit and how the noise property on this qubit-level circuit is determined by the underlying CV circuit.  Thus, the qubit-level circuit may also need to be encoded by a concatenated code as in Ref.~\cite{Gottesman2009}.
Putting this in mind, we define the fault-tolerant-GKP (FT-GKP) circuit as follows.

\begin{definition}[FT-GKP circuit]
    Let $C$ be a circuit on qubits. Let $\{C_i\}_{i\in{\cal I}}$ be a set of locations of $C$, where each $C_i$ is a preparation location, a measurement location, a gate location, or a wait location.  We assume that the circuit $C$ can be divided so that every qubit is involved in exactly one location at each time step.  Two locations are considered \emph{consecutive} if they occur at adjacent time steps and they share a qubit.  Given $C$, the \emph{fault-tolerant-GKP (FT-GKP) protocol} provides a CV circuit $C'$ constructed by replacing each location $C_i$ with a GKP gadget $C'_i$ for $C_i$ and adding GKP EC gadgets between any pair of consecutive locations.  The circuit $C'$ is referred to as a \emph{FT-GKP circuit} for $C$.
\end{definition}

As stated previously, the above definition is not a whole protocol for fault tolerance; the qubit circuit $C$ should be a qubit version of an FT implementation of an original circuit $C_{\rm orig}$ representing the computation in the sense of Ref.~\cite{Gottesman2009}.  If the concatenation level of the qubit circuit $C$ is enough to correct logical errors of the FT-GKP circuit, the whole protocol will be fault-tolerant.  
Our analysis in the following will translate the noise model in the CV system into a noise model at the qubit level so that we can use the established fault-tolerant protocols at the qubit level to achieve fault tolerance.

\subsection{Fault-tolerance conditions for GKP gadgets}\label{sec:FT_conditions}

In Ref.~\cite{Gottesman2009}, they consider the equivalence class of faulty gadgets in which at most $s$ locations inside the gadget are faulty.  Then, the fault-tolerance condition requires that the gadget with $s$ faulty locations causes at most weight-$s$ errors, i.e., errors act on $s$ qubits that comprise the code.  If only this condition is satisfied, the whole protocol can be shown to be fault tolerant for any errors that occur according to the noise model.  Although these definitions have clear meaning in qubit concatenated error-correcting codes, they do not in the GKP code.  After all, the CV system should be in the bottom layer of concatenation in our protocol, so the faults at the physical level are directly caused by the physical noise of the CV systems rather than qubits.  We need to find an alternative equivalence class of faults or errors and the corresponding fault-tolerance criteria in the CV case.

The GKP code is introduced to correct displacement errors in the phase space up to $c/2$ with $c$ in Eq.~\eqref{eq:C}.  Thus, the natural parameterization of errors should be the \quoted{amount} of displacement.
To make this statement rigorous, we introduce the following family of projection operators, which is an appropriate modification of the filter operator in Ref.~\cite{Gottesman2009} for the GKP code.

\begin{definition}[The stabilizer-subsystem (SSS) $r$-filter and the ideal GKP decoder]
Let $0<r\leq c/2$.
The SSS $r$-filter $\hat{\Pi}_{r}$ is defined as
\begin{equation}
    \hat{\Pi}_{r}\coloneqq \sum_{\mu=0,1}\int_{|z_1|< r}dz_1 \int_{|z_2|< r}dz_2 \, \ket{\mu;z_1,z_2}\!\bra{\mu;z_1, z_2}_{LS}=\hat{I}_{L}\otimes \int_{|z_1|< r}dz_1 \int_{|z_2|< r}dz_2 \, \ket{z_1,z_2}\!\bra{z_1, z_2}_{S}.
    \label{eq:SSS_r-filter}
\end{equation}
The ideal GKP decoder ${\cal D}_{\rm GKP}$ is a map from a CV system to a qubit, defined as
\begin{equation}
    {\cal D}_{\rm GKP}(\hat{\rho})\coloneqq \sum_{\mu,\nu\in\{0,1\}}\int_{|z_1|< \frac{c}{2}}dz_1 \int_{|z_2|< \frac{c}{2}}dz_2 \, \ket{\mu}\!\bra{\mu;z_1, z_2}_{LS} \hat{\rho} \ket{\nu;z_1, z_2}_{LS}\!\bra{\nu}  = \Tr_{S}(\hat{\rho}).\label{eq:ideal_decoder}
\end{equation}
\end{definition}

The SSS $r$-filter is a tool to know how much the state is effectively displaced away from the \quoted{ideal GKP code space}, and the ideal GKP decoder is a tool to know the reduced state on the logical qubit of the SSS decomposition.
One can check from Eqs.~\eqref{eq:orthogonality_SSS} and \eqref{eq:completeness_SSS} that the SSS $r$-filter is a projection operator and satisfies $\hat{\Pi}_{\frac{c}{2}}=\hat{I}$.  
One should be careful, however, that $\lim_{r\to 0}\hat{\Pi}_r$ does not converge to a projection onto the \quoted{ideal GKP code space} but the zero operator with respect to the strong operator topology.
The SSS $r$-filter and the ideal GKP decoder will be represented graphically as follows.
\begin{equation*}
\begin{picture}(300,40)
\thicklines

\put(0,20){\line(1,0){10}}
\put(10,10){\framebox(5,20){}}
\put(15,20){\line(1,0){10}}
\put(18,25){\makebox(0,0)[bl]{$\scriptstyle r$}}
\put(40,12){\makebox(20,16)[l]{\text{SSS $r$-filter}}}

\put(140,20){\line(1,0){10}}
\put(150,10){\line(0,1){20}}
\put(150,30){\line(1,-1){10}}
\put(160,20){\line(-1,-1){10}}
\thinlines
\put(160,20){\line(1,0){10}}
\put(185,12){\makebox(55,16)[l]{\text{Ideal GKP decoder}}}

\end{picture}
\end{equation*}

Now, we define a faulty version of the gadgets introduced in the previous section, which are again generalizations of Ref.~\cite{Gottesman2009} for our scenario.  We start with the definition of the $s$-parameterized GKP preparation.  
This definition is the same as the one defined in the previous section, but we explicitly show the dependence on the parameter $s$ in Def.~\ref{def:s-parameterized_state}.  Alternately, one can consider the $s$-preparation as the ideal $s_1$-parameterized GKP state preparation followed by the noise map ${\cal N}^{s_2}$ defined in Def.~\ref{def:def_N_s}, but this is again an $(s_1+s_2)$-parameterized GKP state preparation as can be seen from Eq.~\eqref{eq:s-parameterized_mixed}.
\begin{definition}[$s$-preparation]\label{def:s-preparation}
The $s$-parameterized GKP preparation ($s$-preparation for short) gadget is defined as follows.
\begin{equation}
    \text{$s$-preparation:}\qquad
    \begin{picture}(60,20)
    \thicklines
        \put(30,5){\oval(50,20)[l]}
        \put(30,-5){\line(0,1){20}}
        \put(10,5){\makebox(0,0)[l]{$\ket{\overline{\psi}}$}}
        \put(30,5){\line(1,0){20}}
        \put(32,10){\makebox(0,0)[bl]{$\scriptstyle s$}}
    \end{picture}
    \longrightarrow
    \begin{picture}(60,20)
    \thicklines
        \put(10,5){\makebox(0,0)[l]{$\hat{\rho}_{\psi}^s$}}
        \put(27,5){\line(1,0){20}}
    \end{picture}
    \label{eq:s-preparation}
\end{equation}
In the above, $\psi$ denotes either $0$, $Y$, or $\frac{\pi}{8}$, and $\hat{\rho}_{\psi}^s$ is as defined in Def.~\ref{def:s-parameterized_state}.  
\end{definition}

For the $s$-parameterized GKP gate gadget, the parameter $s$ denotes the degree of displacement noise added during these operations.  
When the noise is Markovian, the right-hand side in the following definition includes general noise models as $s\to\infty$, since any CPTP map ${\cal N}_U$ that is a noisy version of the gate $\hat{U}$ can be decomposed as $({\cal N}_U\circ{\cal U}^{-1})\circ{\cal U}$, where ${\cal U}(\hat{\rho})=\hat{U}\hat{\rho} \hat{U}^{-1}$, and the channel ${\cal N}_U\circ{\cal U}^{-1}$ can be expanded by displacement operators.

\begin{definition}[$s$-gate]\label{def:s-gate}
Let us define an $s$-parameterized GKP gate gadget ($s$-gate for short) as follows.
\begin{equation}
    \text{$s$-gate:} \qquad \begin{picture}(80,20)
    \thicklines
        \put(10,5){\line(1,0){20}}
        \put(40,5){\circle{20}}
        \put(40,5){\makebox(0,0){$\overline{U}$}}
        \put(50,5){\line(1,0){20}}
        \put(51,10){\makebox(0,0)[bl]{$\scriptstyle s$}}
    \end{picture}
    \longrightarrow
    \begin{picture}(110,20)
    \thicklines
        \put(10,5){\line(1,0){20}}
        \put(30,-5){\framebox(20, 20){$\hat{U}$}}
        \put(50,5){\line(1,0){10}}
        \put(60,-5){\framebox(20, 20){${\cal N}^s$}}
        \put(80,5){\line(1,0){20}}
    \end{picture}
    \label{eq:s-gate}
\end{equation}
In the above, $\hat{U}$ is either $\hat{V}(0,c)$, $\hat{V}(c,0)$, $\hat{F}$, $\exp(-i\hat{q}_1\hat{p}_2)$, or $\hat{I}$, and ${\cal N}^s$ is defined in Def.~\ref{def:def_N_s}. 
\end{definition}

In the same way, we define the $s$-parameterized measurement gadget below.
As opposed to the $s$-gate case, the right-hand side of the definition below may not be a general noise model for the measurement even when the noise is Markovian and $s\rightarrow \infty$ since general POVM may not necessarily be represented as a noise channel followed by the ideal homodyne detection and the GKP binning.  More general POVM can instead be approximated by the POVM of the form in Eq.~\eqref{eq:s-measurement} in the energy-constrained diamond-norm distance as considered in Eq.~\eqref{eq:noise_model_for_measurement}.
\begin{definition}[$s$-measurement]\label{def:s-measurement}
The $s$-parameterized GKP measurement gadget ($s$-measurement for short) is defined as follows.
\begin{equation}
    \text{$s$-measurement:} \qquad
    \begin{picture}(70,20)
    \thicklines
        \put(10,5){\line(1,0){20}}
        \put(30,-5){\line(0,1){20}}
        \put(30,5){\oval(58,20)[r]}
        \put(33,5){\makebox(20,0){$\overline{Z}/\overline{X}$}}
        \put(60,10){\makebox(0,0)[bl]{$\scriptstyle s$}}
    \end{picture}
    \longrightarrow
    \begin{picture}(115,20)
    \thicklines
        \put(10,5){\line(1,0){20}}
        \put(30,-5){\framebox(20,20){${\cal N}^s$}}
        \put(50,5){\line(1,0){10}}
        \put(60,-5){\line(0,1){20}}
        \put(60,5){\oval(80,20)[r]}
        \put(65,5){\makebox(27,0){$\hat{q}/\hat{p}=t$}}
        \put(100,6){\line(1,0){10}}
        \put(100,4){\line(1,0){10}}
    \end{picture}
    \lfloor t/c \rceil \text{ mod }2
    \label{eq:s-measurement}
\end{equation}
\end{definition}

Finally, we need to introduce a parameterized version of the EC gadget.
The $s$-EC gadget should work as resetting the degree of displacement.  This resetting of the wave function in GKP error correction is first pointed out in Ref.~\cite{Menicucci2014} for Gaussian random displacement noise, but here we make the statement more rigorous and applicable even beyond Gaussian noise using the tools we developed so far.

\begin{definition}[$s$-EC] \label{def:s-EC}
We define an $s$-parameterized EC gadget as follows, replacing all the circuit elements in Eq.~\eqref{eq:knill_EC} with $s$-preparations, $s$-gates, and $s$-measurements:
\begin{align}
    \text{$s$-EC:} \qquad 
    \begin{picture}(90,20)
    \thicklines
        \put(10,5){\line(1,0){20}}
        \put(30,-5){\framebox(30,20){EC}}
        \put(60,5){\line(1,0){20}}
        \put(62,10){\makebox(0,0)[bl]{$\scriptstyle s$}}
    \end{picture}
    &=
    \begin{picture}(190,40)
    \thicklines
        \put(30,5){\oval(40,20)[l]}
        \put(30,-5){\line(0,1){20}}
        \put(15,5){\makebox(0,0)[l]{$\ket{\overline{0}}$}}
        \put(30,5){\line(1,0){20}}
        \put(32,10){\makebox(0,0)[bl]{$\scriptstyle s_0$}}
        \put(90,30){\line(1,0){50}}
        \put(70,-20){\oval(40,20)[l]}
        \put(70,-30){\line(0,1){20}}
        \put(55,-20){\makebox(0,0)[l]{$\ket{\overline{0}}$}}
        \put(72,-15){\makebox(0,0)[bl]{$\scriptstyle s_0$}}
        \put(30,5){\line(1,0){20}}
        \put(70,-20){\line(1,0){110}}
        \put(60,5){\circle{20}}
        \put(60,5){\makebox(0,0){$\overline{H}$}}
        \put(71,10){\makebox(0,0)[bl]{$\scriptstyle s_H$}}
        \put(70,5){\line(1,0){70}}
        \put(90,-25){\line(0,1){30}}
        \put(90,5){\circle*{5}}
        \put(93,7){\makebox(0,0)[bl]{$\scriptstyle s_{\oplus}$}}
        \put(90,-20){\circle{10}}
        \put(118,32){\makebox(0,0)[bl]{$\scriptstyle s_{\oplus}$}}
        \put(115,30){\circle*{5}}
        \put(115,5){\circle{10}}
        \put(115,0){\line(0,1){30}}
        \put(118,-18){\makebox(0,0)[bl]{$\scriptstyle s_{I}$}}
        \put(140,20){\line(0,1){20}}
        \put(140,30){\oval(40,20)[r]}
        \put(148,30){\makebox(0,0){$\overline{X}$}}
        \put(161,35){\makebox(0,0)[bl]{$\scriptstyle s_{X}$}}
        \put(140,-5){\line(0,1){20}}
        \put(140,5){\oval(40,20)[r]}
        \put(148,5){\makebox(0,0){$\overline{Z}$}}
        \put(161,10){\makebox(0,0)[bl]{$\scriptstyle s_{Z}$}}
        \put(153,-18){\makebox(0,0)[bl]{$\scriptstyle s_{I}$}}
    \end{picture},\label{eq:s-EC}\\
    &\nonumber \\
    &\nonumber
\end{align}
where the parameter $s$ is given as a function of $s_0$, $s_H$, $s_{\oplus}$, $s_X$, and $s_Z$ as
\begin{align}
    s \coloneqq 2s_0 + s_H + s_{\oplus} + \max\{ s_{\oplus} + \max\{s_X,s_Z\}, 2s_{I}\}. \label{eq:s_in_EC}
\end{align}
\end{definition}

Having defined parameterized gadgets, we can state a fault-tolerance (FT) condition for these gadgets, which will be used later.
The following proposition shows that the $s$-measurement works as if the incoming mode is decoded by an ideal GKP decoder followed by the logical qubit measurement as long as the total displacement $r+s$ is below $\frac{c}{2}$.  As the proof suggests, it is the homodyne detection that allows the $s$-measurement gadget to satisfy the FT condition.  However, other implementations of the measurement gadget are not prohibited in principle as long as the above condition is satisfied.
\begin{proposition}[FT condition for $s$-measurement] \label{prop:s-meas}
    The GKP $s$-measurement gadget defined in Def.~\ref{def:s-measurement} satisfies the following:
    \begin{trivlist}
        \item {Meas:} 
    $
    \begin{picture}(70,15)
    \thicklines
    
    \put(10,5){\line(1,0){10}}
    \put(20,-5){\framebox(5,20){}}
    \put(25,5){\line(1,0){15}}
    \put(27,10){\makebox(0,0)[bl]{$\scriptstyle r$}}
    \put(40,5){\oval(40,20)[r]}
    \put(40,-5){\line(0,1){20}}
    \put(61,10){\makebox(0,0)[bl]{$\scriptstyle s$}}
    \end{picture}
    =
    \begin{picture}(75,15)
    \thicklines
    \put(10,5){\line(1,0){10}}
    \put(20,-5){\framebox(5,20){}}
    \put(25,5){\line(1,0){15}}
    \put(27,10){\makebox(0,0)[bl]{$\scriptstyle r$}}
    \put(40,-5){\line(0,1){20}}
    \put(40,15){\line(1,-1){10}}
    \put(50,5){\line(-1,-1){10}}
    \thinlines
    \put(50,5){\line(1,0){10}}
    \put(60,5){\oval(20,20)[r]}
    \put(60,-5){\line(0,1){20}}
    \end{picture}$
    when $r+s < \frac{c}{2}$,
    \vspace{10pt}
    where the thin right half circle denotes a qubit measurement.
    \end{trivlist}
\end{proposition}
\begin{proof}
    We consider the case of $\overline{Z}$ measurement while the same argument holds for $\overline{X}$ measurement.  From Eq.~\eqref{eq:GKP_Z_meas}, the POVM elements of the ideal GKP $\overline{Z}$ measurement can be given by
    \begin{equation}
        \left\{\ket{0}\!\bra{0}_L\otimes\hat{I}_S, \ket{1}\!\bra{1}_L\otimes \hat{I}_S\right\}=
        \left\{\int_{-\frac{c}{2}}^{\frac{c}{2}}dz_1\int_{-\frac{c}{2}}^{\frac{c}{2}}dz_2\ket{0;z_1,z_2}\!\bra{0;z_1,z_2}_{LS},\ \int_{-\frac{c}{2}}^{\frac{c}{2}}dz_1\int_{-\frac{c}{2}}^{\frac{c}{2}}dz_2\ket{1;z_1,z_2}\!\bra{1;z_1,z_2}_Z\right\}.
    \end{equation}
    Therefore, for any density operator $\hat{\rho}$ and for any noise channel ${\cal N}^s$ whose Kraus operators are linear combinations of elements in $\{\hat{V}(z_1,z_2):(z_1,z_2)\in [-s,s)\times [-s,s)\}$, we have, for $\mu\in\{0,1\}$,
    \begin{align}
        &\Tr \left[(\ket{\mu}\!\bra{\mu}_L\otimes\hat{I}_S)\mathcal{N}^s(\hat{\Pi}_r\hat{\rho}\hat{\Pi}_r)\right] \nonumber\\
        &=\Tr \left[(\ket{\mu}\!\bra{\mu}_L\otimes\hat{I}_S)\;\mathcal{N}^s\!\left(\sum_{\nu=0,1}\iint_{|z_1|,|z_2|<r} dz_1 dz_2 \ket{\nu;z_1,z_2}\!\bra{\nu;z_1,z_2}\hat{\rho}\sum_{\nu'=0,1}\iint_{|z'_1|,|z'_2|<r} dz'_1 dz'_2 \ket{\nu';z'_1,z'_2}\!\bra{\nu';z'_1,z'_2}\right)\right] \\
        \begin{split}
        &= \Tr \biggl[(\ket{\mu}\!\bra{\mu}_L\otimes\hat{I}_S)\;\int_{\mathbb{R}^4}d^2w d^2w'\;{\cal E}_s(w_1,w_2,w'_1,w'_2) \\
        &\qquad \qquad \sum_{\nu,\nu'=0,1}\iiiint_{|z_1|,|z_2|,|z'_1|,|z'_2|<r} dz_1 dz_2 dz'_1 dz'_2\; \hat{V}(w_1,w_2)\ket{\nu;z_1,z_2}\!\bra{\nu;z_1,z_2}\hat{\rho}  \ket{\nu';z'_1,z'_2}\!\bra{\nu';z'_1,z'_2}\hat{V}(w'_1,w'_2)^{\dagger}\biggr] 
        \end{split}\\
        \begin{split}
        &= \Tr \biggl[(\ket{\mu}\!\bra{\mu}_L\otimes\hat{I}_S)\;\int_{\mathbb{R}^4}d^2w d^2w'\;{\cal E}_s(w_1,w_2,w'_1,w'_2)e^{i(w_2z_1-w_1z_2)/2-i(w'_2z'_1-w'_1z'_2)/2} \\
        &\qquad \qquad \sum_{\nu,\nu'=0,1}\iiiint_{|z_1|,|z_2|,|z'_1|,|z'_2|<r} dz_1 dz_2 dz'_1 dz'_2\; \ket{\nu;z_1+w_1,z_2+w_2}\!\bra{\nu;z_1,z_2}\hat{\rho}  \ket{\nu';z'_1,z'_2}\!\bra{\nu';z'_1+w'_1,z'_2+w'_2}\biggr] ,
        \end{split} \label{eq:intermediate}
        \end{align}
        where we used Eq.~\eqref{eq:def_N_s} in the second equality.  Recall that ${\cal E}_s(z'_1,z'_2,z''_1,z''_2)$ has support only in $\prod_{i=1}^4[-s,s)$, and thus the right-hand side of \eqref{eq:intermediate} has nonzero value only for $|z_i+w_i| \leq r+s < \frac{c}{2}$ and $|z'_i+w'_i| \leq r+s < \frac{c}{2}$.  Therefore, we have
        \begin{align}
        \begin{split}
        \eqref{eq:intermediate} &= \Tr \biggl[\int_{\mathbb{R}^4}d^2w d^2w'\;{\cal E}_s(w_1,w_2,w'_1,w'_2)e^{i(w_2z_1-w_1z_2)/2-i(w'_2z'_1-w'_1z'_2)/2} \\
        &\qquad \qquad \iiiint_{|z_1|,|z_2|,|z'_1|,|z'_2|<r} dz_1 dz_2 dz'_1 dz'_2\; \ket{\mu;z_1+w_1,z_2+w_2}\!\bra{\mu;z_1,z_2}\hat{\rho}  \ket{\mu;z'_1,z'_2}\!\bra{\mu;z'_1+w'_1,z'_2+w'_2}\biggr] 
        \end{split} \\
        \begin{split}
            &= \Tr \biggl[\int_{\mathbb{R}^4}d^2w d^2w'\;{\cal E}_s(w_1,w_2,w'_1,w'_2)\\
            &\qquad \qquad \iiiint_{|z_1|,|z_2|,|z'_1|,|z'_2|<r} dz_1 dz_2 dz'_1 dz'_2\; \hat{V}(w_1,w_2)\ket{\mu;z_1,z_2}\!\bra{\mu;z_1,z_2}\hat{\rho}  \ket{\mu;z'_1,z'_2}\!\bra{\mu;z'_1,z'_2} \hat{V}(w'_1,w'_2)^{\dagger}\biggr] 
        \end{split} \\
        &= \Tr \biggl[{\cal N}^s\left(\iiiint_{|z_1|,|z_2|,|z'_1|,|z'_2|<r} dz_1 dz_2 dz'_1 dz'_2\; \ket{\mu;z_1,z_2}\!\bra{\mu;z_1,z_2}\hat{\rho}  \ket{\mu;z'_1,z'_2}\!\bra{\mu;z'_1,z'_2} \right)\biggr] \\
        &= \Tr[\bra{\mu}_L\hat{\Pi}_r\hat{\rho}\hat{\Pi}_r \ket{\mu}_L] \\
        &= \bra{\mu}_L\Tr_S[\hat{\Pi}_r\hat{\rho}\hat{\Pi}_r]\ket{\mu}_L,
    \end{align}
    where we used the fact that ${\cal N}^s$ is trace-preserving in the second last equality.
    This proves the statement.
\end{proof}
Next, we state the FT condition for the preparation gadget.

\begin{proposition}[FT condition for $s$-preparation] \label{prop:s-preparation}
    The GKP $s$-preparation gadget defined in Def.~\ref{def:s-preparation} satisfies the following conditions.
    \begin{trivlist}
        \item {Prep A:} 
        $
        \begin{picture}(55,20)
    \thicklines
    \put(30,5){\oval(40,20)[l]}
    \put(30,-5){\line(0,1){20}}
    \put(32,10){\makebox(0,0)[bl]{$\scriptstyle s$}}
    \put(30,5){\line(1,0){20}}
    \end{picture}
        =
        \begin{picture}(75,20)
    \thicklines
    \put(30,5){\oval(40,20)[l]}
    \put(30,-5){\line(0,1){20}}
    \put(32,10){\makebox(0,0)[bl]{$\scriptstyle s$}}
    \put(30,5){\line(1,0){15}}
    \put(45,-5){\framebox(5,20){}}
    \put(50,5){\line(1,0){20}}
    \put(52,10){\makebox(0,0)[bl]{$\scriptstyle s$}}
    \end{picture}$
    when $s< \frac{c}{2}$.
    \item {Prep B:} 
    \begin{picture}(70,25)
    \thicklines
    \put(30,5){\oval(40,20)[l]}
    \put(30,-5){\line(0,1){20}}
    \put(32,10){\makebox(0,0)[bl]{$\scriptstyle s$}}
    \put(30,5){\line(1,0){15}}
    \put(45,-5){\line(0,1){20}}
    \put(45,15){\line(1,-1){10}}
    \put(55,5){\line(-1,-1){10}}
    \thinlines
    \put(55,5){\line(1,0){10}}
    \end{picture}
    = 
    \begin{picture}(45,20)
    \thinlines
    \put(20,5){\oval(20,20)[l]}
    \put(20,-5){\line(0,1){20}}
    \put(20,5){\line(1,0){15}}
    \end{picture}
    \vspace{10pt} when $s< \frac{c}{2}$, where the thin left half circle denotes a qubit state preparation.
    \end{trivlist}
\end{proposition}
\begin{proof}
    It suffices to show that the mentioned property holds for any $i$ in the mixture of Eq.~\eqref{eq:s-parameterized_mixed}.
    When $s< \frac{c}{2}$, the $s$-parameterized state defined in Eq.~\eqref{eq:s-parameterized_pure} can be rewritten as
    \begin{equation}
        \ket{\rho_i^s[\psi]}=\ket{\psi}_L\otimes \int_{-\frac{c}{2}}^{\frac{c}{2}}dz_1\int_{-\frac{c}{2}}^{\frac{c}{2}}dz_2\;f^i_s(z_1, z_2)\ket{z_1,z_2}_S.
    \end{equation}
    From this and the definition of the SSS $r$-filter in Eq.~\eqref{eq:SSS_r-filter}, we have $\hat{\Pi}_s\ket{\rho_i^s[\psi]}=\ket{\rho_i^s[\psi]}$ when $s< \frac{c}{2}$, which proves Prep A.  Prep B also follows from the same argument.
\end{proof}

We then define an FT condition for the $s$-gate, which is done similarly to the previous gadgets.
\begin{proposition}[FT condition for $s$-gate] \label{prop:s-gate}
    The GKP $s$-gate gadget defined in Def.~\ref{def:s-gate} satisfies the following conditions:
    \begin{trivlist}
        \item {Gate A:} 
        $
        \begin{picture}(95,20)
    \thicklines
    \put(10,5){\line(1,0){10}}
    \put(20,-5){\framebox(5,20){}}
    \put(25,5){\line(1,0){13}}
    \put(27,10){\makebox(0,0)[bl]{$\scriptstyle r_i$}}
    \put(50,5){\circle{24}}
    \put(50,5){\makebox(0,0){$\overline{U}$}}
    \put(62,5){\line(1,0){25}}
    \put(62,10){\makebox(0,0)[bl]{$\scriptstyle s$}}
    \end{picture}
        =
        \begin{picture}(130,20)
    \thicklines
    \put(10,5){\line(1,0){10}}
    \put(20,-5){\framebox(5,20){}}
    \put(25,5){\line(1,0){13}}
    \put(27,10){\makebox(0,0)[bl]{$\scriptstyle r_i$}}
    \put(50,5){\circle{24}}
    \put(50,5){\makebox(0,0){$\overline{U}$}}
    \put(62,5){\line(1,0){15}}
    \put(62,10){\makebox(0,0)[bl]{$\scriptstyle s$}}
    \put(77,-5){\framebox(5,20){}}
    \put(82,5){\line(1,0){33}}
    \put(84,10){\makebox(0,0)[bl]{$\scriptstyle \sum_i r_i + s$}}
    \end{picture}$
    when $\sum_i r_i + s < \frac{c}{2}$.
    \item{Gate B:}
    $\begin{picture}(110,20)
    \thicklines
    \put(10,5){\line(1,0){10}}
    \put(20,-5){\framebox(5,20){}}
    \put(25,5){\line(1,0){13}}
    \put(27,10){\makebox(0,0)[bl]{$\scriptstyle r_i$}}
    \put(50,5){\circle{24}}
    \put(50,5){\makebox(0,0){$\overline{U}$}}
    \put(62,5){\line(1,0){15}}
    \put(62,10){\makebox(0,0)[bl]{$\scriptstyle s$}}
    \put(77,-5){\line(0,1){20}}
    \put(77,15){\line(1,-1){10}}
    \put(87,5){\line(-1,-1){10}}
    \thinlines
    \put(87,5){\line(1,0){10}}
    \end{picture}
    = 
    \begin{picture}(110,20)
    \thicklines
    \put(10,5){\line(1,0){10}}
    \put(20,-5){\framebox(5,20){}}
    \put(25,5){\line(1,0){15}}
    \put(27,10){\makebox(0,0)[bl]{$\scriptstyle r_i$}}
    \put(40,-5){\line(0,1){20}}
    \put(40,15){\line(1,-1){10}}
    \put(50,5){\line(-1,-1){10}}
    \thinlines
    \put(50,5){\line(1,0){15}}
    \put(75,5){\circle{20}}
    \put(75,5){\makebox(0,0){$U$}}
    \put(85,5){\line(1,0){15}}
    \end{picture}
    $ \vspace{10pt} when $\sum_i r_i + s < \frac{c}{2}$, where $i$ runs over all the modes that this gate gadget acts on, and the thin circle denotes a qubit unitary.
    \end{trivlist}
\end{proposition}
\begin{proof}
    We only need to check this condition for our gate set $\overline{Z}$, $\overline{X}$, $\overline{H}$, and CNOT.
    From Eqs.~\eqref{eq:logical_X}--\eqref{eq:rotation_SSS}, we have that when $r<\frac{c}{2}$,
    \begin{align}
        \hat{V}(0,c) \hat{\Pi}_r \hat{V}(0,c)^{\dagger} &= \hat{\Pi}_r,\\
        \hat{V}(c,0) \hat{\Pi}_r \hat{V}(c,0)^{\dagger} &= \hat{\Pi}_r, \\
        \hat{F} \hat{\Pi}_r \hat{F}^{\dagger} &= \hat{\Pi}_r,
    \end{align}
    where we used that the SSS $r$-filter is invariant under sign changes of $z_1$ and $z_2$ inside the integral.  For the SUM gate, we have from Eq.~\eqref{eq:CX} that when $r_j+r_k<\frac{c}{2}$,
    \begin{align}
        \begin{split}
        \exp(-i\hat{q}_j\hat{p}_k)(\hat{\Pi}_{r_j}\otimes\hat{\Pi}_{r_k})\exp(i\hat{q}_j\hat{p}_k) &= \int_{|z_1|<r_j} dz_1 \int_{|z_2|<r_j}dz_2 \int_{|z'_1|<r_k}dz'_1 \int_{|z'_2|<r_k}dz'_2\, (\hat{I}_L\otimes \ket{z_1,z_2-z'_2}\!\bra{z_1,z_2-z'_2}_{S})_j  \\
        &\hspace{7cm} \otimes (\hat{I}_L\otimes \ket{z'_1+z_1,z'_2}\!\bra{z'_1+z_1,z'_2}_{S})_k 
        \end{split}\\
        \begin{split}
        &\leq \int_{|z_1|<r_j} dz_1 \int_{|z_2|<r_j+r_k}dz_2 \int_{|z'_1|<r_k+r_j}dz'_1 \int_{|z'_2|<r_k}dz'_2\, (\hat{I}_L\otimes \ket{z_1,z_2}\!\bra{z_1,z_2}_{S})_j \\
        &\hspace{7cm} \otimes (\hat{I}_L\otimes \ket{z'_1,z'_2}\!\bra{z'_1,z'_2}_{S})_k
        \end{split} \\
        &\leq \hat{\Pi}_{r_j+r_k}\otimes\hat{\Pi}_{r_k+r_j}.
    \end{align}
    Therefore, when $r_j+r_k<\frac{c}{2}$, we have
    \begin{align}
        \exp(-i\hat{q}_j\hat{p}_k)(\hat{\Pi}_{r_j}\otimes\hat{\Pi}_{r_k})&= \exp(-i\hat{q}_j\hat{p}_k)(\hat{\Pi}_{r_j}\otimes\hat{\Pi}_{r_k})\exp(i\hat{q}_j\hat{p}_k) \exp(-i\hat{q}_j\hat{p}_k)(\hat{\Pi}_{r_j}\otimes\hat{\Pi}_{r_k})\\& = \hat{\Pi}_{r_j+r_k}\otimes\hat{\Pi}_{r_k+r_j} \exp(-i\hat{q}_j\hat{p}_k)(\hat{\Pi}_{r_j}\otimes\hat{\Pi}_{r_k}),
    \end{align}
    where we used the fact that $\hat{P}\hat{Q}=\hat{Q}$ when $\hat{P}\geq \hat{Q}$ for two projection operators $\hat{P}$ and $\hat{Q}$.
    With these and the definition of the $s$-gate in Eq.~\eqref{eq:s-gate}, the condition Gate A holds for all gate gadgets as long as $\sum_{j}r_j+s<\frac{c}{2}$.  The condition Gate B also holds from the above and Eqs.~\eqref{eq:logical_X}, \eqref{eq:logical_Z}, \eqref{eq:fourier_subsystem}, and \eqref{eq:CX}.
\end{proof}

It is the choice of the physical CV gates implementing these gadgets that allows them to satisfy the FT conditions.  In other words, other choices of the CV physical gates to implement the gate gadgets are allowed as long as the constructed gate gadgets satisfy the above FT conditions.
Here, we can explain why we do not use the conventional CV shear gate $\exp(i\hat{q}^2/2)$ to implement the GKP logical phase gate: the GKP phase gate composed of the CV shear gate will satisfy neither Gate A nor Gate B.  To include $\exp(i\hat{q}^2/2)$ as a fault-tolerant gadget, we need to modify the FT conditions so that $\exp(i\hat{q}^2/2)$ can satisfy them.  
Concretely, we should modify conditions such that the last SSS $r$-filter $\hat{\Pi}_{r+s}$ in the right-hand side of Gate A is replaced with $\hat{\Pi}_{2r+s}$ for a single-mode gate, and both Gate A and Gate B equations only hold when $2r+s<\frac{c}{2}$ for a single-mode gate.  To avoid the complication, however, we stick to our gate set introduced in Sec.~\ref{sec:physical_operations} hereafter.

We also remark that a non-Gaussian unitary gate does not satisfy the above criteria in general.  This is because non-Gaussian unitary transforms a displacement to a Gaussian unitary, which then requires infinitely many displacement operators to approximate it with a linear combination.  In terms of the stabilizer subsystem decomposition, this means that even a \quoted{wave function} on the size-$\sqrt{\pi}$ Cartesian square of the syndrome subsystem that is supported in a small neighborhood of the origin might be spread across the whole region in the Cartesian square by the action of non-Gaussian unitary, which may not be correctable.  Perhaps, there may exist a class of non-Gaussian gates whose effects are limited enough to preserve fault tolerance, but we leave this to future work.
To avoid these complications, we implement a GKP non-Clifford gate via teleportation of a magic state \cite{Gottesman2001, Baragiola2019, Yamasaki2020} in our analysis.  The FT condition for a non-Clifford gate is therefore reduced to the FT condition for a magic state preparation explained previously.

Finally, the FT condition for the $s$-EC gadget is given in the following.
\begin{proposition}[FT condition for $s$-EC]\label{prop:s-EC}
    The GKP $s$-EC gadget satisfies the following conditions:
    \begin{trivlist}
        \item {EC A:} 
        $
        \begin{picture}(80,20)
    \thicklines
    \put(10,5){\line(1,0){15}}
    \put(25,-5){\framebox(30,20){\text{EC}}}
    \put(55,5){\line(1,0){15}}
    \put(57,10){\makebox(0,0)[bl]{$\scriptstyle s$}}
    \end{picture}
        =
        \begin{picture}(100,20)
    \thicklines
    \put(10,5){\line(1,0){15}}
    \put(25,-5){\framebox(30,20){\text{EC}}}
    \put(55,5){\line(1,0){15}}
    \put(57,10){\makebox(0,0)[bl]{$\scriptstyle s$}}
    \put(70,-5){\framebox(5,20){}}
    \put(75,5){\line(1,0){15}}
    \put(77,10){\makebox(0,0)[bl]{$\scriptstyle s$}}
    \end{picture}$
    when $s< \frac{c}{2}$.
    \item{EC B:}
    $\begin{picture}(115,20)
    \thicklines
    \put(10,5){\line(1,0){10}}
    \put(20,-5){\framebox(5,20){}}
    \put(25,5){\line(1,0){15}}
    \put(27,10){\makebox(0,0)[bl]{$\scriptstyle r$}}
    \put(40,-5){\framebox(30,20){\text{EC}}}
    \put(70,5){\line(1,0){15}}
    \put(72,10){\makebox(0,0)[bl]{$\scriptstyle s$}}
    \put(85,-5){\line(0,1){20}}
    \put(85,15){\line(1,-1){10}}
    \put(95,5){\line(-1,-1){10}}
    \thinlines
    \put(95,5){\line(1,0){10}}
    \end{picture}
    = 
    \begin{picture}(70,20)
    \thicklines
    \put(10,5){\line(1,0){10}}
    \put(20,-5){\framebox(5,20){}}
    \put(25,5){\line(1,0){15}}
    \put(27,10){\makebox(0,0)[bl]{$\scriptstyle r$}}
    \put(40,-5){\line(0,1){20}}
    \put(40,15){\line(1,-1){10}}
    \put(50,5){\line(-1,-1){10}}
    \thinlines
    \put(50,5){\line(1,0){10}}
    \end{picture}
    $ when $r + s< \frac{c}{2}$.
    \vspace{10pt}
    \end{trivlist}
\end{proposition}
\begin{proof}
    Due to our construction of the $s$-EC gadget in Eq.~\eqref{eq:s-EC} and Proposition~\ref{prop:s-meas}--\ref{prop:s-gate}, we have the following chain of equalities when $2s_0+s_H + s_{\oplus} + 2s_I < \frac{c}{2}$.
    \begin{align}
        &\begin{picture}(260,40)
    \thicklines
        \put(90,30){\line(1,0){50}}
        \put(30,5){\oval(40,20)[l]}
        \put(30,-5){\line(0,1){20}}
        \put(15,5){\makebox(0,0)[l]{$\ket{\overline{0}}$}}
        \put(32,10){\makebox(0,0)[bl]{$\scriptstyle s_0$}}
        \put(70,-20){\oval(40,20)[l]}
        \put(70,-30){\line(0,1){20}}
        \put(55,-20){\makebox(0,0)[l]{$\ket{\overline{0}}$}}
        \put(72,-15){\makebox(0,0)[bl]{$\scriptstyle s_0$}}
        \put(30,5){\line(1,0){20}}
        \put(70,-20){\line(1,0){110}}
        \put(60,5){\circle{20}}
        \put(60,5){\makebox(0,0){$\overline{H}$}}
        \put(71,10){\makebox(0,0)[bl]{$\scriptstyle s_H$}}
        \put(70,5){\line(1,0){70}}
        \put(90,-25){\line(0,1){30}}
        \put(90,5){\circle*{5}}
        \put(93,7){\makebox(0,0)[bl]{$\scriptstyle s_{\oplus}$}}
        \put(117,-18){\makebox(0,0)[bl]{$\scriptstyle s_{I}$}}
        \put(90,-20){\circle{10}}
        \put(118,32){\makebox(0,0)[bl]{$\scriptstyle s_{\oplus}$}}
        \put(115,30){\circle*{5}}
        \put(115,5){\circle{10}}
        \put(115,0){\line(0,1){30}}
        \put(140,20){\line(0,1){20}}
        \put(140,30){\oval(40,20)[r]}
        \put(148,30){\makebox(0,0){$\overline{X}$}}
        \put(161,35){\makebox(0,0)[bl]{$\scriptstyle s_{X}$}}
        \put(140,-5){\line(0,1){20}}
        \put(140,5){\oval(40,20)[r]}
        \put(148,5){\makebox(0,0){$\overline{Z}$}}
        \put(161,10){\makebox(0,0)[bl]{$\scriptstyle s_{Z}$}}
        \put(152,-18){\makebox(0,0)[bl]{$\scriptstyle s_{I}$}}
    \end{picture}\\
    & =
    \begin{picture}(260,80)
    \thicklines
        \put(105,30){\line(1,0){50}}
        \put(30,5){\oval(40,20)[l]}
        \put(30,-5){\line(0,1){20}}
        \put(15,5){\makebox(0,0)[l]{$\ket{\overline{0}}$}}
        \put(32,10){\makebox(0,0)[bl]{$\scriptstyle s_0$}}
        \put(70,-20){\oval(40,20)[l]}
        \put(70,-30){\line(0,1){20}}
        \put(55,-20){\makebox(0,0)[l]{$\ket{\overline{0}}$}}
        \put(72,-15){\makebox(0,0)[bl]{$\scriptstyle s_0$}}
        \put(30,5){\line(1,0){15}}
        \put(45,-5){\framebox(5,20)}
        \put(50,5){\line(1,0){15}}
        \put(52,10){\makebox(0,0)[bl]{$\scriptstyle s_0$}}
        \put(70,-20){\line(1,0){15}}
        \put(85,-30){\framebox(5,20)}
        \put(92,-15){\makebox(0,0)[bl]{$\scriptstyle s_0$}}
        \put(90,-20){\line(1,0){105}}
        \put(75,5){\circle{20}}
        \put(75,5){\makebox(0,0){$\overline{H}$}}
        \put(86,10){\makebox(0,0)[bl]{$\scriptstyle s_H$}}
        \put(85,5){\line(1,0){70}}
        \put(105,-25){\line(0,1){30}}
        \put(105,5){\circle*{5}}
        \put(108,7){\makebox(0,0)[bl]{$\scriptstyle s_{\oplus}$}}
        \put(105,-20){\circle{10}}
        \put(133,32){\makebox(0,0)[bl]{$\scriptstyle s_{\oplus}$}}
        \put(132,-18){\makebox(0,0)[bl]{$\scriptstyle s_{I}$}}
        \put(130,30){\circle*{5}}
        \put(130,5){\circle{10}}
        \put(130,0){\line(0,1){30}}
        \put(155,20){\line(0,1){20}}
        \put(155,30){\oval(40,20)[r]}
        \put(163,30){\makebox(0,0){$\overline{X}$}}
        \put(176,35){\makebox(0,0)[bl]{$\scriptstyle s_{X}$}}
        \put(155,-5){\line(0,1){20}}
        \put(155,5){\oval(40,20)[r]}
        \put(163,5){\makebox(0,0){$\overline{Z}$}}
        \put(176,10){\makebox(0,0)[bl]{$\scriptstyle s_{Z}$}}
        \put(167,-18){\makebox(0,0)[bl]{$\scriptstyle s_{I}$}}
    \end{picture}\\
    & =
    \begin{picture}(260,80)
    \thicklines
        \put(135,30){\line(1,0){50}}
        \put(30,5){\oval(40,20)[l]}
        \put(30,-5){\line(0,1){20}}
        \put(15,5){\makebox(0,0)[l]{$\ket{\overline{0}}$}}
        \put(32,10){\makebox(0,0)[bl]{$\scriptstyle s_0$}}
        \put(85,-20){\oval(40,20)[l]}
        \put(85,-30){\line(0,1){20}}
        \put(70,-20){\makebox(0,0)[l]{$\ket{\overline{0}}$}}
        \put(87,-15){\makebox(0,0)[bl]{$\scriptstyle s_0$}}
        \put(30,5){\line(1,0){15}}
        \put(85,-20){\line(1,0){15}}
        \put(45,-5){\framebox(5,20)}
        \put(50,5){\line(1,0){15}}
        \put(52,10){\makebox(0,0)[bl]{$\scriptstyle s_0$}}
        \put(100,-30){\framebox(5,20)}
        \put(107,-15){\makebox(0,0)[bl]{$\scriptstyle s_0$}}
        \put(105,-20){\line(1,0){120}}
        \put(75,5){\circle{20}}
        \put(75,5){\makebox(0,0){$\overline{H}$}}
        \put(86,10){\makebox(0,0)[bl]{$\scriptstyle s_H$}}
        \put(85,5){\line(1,0){15}}
        \put(100,-5){\framebox(5,20)}
        \put(107,10){\makebox(0,0)[bl]{$\scriptstyle s_0+s_H$}}
        \put(105,5){\line(1,0){80}}
        \put(135,-25){\line(0,1){30}}
        \put(135,5){\circle*{5}}
        \put(138,7){\makebox(0,0)[bl]{$\scriptstyle s_{\oplus}$}}
        \put(135,-20){\circle{10}}
        \put(163,32){\makebox(0,0)[bl]{$\scriptstyle s_{\oplus}$}}
        \put(160,30){\circle*{5}}
        \put(162,-18){\makebox(0,0)[bl]{$\scriptstyle s_{I}$}}
        \put(160,5){\circle{10}}
        \put(160,0){\line(0,1){30}}
        \put(185,20){\line(0,1){20}}
        \put(185,30){\oval(40,20)[r]}
        \put(193,30){\makebox(0,0){$\overline{X}$}}
        \put(206,35){\makebox(0,0)[bl]{$\scriptstyle s_{X}$}}
        \put(185,-5){\line(0,1){20}}
        \put(185,5){\oval(40,20)[r]}
        \put(193,5){\makebox(0,0){$\overline{Z}$}}
        \put(206,10){\makebox(0,0)[bl]{$\scriptstyle s_{Z}$}}
        \put(197,-18){\makebox(0,0)[bl]{$\scriptstyle s_{I}$}}
    \end{picture}\\
        & =
    \begin{picture}(270,80)
    \thicklines
        \put(185,30){\line(1,0){50}}
        \put(30,5){\oval(40,20)[l]}
        \put(30,-5){\line(0,1){20}}
        \put(15,5){\makebox(0,0)[l]{$\ket{\overline{0}}$}}
        \put(32,10){\makebox(0,0)[bl]{$\scriptstyle s_0$}}
        \put(85,-20){\oval(40,20)[l]}
        \put(85,-30){\line(0,1){20}}
        \put(70,-20){\makebox(0,0)[l]{$\ket{\overline{0}}$}}
        \put(87,-15){\makebox(0,0)[bl]{$\scriptstyle s_0$}}
        \put(30,5){\line(1,0){15}}
        \put(85,-20){\line(1,0){15}}
        \put(30,5){\line(1,0){15}}
        \put(85,-20){\line(1,0){15}}
        \put(45,-5){\framebox(5,20)}
        \put(50,5){\line(1,0){15}}
        \put(52,10){\makebox(0,0)[bl]{$\scriptstyle s_0$}}
        \put(100,-30){\framebox(5,20)}
        \put(107,-15){\makebox(0,0)[bl]{$\scriptstyle s_0$}}
        \put(105,-20){\line(1,0){50}}
        \put(75,5){\circle{20}}
        \put(75,5){\makebox(0,0){$\overline{H}$}}
        \put(86,10){\makebox(0,0)[bl]{$\scriptstyle s_H$}}
        \put(85,5){\line(1,0){15}}
        \put(100,-5){\framebox(5,20)}
        \put(107,10){\makebox(0,0)[bl]{$\scriptstyle s_0+s_H$}}
        \put(105,5){\line(1,0){50}}
        \put(135,-25){\line(0,1){30}}
        \put(135,5){\circle*{5}}
        \put(138,7){\makebox(0,0)[bl]{$\scriptstyle s_{\oplus}$}}
        \put(135,-20){\circle{10}}
        \put(155,-30){\framebox(5,20)}
        \put(162,-15){\makebox(0,0)[bl]{$\scriptstyle 2s_0 + s_H + s_{\oplus}$}}
        \put(155,-5){\framebox(5,20)}
        \put(160,5){\line(1,0){75}}
        \put(160,-20){\line(1,0){110}}
        \put(162,10){\makebox(0,0)[bl]{$\scriptstyle 2s_0 + s_H + s_{\oplus}$}}
        \put(213,32){\makebox(0,0)[bl]{$\scriptstyle s_{\oplus}$}}
        \put(210,30){\circle*{5}}
        \put(212,-18){\makebox(0,0)[bl]{$\scriptstyle s_{I}$}}
        \put(210,5){\circle{10}}
        \put(210,0){\line(0,1){30}}
        \put(235,20){\line(0,1){20}}
        \put(235,30){\oval(40,20)[r]}
        \put(243,30){\makebox(0,0){$\overline{X}$}}
        \put(256,35){\makebox(0,0)[bl]{$\scriptstyle s_{X}$}}
        \put(235,-5){\line(0,1){20}}
        \put(235,5){\oval(40,20)[r]}
        \put(243,5){\makebox(0,0){$\overline{Z}$}}
        \put(256,10){\makebox(0,0)[bl]{$\scriptstyle s_{Z}$}}
        \put(247,-18){\makebox(0,0)[bl]{$\scriptstyle s_{I}$}}
    \end{picture}\\
    &=
    \begin{picture}(270,80)
    \thicklines
        \put(185,30){\line(1,0){50}}
        \put(30,5){\oval(40,20)[l]}
        \put(30,-5){\line(0,1){20}}
        \put(15,5){\makebox(0,0)[l]{$\ket{\overline{0}}$}}
        \put(32,10){\makebox(0,0)[bl]{$\scriptstyle s_0$}}
        \put(85,-20){\oval(40,20)[l]}
        \put(85,-30){\line(0,1){20}}
        \put(70,-20){\makebox(0,0)[l]{$\ket{\overline{0}}$}}
        \put(87,-15){\makebox(0,0)[bl]{$\scriptstyle s_0$}}
        \put(30,5){\line(1,0){15}}
        \put(85,-20){\line(1,0){15}}
        \put(30,5){\line(1,0){15}}
        \put(85,-20){\line(1,0){15}}
        \put(45,-5){\framebox(5,20)}
        \put(50,5){\line(1,0){15}}
        \put(52,10){\makebox(0,0)[bl]{$\scriptstyle s_0$}}
        \put(100,-30){\framebox(5,20)}
        \put(107,-15){\makebox(0,0)[bl]{$\scriptstyle s_0$}}
        \put(105,-20){\line(1,0){50}}
        \put(75,5){\circle{20}}
        \put(75,5){\makebox(0,0){$\overline{H}$}}
        \put(86,10){\makebox(0,0)[bl]{$\scriptstyle s_H$}}
        \put(85,5){\line(1,0){15}}
        \put(100,-5){\framebox(5,20)}
        \put(107,10){\makebox(0,0)[bl]{$\scriptstyle s_0+s_H$}}
        \put(105,5){\line(1,0){50}}
        \put(135,-25){\line(0,1){30}}
        \put(135,5){\circle*{5}}
        \put(138,7){\makebox(0,0)[bl]{$\scriptstyle s_{\oplus}$}}
        \put(135,-20){\circle{10}}
        \put(155,-30){\framebox(5,20)}
        \put(162,-15){\makebox(0,0)[bl]{$\scriptstyle 2s_0 + s_H + s_{\oplus}$}}
        \put(155,-5){\framebox(5,20)}
        \put(160,5){\line(1,0){75}}
        \put(160,-20){\line(1,0){75}}
        \put(162,10){\makebox(0,0)[bl]{$\scriptstyle 2s_0 + s_H + s_{\oplus}$}}
        \put(213,32){\makebox(0,0)[bl]{$\scriptstyle s_{\oplus}$}}
        \put(210,30){\circle*{5}}
        \put(212,-18){\makebox(0,0)[bl]{$\scriptstyle s_{I}$}}
        \put(210,5){\circle{10}}
        \put(210,0){\line(0,1){30}}
        \put(235,20){\line(0,1){20}}
        \put(235,30){\oval(40,20)[r]}
        \put(243,30){\makebox(0,0){$\overline{X}$}}
        \put(256,35){\makebox(0,0)[bl]{$\scriptstyle s_{X}$}}
        \put(235,-5){\line(0,1){20}}
        \put(235,5){\oval(40,20)[r]}
        \put(243,5){\makebox(0,0){$\overline{Z}$}}
        \put(256,10){\makebox(0,0)[bl]{$\scriptstyle s_{Z}$}}
        \put(235,-30){\framebox(5,20)}
        \put(240,-20){\line(1,0){80}}
        \put(242,-15){\makebox(0,0)[bl]{$\scriptstyle 2s_0 + s_H + s_{\oplus}+s_I$}}
        \put(302,-18){\makebox(0,0)[bl]{$\scriptstyle s_I$}}
    \end{picture}\\
    &=
    \begin{picture}(360,80)
    \thicklines
        \put(185,30){\line(1,0){50}}
        \put(30,5){\oval(40,20)[l]}
        \put(30,-5){\line(0,1){20}}
        \put(15,5){\makebox(0,0)[l]{$\ket{\overline{0}}$}}
        \put(32,10){\makebox(0,0)[bl]{$\scriptstyle s_0$}}
        \put(85,-20){\oval(40,20)[l]}
        \put(85,-30){\line(0,1){20}}
        \put(70,-20){\makebox(0,0)[l]{$\ket{\overline{0}}$}}
        \put(87,-15){\makebox(0,0)[bl]{$\scriptstyle s_0$}}
        \put(30,5){\line(1,0){15}}
        \put(85,-20){\line(1,0){15}}
        \put(30,5){\line(1,0){15}}
        \put(85,-20){\line(1,0){15}}
        \put(45,-5){\framebox(5,20)}
        \put(50,5){\line(1,0){15}}
        \put(52,10){\makebox(0,0)[bl]{$\scriptstyle s_0$}}
        \put(100,-30){\framebox(5,20)}
        \put(107,-15){\makebox(0,0)[bl]{$\scriptstyle s_0$}}
        \put(105,-20){\line(1,0){50}}
        \put(75,5){\circle{20}}
        \put(75,5){\makebox(0,0){$\overline{H}$}}
        \put(86,10){\makebox(0,0)[bl]{$\scriptstyle s_H$}}
        \put(85,5){\line(1,0){15}}
        \put(100,-5){\framebox(5,20)}
        \put(107,10){\makebox(0,0)[bl]{$\scriptstyle s_0+s_H$}}
        \put(105,5){\line(1,0){50}}
        \put(135,-25){\line(0,1){30}}
        \put(135,5){\circle*{5}}
        \put(138,7){\makebox(0,0)[bl]{$\scriptstyle s_{\oplus}$}}
        \put(135,-20){\circle{10}}
        \put(155,-30){\framebox(5,20)}
        \put(162,-15){\makebox(0,0)[bl]{$\scriptstyle 2s_0 + s_H + s_{\oplus}$}}
        \put(155,-5){\framebox(5,20)}
        \put(160,5){\line(1,0){75}}
        \put(160,-20){\line(1,0){75}}
        \put(162,10){\makebox(0,0)[bl]{$\scriptstyle 2s_0 + s_H + s_{\oplus}$}}
        \put(213,32){\makebox(0,0)[bl]{$\scriptstyle s_{\oplus}$}}
        \put(210,30){\circle*{5}}
        \put(212,-18){\makebox(0,0)[bl]{$\scriptstyle s_{I}$}}
        \put(210,5){\circle{10}}
        \put(210,0){\line(0,1){30}}
        \put(235,20){\line(0,1){20}}
        \put(235,30){\oval(40,20)[r]}
        \put(243,30){\makebox(0,0){$\overline{X}$}}
        \put(256,35){\makebox(0,0)[bl]{$\scriptstyle s_{X}$}}
        \put(235,-5){\line(0,1){20}}
        \put(235,5){\oval(40,20)[r]}
        \put(243,5){\makebox(0,0){$\overline{Z}$}}
        \put(256,10){\makebox(0,0)[bl]{$\scriptstyle s_{Z}$}}
        \put(235,-30){\framebox(5,20)}
        \put(240,-20){\line(1,0){75}}
        \put(242,-15){\makebox(0,0)[bl]{$\scriptstyle 2s_0 + s_H + s_{\oplus}+s_I$}}
        \put(302,-18){\makebox(0,0)[bl]{$\scriptstyle s_I$}}
        \put(315,-30){\framebox(5,20)}
        \put(320,-20){\line(1,0){20}}
        \put(322,-15){\makebox(0,0)[bl]{$\scriptstyle 2s_0 + s_H + s_{\oplus}+2s_I$}}
    \end{picture}\label{eq:before_final_CNOT}\\
    &\nonumber \rule{0pt}{30pt}
    \end{align}
Since $2s_0+s_H+s_{\oplus}+2s_I\leq s$ always holds with $s$ given in Eq.~\eqref{eq:s_in_EC}, we have the condition EC A.
Furthermore, when the SSS $r$-filter is inserted to the input, the ideal decoder is inserted to the output, $r+2s_0+s_H+2s_{\oplus} + \max\{s_X,s_Z\} <\frac{c}{2}$, and $2s_0+s_H+s_{\oplus} + 2s_I <\frac{c}{2}$, we have 
\begin{align}
    &\begin{picture}(260,40)
    \thicklines
        \put(80,30){\line(1,0){15}}
        \put(95,20){\framebox(5,20)}
        \put(102,35){\makebox(0,0)[bl]{$\scriptstyle r$}}
        \put(100,30){\line(1,0){40}}
        \put(30,5){\oval(40,20)[l]}
        \put(30,-5){\line(0,1){20}}
        \put(15,5){\makebox(0,0)[l]{$\ket{\overline{0}}$}}
        \put(32,10){\makebox(0,0)[bl]{$\scriptstyle s_0$}}
        \put(70,-20){\oval(40,20)[l]}
        \put(70,-30){\line(0,1){20}}
        \put(55,-20){\makebox(0,0)[l]{$\ket{\overline{0}}$}}
        \put(72,-15){\makebox(0,0)[bl]{$\scriptstyle s_0$}}
        \put(30,5){\line(1,0){20}}
        \put(70,-20){\line(1,0){100}}
        \put(60,5){\circle{20}}
        \put(60,5){\makebox(0,0){$\overline{H}$}}
        \put(71,10){\makebox(0,0)[bl]{$\scriptstyle s_H$}}
        \put(70,5){\line(1,0){70}}
        \put(90,-25){\line(0,1){30}}
        \put(90,5){\circle*{5}}
        \put(93,7){\makebox(0,0)[bl]{$\scriptstyle s_{\oplus}$}}
        \put(90,-20){\circle{10}}
        \put(118,32){\makebox(0,0)[bl]{$\scriptstyle s_{\oplus}$}}
        \put(115,30){\circle*{5}}
        \put(115,5){\circle{10}}
        \put(115,0){\line(0,1){30}}
        \put(117,-18){\makebox(0,0)[bl]{$\scriptstyle s_{I}$}}
        \put(140,20){\line(0,1){20}}
        \put(140,30){\oval(40,20)[r]}
        \put(148,30){\makebox(0,0){$\overline{X}$}}
        \put(161,35){\makebox(0,0)[bl]{$\scriptstyle s_{X}$}}
        \put(140,-5){\line(0,1){20}}
        \put(140,5){\oval(40,20)[r]}
        \put(148,5){\makebox(0,0){$\overline{Z}$}}
        \put(161,10){\makebox(0,0)[bl]{$\scriptstyle s_{Z}$}}
        \put(152,-18){\makebox(0,0)[bl]{$\scriptstyle s_{I}$}}
        \put(170,-30){\line(0,1){20}}
        \put(170,-10){\line(1,-1){10}}
        \put(180,-20){\line(-1,-1){10}}
        \thinlines
        \put(180,-20){\line(1,0){10}}
    \end{picture}\\
    &=
    \begin{picture}(400,80)
    \thicklines
        \put(140,30){\line(1,0){15}}
        \put(30,5){\oval(40,20)[l]}
        \put(30,-5){\line(0,1){20}}
        \put(15,5){\makebox(0,0)[l]{$\ket{\overline{0}}$}}
        \put(32,10){\makebox(0,0)[bl]{$\scriptstyle s_0$}}
        \put(85,-20){\oval(40,20)[l]}
        \put(85,-30){\line(0,1){20}}
        \put(70,-20){\makebox(0,0)[l]{$\ket{\overline{0}}$}}
        \put(87,-15){\makebox(0,0)[bl]{$\scriptstyle s_0$}}
        \put(30,5){\line(1,0){15}}
        \put(85,-20){\line(1,0){15}}
        \put(45,-5){\framebox(5,20)}
        \put(50,5){\line(1,0){15}}
        \put(52,10){\makebox(0,0)[bl]{$\scriptstyle s_0$}}
        \put(100,-30){\framebox(5,20)}
        \put(107,-15){\makebox(0,0)[bl]{$\scriptstyle s_0$}}
        \put(105,-20){\line(1,0){50}}
        \put(75,5){\circle{20}}
        \put(75,5){\makebox(0,0){$\overline{H}$}}
        \put(86,10){\makebox(0,0)[bl]{$\scriptstyle s_H$}}
        \put(85,5){\line(1,0){15}}
        \put(100,-5){\framebox(5,20)}
        \put(107,10){\makebox(0,0)[bl]{$\scriptstyle s_0+s_H$}}
        \put(105,5){\line(1,0){50}}
        \put(135,-25){\line(0,1){30}}
        \put(135,5){\circle*{5}}
        \put(138,7){\makebox(0,0)[bl]{$\scriptstyle s_{\oplus}$}}
        \put(135,-20){\circle{10}}
        \put(155,20){\framebox(5,20)}
        \put(162,35){\makebox(0,0)[bl]{$\scriptstyle r$}}
        \put(160,30){\line(1,0){75}}
        \put(155,-30){\framebox(5,20)}
        \put(162,-15){\makebox(0,0)[bl]{$\scriptstyle 2s_0 + s_H + s_{\oplus}$}}
        \put(155,-5){\framebox(5,20)}
        \put(160,5){\line(1,0){75}}
        \put(160,-20){\line(1,0){70}}
        \put(162,10){\makebox(0,0)[bl]{$\scriptstyle 2s_0 + s_H + s_{\oplus}$}}
        \put(213,32){\makebox(0,0)[bl]{$\scriptstyle s_{\oplus}$}}
        \put(210,30){\circle*{5}}
        \put(210,5){\circle{10}}
        \put(210,0){\line(0,1){30}}
        \put(212,-18){\makebox(0,0)[bl]{$\scriptstyle s_{I}$}}
        \put(235,20){\line(0,1){20}}
        \put(235,30){\oval(40,20)[r]}
        \put(243,30){\makebox(0,0){$\overline{X}$}}
        \put(256,35){\makebox(0,0)[bl]{$\scriptstyle s_{X}$}}
        \put(235,-5){\line(0,1){20}}
        \put(235,5){\oval(40,20)[r]}
        \put(243,5){\makebox(0,0){$\overline{Z}$}}
        \put(256,10){\makebox(0,0)[bl]{$\scriptstyle s_{Z}$}}
        \put(230,-30){\framebox(5,20)}
        \put(235,-20){\line(1,0){75}}
        \put(237,-15){\makebox(0,0)[bl]{$\scriptstyle 2s_0 + s_H + s_{\oplus}+s_I$}}
        \put(297,-18){\makebox(0,0)[bl]{$\scriptstyle s_I$}}
        \put(310,-30){\framebox(5,20)}
        \put(315,-20){\line(1,0){70}}
        \put(317,-15){\makebox(0,0)[bl]{$\scriptstyle 2s_0 + s_H + s_{\oplus}+2s_I$}}
        \put(385,-30){\line(0,1){20}}
        \put(385,-10){\line(1,-1){10}}
        \put(395,-20){\line(-1,-1){10}}
        \thinlines
        \put(395,-20){\line(1,0){10}}
    \end{picture}\\
    &=
        \begin{picture}(340,80)
    \thicklines
        \put(140,30){\line(1,0){15}}
        \put(30,5){\oval(40,20)[l]}
        \put(30,-5){\line(0,1){20}}
        \put(15,5){\makebox(0,0)[l]{$\ket{\overline{0}}$}}
        \put(32,10){\makebox(0,0)[bl]{$\scriptstyle s_0$}}
        \put(85,-20){\oval(40,20)[l]}
        \put(85,-30){\line(0,1){20}}
        \put(70,-20){\makebox(0,0)[l]{$\ket{\overline{0}}$}}
        \put(87,-15){\makebox(0,0)[bl]{$\scriptstyle s_0$}}
        \put(30,5){\line(1,0){15}}
        \put(85,-20){\line(1,0){15}}
        \put(45,-5){\framebox(5,20)}
        \put(50,5){\line(1,0){15}}
        \put(52,10){\makebox(0,0)[bl]{$\scriptstyle s_0$}}
        \put(100,-30){\framebox(5,20)}
        \put(107,-15){\makebox(0,0)[bl]{$\scriptstyle s_0$}}
        \put(105,-20){\line(1,0){50}}
        \put(75,5){\circle{20}}
        \put(75,5){\makebox(0,0){$\overline{H}$}}
        \put(86,10){\makebox(0,0)[bl]{$\scriptstyle s_H$}}
        \put(85,5){\line(1,0){15}}
        \put(100,-5){\framebox(5,20)}
        \put(107,10){\makebox(0,0)[bl]{$\scriptstyle s_0+s_H$}}
        \put(105,5){\line(1,0){50}}
        \put(135,-25){\line(0,1){30}}
        \put(135,5){\circle*{5}}
        \put(138,7){\makebox(0,0)[bl]{$\scriptstyle s_{\oplus}$}}
        \put(135,-20){\circle{10}}
        \put(155,20){\framebox(5,20)}
        \put(162,35){\makebox(0,0)[bl]{$\scriptstyle r$}}
        \put(160,30){\line(1,0){70}}
        \put(155,-5){\framebox(5,20)}
        \put(160,5){\line(1,0){70}}
        \put(162,10){\makebox(0,0)[bl]{$\scriptstyle 2s_0 + s_H + s_{\oplus}$}}
        \put(155,-30){\framebox(5,20)}
        \put(160,-20){\line(1,0){70}}
        \put(162,-15){\makebox(0,0)[bl]{$\scriptstyle 2s_0 + s_H + s_{\oplus}$}}
        \put(213,32){\makebox(0,0)[bl]{$\scriptstyle s_{\oplus}$}}
        \put(210,30){\circle*{5}}
        \put(210,5){\circle{10}}
        \put(210,0){\line(0,1){30}}
        \put(212,-18){\makebox(0,0)[bl]{$\scriptstyle s_{I}$}}
        \put(230,-30){\framebox(5,20)}
        \put(237,-15){\makebox(0,0)[bl]{$\scriptstyle 2s_0 + s_H + s_{\oplus}+s_I$}}
        \put(230,-5){\framebox(5,20)}
        \put(235,5){\line(1,0){65}}
        \put(235,-20){\line(1,0){80}}
        \put(237,10){\makebox(0,0)[bl]{$\scriptstyle r+2s_0 + s_H + 2s_{\oplus}$}}
        \put(230,20){\framebox(5,20)}
        \put(237,35){\makebox(0,0)[bl]{$\scriptstyle r+2s_0 + s_H + 2s_{\oplus}$}}
        \put(235,30){\line(1,0){65}}
        \put(297,-18){\makebox(0,0)[bl]{$\scriptstyle s_I$}}
        \put(300,20){\line(0,1){20}}
        \put(300,30){\oval(40,20)[r]}
        \put(308,30){\makebox(0,0){$\overline{X}$}}
        \put(321,35){\makebox(0,0)[bl]{$\scriptstyle s_{X}$}}
        \put(300,-5){\line(0,1){20}}
        \put(300,5){\oval(40,20)[r]}
        \put(308,5){\makebox(0,0){$\overline{Z}$}}
        \put(321,10){\makebox(0,0)[bl]{$\scriptstyle s_{Z}$}}
        \put(315,-30){\line(0,1){20}}
        \put(315,-10){\line(1,-1){10}}
        \put(325,-20){\line(-1,-1){10}}
        \thinlines
        \put(325,-20){\line(1,0){10}}
    \end{picture}\\
    &=
        \begin{picture}(260,80)
    \thicklines
        \put(140,30){\line(1,0){15}}
        \put(30,5){\oval(40,20)[l]}
        \put(30,-5){\line(0,1){20}}
        \put(15,5){\makebox(0,0)[l]{$\ket{\overline{0}}$}}
        \put(32,10){\makebox(0,0)[bl]{$\scriptstyle s_0$}}
        \put(85,-20){\oval(40,20)[l]}
        \put(85,-30){\line(0,1){20}}
        \put(70,-20){\makebox(0,0)[l]{$\ket{\overline{0}}$}}
        \put(87,-15){\makebox(0,0)[bl]{$\scriptstyle s_0$}}
        \put(30,5){\line(1,0){15}}
        \put(85,-20){\line(1,0){15}}
        \put(45,-5){\framebox(5,20)}
        \put(50,5){\line(1,0){15}}
        \put(52,10){\makebox(0,0)[bl]{$\scriptstyle s_0$}}
        \put(100,-30){\framebox(5,20)}
        \put(107,-15){\makebox(0,0)[bl]{$\scriptstyle s_0$}}
        \put(105,-20){\line(1,0){50}}
        \put(75,5){\circle{20}}
        \put(75,5){\makebox(0,0){$\overline{H}$}}
        \put(86,10){\makebox(0,0)[bl]{$\scriptstyle s_H$}}
        \put(85,5){\line(1,0){15}}
        \put(100,-5){\framebox(5,20)}
        \put(107,10){\makebox(0,0)[bl]{$\scriptstyle s_0+s_H$}}
        \put(105,5){\line(1,0){50}}
        \put(135,-25){\line(0,1){30}}
        \put(135,5){\circle*{5}}
        \put(138,7){\makebox(0,0)[bl]{$\scriptstyle s_{\oplus}$}}
        \put(135,-20){\circle{10}}
        \put(155,20){\framebox(5,20)}
        \put(162,35){\makebox(0,0)[bl]{$\scriptstyle r$}}
        \put(160,30){\line(1,0){70}}
        \put(155,-5){\framebox(5,20)}
        \put(160,5){\line(1,0){70}}
        \put(162,10){\makebox(0,0)[bl]{$\scriptstyle 2s_0 + s_H + s_{\oplus}$}}
        \put(155,-30){\framebox(5,20)}
        \put(160,-20){\line(1,0){70}}
        \put(162,-15){\makebox(0,0)[bl]{$\scriptstyle 2s_0 + s_H + s_{\oplus}$}}
        \put(213,32){\makebox(0,0)[bl]{$\scriptstyle s_{\oplus}$}}
        \put(210,30){\circle*{5}}
        \put(210,5){\circle{10}}
        \put(210,0){\line(0,1){30}}
        \put(212,-18){\makebox(0,0)[bl]{$\scriptstyle s_I$}}
        \put(230,-30){\framebox(5,20)}
        \put(237,-15){\makebox(0,0)[bl]{$\scriptstyle 2s_0 + s_H + s_{\oplus} + s_I$}}
        \put(230,-5){\framebox(5,20)}
        \put(235,5){\line(1,0){65}}
        \put(235,-20){\line(1,0){65}}
        \put(237,10){\makebox(0,0)[bl]{$\scriptstyle r+2s_0 + s_H + 2s_{\oplus}$}}
        \put(230,20){\framebox(5,20)}
        \put(237,35){\makebox(0,0)[bl]{$\scriptstyle r+2s_0 + s_H + 2s_{\oplus}$}}
        \put(235,30){\line(1,0){65}}
        \put(300,20){\line(0,1){20}}
        \put(300,40){\line(1,-1){10}}
        \put(310,30){\line(-1,-1){10}}
        \put(300,-5){\line(0,1){20}}
        \put(300,15){\line(1,-1){10}}
        \put(310,5){\line(-1,-1){10}}
        \put(300,-30){\line(0,1){20}}
        \put(300,-10){\line(1,-1){10}}
        \put(310,-20){\line(-1,-1){10}}
        \thinlines
        \put(310,30){\line(1,0){10}}
        \put(320,20){\line(0,1){20}}
        \put(320,30){\oval(20,20)[r]}
        \put(324,30){\makebox(0,0){$X$}}
        \put(310,5){\line(1,0){10}}
        \put(320,-5){\line(0,1){20}}
        \put(320,5){\oval(20,20)[r]}
        \put(324,5){\makebox(0,0){$Z$}}
        \put(310,-20){\line(1,0){35}}
    \end{picture}\\
    &=
    \begin{picture}(260,80)
    \thicklines
        \put(90,30){\line(1,0){15}}
        \put(105,20){\framebox(5,20)}
        \put(112,35){\makebox(0,0)[bl]{$\scriptstyle r$}}
        \put(110,30){\line(1,0){15}}
        \put(30,5){\oval(40,20)[l]}
        \put(30,-5){\line(0,1){20}}
        \put(15,5){\makebox(0,0)[l]{$\ket{\overline{0}}$}}
        \put(32,10){\makebox(0,0)[bl]{$\scriptstyle s_0$}}
        \put(65,-20){\oval(40,20)[l]}
        \put(65,-30){\line(0,1){20}}
        \put(50,-20){\makebox(0,0)[l]{$\ket{\overline{0}}$}}
        \put(67,-15){\makebox(0,0)[bl]{$\scriptstyle s_0$}}
        \put(30,5){\line(1,0){15}}
        \put(45,-5){\framebox(5,20)}
        \put(50,5){\line(1,0){15}}
        \put(52,10){\makebox(0,0)[bl]{$\scriptstyle s_0$}}
        \put(65,-20){\line(1,0){15}}
        \put(80,-30){\framebox(5,20)}
        \put(87,-15){\makebox(0,0)[bl]{$\scriptstyle s_0$}}
        \put(85,-20){\line(1,0){15}}
        \put(125,20){\line(0,1){20}}
        \put(125,40){\line(1,-1){10}}
        \put(135,30){\line(-1,-1){10}}
        \put(65,-5){\line(0,1){20}}
        \put(65,15){\line(1,-1){10}}
        \put(75,5){\line(-1,-1){10}}
        \put(100,-30){\line(0,1){20}}
        \put(100,-10){\line(1,-1){10}}
        \put(110,-20){\line(-1,-1){10}}
        \thinlines
        \put(135,30){\line(1,0){45}}
        \put(75,5){\line(1,0){15}}
        \put(110,-20){\line(1,0){95}}
        \put(100,5){\circle{20}}
        \put(100,5){\makebox(0,0){$H$}}
        \put(110,5){\line(1,0){70}}
        \put(130,-25){\line(0,1){30}}
        \put(130,5){\circle*{5}}
        \put(130,-20){\circle{10}}
        \put(155,30){\circle*{5}}
        \put(155,5){\circle{10}}
        \put(155,0){\line(0,1){30}}
        \put(180,20){\line(0,1){20}}
        \put(180,30){\oval(20,20)[r]}
        \put(184,30){\makebox(0,0){$X$}}
        \put(180,-5){\line(0,1){20}}
        \put(180,5){\oval(20,20)[r]}
        \put(184,5){\makebox(0,0){$Z$}}
    \end{picture}\\
    &=
    \begin{picture}(260,80)
    \thicklines
        \put(10,30){\line(1,0){10}}
        \put(20,20){\framebox(5,20)}
        \put(27,35){\makebox(0,0)[bl]{$\scriptstyle r$}}
        \put(25,30){\line(1,0){15}}
        \put(40,20){\line(0,1){20}}
        \put(40,40){\line(1,-1){10}}
        \put(50,30){\line(-1,-1){10}}
        \thinlines
        \put(70,5){\oval(20,20)[l]}
        \put(60,5){\makebox(0,0)[l]{$\ket{0}$}}
        \put(70,-5){\line(0,1){20}}
        \put(70,-20){\oval(20,20)[l]}
        \put(60,-20){\makebox(0,0)[l]{$\ket{0}$}}
        \put(70,-30){\line(0,1){20}}
        \put(50,30){\line(1,0){130}}
        \put(70,5){\line(1,0){20}}
        \put(70,-20){\line(1,0){130}}
        \put(100,5){\circle{20}}
        \put(100,5){\makebox(0,0){$H$}}
        \put(110,5){\line(1,0){70}}
        \put(130,-25){\line(0,1){30}}
        \put(130,5){\circle*{5}}
        \put(130,-20){\circle{10}}
        \put(155,30){\circle*{5}}
        \put(155,5){\circle{10}}
        \put(155,0){\line(0,1){30}}
        \put(180,20){\line(0,1){20}}
        \put(180,30){\oval(20,20)[r]}
        \put(184,30){\makebox(0,0){$X$}}
        \put(180,-5){\line(0,1){20}}
        \put(180,5){\oval(20,20)[r]}
        \put(184,5){\makebox(0,0){$Z$}}
    \end{picture}\label{eq:qubit_teleportation}\\
    &=
    \begin{picture}(100,80)
    \thicklines
        \put(10,5){\line(1,0){10}}
        \put(20,-5){\framebox(5,20)}
        \put(27,10){\makebox(0,0)[bl]{$\scriptstyle r$}}
        \put(25,5){\line(1,0){15}}
        \put(40,-5){\line(0,1){20}}
        \put(40,15){\line(1,-1){10}}
        \put(50,5){\line(-1,-1){10}}
    \thinlines
        \put(50,5){\line(1,0){15}}
        \put(75,5){\circle{20}}
        \put(75,5){\makebox(0,0){$P$}}
        \put(85,5){\line(1,0){20}}
    \end{picture}\\
    &\nonumber 
\end{align}
where $P\in\{I,X,Y,Z\}$ denotes a Pauli correction depending on the measurement outcomes.  Due to the definition of $s$ in Eq.~\eqref{eq:s_in_EC}, we have the condition EC B up to a qubit Pauli correction.  As we stated previously, this correction operation can be dealt with by the change of the Pauli frame or the change of the successive gate or measurement.
\end{proof}
One thing to note is that when $r+s\geq \frac{c}{2}$, the Pauli error correction $P$ may not be inferred correctly.  This causes the qubit-level logical error, which is passed on to the higher level and to be corrected by a qubit quantum error-correcting code.  The detail will be described in Sec.~\ref{sec:threshold_theorem}.

\subsection{Energy-constraint conditions for GKP gadgets}\label{sec:energy_constraint_conditions}

As has already been discussed in Sec.~\ref{sec:energy_constrained_diamond_norm}, a quantum state on a mode $Q$ during the computation should be an element of $\mathfrak{S}_{E}({\cal H}_Q)$ with a finite $E$.  This is necessary for eliminating the unrealistic quantum states that have arbitrarily high susceptibility to the given noise, and thus for applying the appropriate distance measure (i.e., the energy-constrained diamond-norm distance) to noise.
Thus, our noisy quantum circuit needs to maintain a quantum state during the computation in $\mathfrak{S}_{E}({\cal H}_Q)$.
Furthermore, we need to compare our noisy quantum circuit with the \quoted{ideal} quantum circuit, which always gives us the correct result.  Thus, the \quoted{ideal} circuit should also keep the energy of a quantum state during the computation finite.  Otherwise, the noise model we can treat will be severely restricted.  For these reasons, we make an additional requirement on the actual noisy gadgets as well as $s$-parameterized gadgets in Defs.~\ref{def:s-preparation} and \ref{def:s-gate}.

\begin{definition}[$E_{\rm prep}$-energy constraint on preparation gadgets] \label{def:energy_constrained_prep}
    For a positive constant $E_{\rm prep}$, a preparation gadget is said to satisfy the $E_{\rm prep}$-energy constraint if the output state of the gadget is an element of $\mathfrak{S}_{E_{\rm prep}}({\cal H}_Q)$.  
\end{definition}
\begin{definition}[$g_{\rm sup}$-energy constraint on gate gadgets]\label{def:energy_constrained_gate}
    For a positive, monotonically increasing, locally bounded function $g_{\rm sup}(E)$ of $E\in[0,\infty)$, a single-mode gate gadget is said to satisfy the $g_{\rm sup}$-energy constraint tputs a quantum state that is contained in $\mathfrak{S}_{g_{\rm sup}(E)}({\cal H}_Q)$ when an element of $\mathfrak{S}_{E}({\cal H}_Q)$ is input for any $E\in[0,\infty)$.  A two-mode gate gadget is said to satisfy the $g_{\rm sup}$-energy constraint if for any input state whose reduced states on respective modes $Q_1$ and $Q_2$ are elements of $\mathfrak{S}_{g_{\rm sup}(E_1)}({\cal H}_{Q_1})$ and $\mathfrak{S}_{g_{\rm sup}(E_2)}({\cal H}_{Q_2})$, respectively,  it outputs a state whose reduced states are contained in $\mathfrak{S}_{g(E_1,E_2)}({\cal H}_{Q_1})$ and $\mathfrak{S}_{g(E_1,E_2)}({\cal H}_{Q_2})$, respectively, where a function $g:[0,\infty)\times[0,\infty)\to[0,\infty)$ is monotonically increasing and locally bounded in both arguments and satisfies $g(E, E)\leq g_{\rm sup}(E)$ for any $E\in[0,\infty)$.
\end{definition}

For the actual noisy gadgets that we can implement in a lab, the above requirements may be testable.
However, we need to check beforehand whether the $s$-preparation in Def.~\ref{def:s-preparation} can in principle satisfy the $E_{\rm prep}$-energy constraint and whether the $s$-parameterized gate in Def.~\ref{def:s-gate} can in principle satisfy the $g_{\rm sup}$-energy constraint.
For the $s$-preparation gadget, there exists a (pure) state in the form of Eq.~$\eqref{eq:s-parameterized_pure}$ with finite energy, which is given as follows.
Let $s\leq \frac{c}{2}$ and let $f^s(x)$ be a bump function defined as
\begin{align}
    f^s(x) &\coloneqq \begin{cases}\exp\bigl(-\frac{1}{1-(x/s)^2}\bigr) & |x|<s\\
    0 & \text{otherwise}
    \end{cases},\label{eq:bump_func}
\intertext{and let}
\tilde{f}_s(k) &\coloneqq \frac{1}{\sqrt{2\pi}}\int f^s(x) e^{ikx} dx    
\end{align}
be its Fourier transform.
Then, the state $\ket{\rho^s[\psi]}$ that approximates the GKP state $\ket{\overline{\psi}}=\alpha\ket{\overline{0}}+\beta\ket{\overline{1}}$ in the sense of Eq.~\eqref{eq:s-parameterized_pure} is given by
\begin{align}
    \ket{\rho^s[\psi]} &\coloneqq \frac{1}{\sqrt{N_s}}\iint_{-\infty}^{\infty}dz_1 dz_2\, f^s(z_1)f^s(z_2)(\alpha \ket{0;z_1,z_2}_{LS} + \beta\ket{1;z_1,z_2}_{LS}) \label{eq:def_in_terms_of_subsystem} \\
    &=\frac{1}{\sqrt{c N_s}}\int_{-\infty}^{\infty} dx \iint_{-s}^{s}dz_1 dz_2\,e^{iz_1z_2/2}
\nonumber \\
    &\qquad \times \sum_{m\in\mathbb{Z}}  e^{2ciz_2 m} f^s(z_1)f^s(z_2) (\alpha \delta(z_1+2mc-x) + \beta e^{ciz_2}\delta(z_1+(2m+1)c-x)  \ket{x}_q, \label{eq:expansion_bump}
\end{align}
where $N_s$ is a normalization constant, and the second equality follows from Eqs.~\eqref{eq:Zak_in_position} and \eqref{eq:shifted_one}.  Here, we can decompose any $x\in\mathbb{R}$ into the nearest multiple of~$c$ plus a remainder~\cite{Pantaleoni2020}:
\begin{align}
    x = \lfloor x \rceil_c + \{ x \}_c,
\end{align}
and we now define the unique integer
\begin{align}
    n(x) \coloneqq \frac {\lfloor x \rceil_c} {c}
    = \left\lfloor \frac x c \right\rceil
\end{align}
that satisfies $-\frac{c}{2}\leq x-cn(x)<\frac{c}{2}$.  Then, from Eq.~\eqref{eq:expansion_bump}, we have
\begin{equation}
    \ket{\rho^s[\psi]} = \frac{1}{\sqrt{c N_s}}\int_{-\infty}^{\infty} dx \int_{-s}^{s} dz_2\,\gamma(n(x)) e^{i(x+cn(x))z_2/2}f_s(x-cn(x)) f^s(z_2)\ket{x}_q
\end{equation}
where $\gamma(n)$ with $n\in\mathbb{Z}$ is defined as
\begin{equation}
    \gamma(n)\coloneqq
    \begin{cases}
        \alpha & (n \text{ is even}),\\
        \beta & (n \text{ is odd}).
    \end{cases}
\end{equation}
Since $f^s$ is supported only on $[-s,s) \subseteq \bigl[-\tfrac c 2, \tfrac c 2\bigr)$, we have, for any function $g(x,n(x))$ of $x$ and $n(x)$,
\begin{equation}
    \int dx\, f^s(x-cn(x)) g(x,n(x)) = \int dx \sum_{n\in\mathbb{Z}} f^s(x-cn) g(x,n),
\end{equation}
and 
\begin{equation}
    \int_{-s}^{s}dz_2\, f^s(z_2)e^{iz_2(x+a)/2} = \int_{-\infty}^{\infty}dz_2\, f^s(z_2)e^{iz_2(x+a)/2} = \sqrt{2\pi}\, \tilde{f}^s\Bigl(\frac{x+a}{2}\Bigr).
\end{equation}
Thus, we have
\begin{align}
    \ket{\rho^s[\psi]}
    &=\sqrt{\frac{2c}{N_s}}\int dx\sum_{n\in{\mathbb{Z}}}\gamma(n)\tilde{f}^{s}\Bigl(\frac{x+cn}{2}\Bigr)f^s(x-cn)\ket{x}_q \\ 
    &=\sqrt{\frac{2c}{N_s}}\int dx\sum_{m\in\mathbb{Z}} \left(\alpha\tilde{f}^{s}\Bigl(\frac{x+2mc}{2}\Bigr) f^s(x-2mc)  + \beta\tilde{f}^{s}\Bigl(\frac{x+(2m+1)c}{2}\Bigr) f^s(x-(2m+1)c)\right)\ket{x}_q.\label{eq:expression_s-parameterized_GKP}
\end{align}

The momentum wave function has a similar form as the above position wave function due to the symmetry in the subsystem decomposition in Eq.~\eqref{eq:def_in_terms_of_subsystem}.
The large-scale behavior of the above wave function is determined by the function $\tilde{f}^s$.  Since the asymptotic behavior of $\tilde{f}^s$ is known to be $\sim 
|k|^{-\frac{3}{4}}e^{-\sqrt{s|k|}}$ \cite{Johnson2015} and thus decays faster than any polynomial, this $s$-parameterized approximate GKP state $\ket{\rho^{s}[\psi]}$ has a finite energy.
(Recall that the energy of a quantum state is given by the sum of the second moments of the position and momentum probability distribution.)  
Therefore, an $s$-preparation gadget that prepares the quantum state illustrated above---or any mixtures thereof---can satisfy the $E_{\rm prep}$-energy constraint for some sufficiently large yet constant $E_{\rm prep}$.
Note that, with the same reasoning, a state prepared as above but with $f^s$ replaced by
\begin{equation}
    f^s_{c_1, c_2}(x) \coloneqq \begin{cases}\exp\bigl(-\frac{c_2}{(1-(x/s)^2)^{c_1}}\bigr) & |x|<s,
    \\
    0 & \text{otherwise},
    \end{cases} \label{eq:f_s^alpha_beta}
\end{equation}
for any $c_1,c_2>0$,
also has finite energy \cite{Johnson2015}.

Below, we will also prove that the ideal gates given in Sec.~\ref{sec:gadget_construction} satisfy the $g_{\rm sup}$-energy constraint with a positive, monotonically increasing, locally bounded function $g_{\rm sup}$.  First, since phase rotation does not change the mode energy, the gate $\hat{F}$ satisfies the $g_{\rm sup}$-energy constraint as long as $E \leq g_{\rm sup}(E)$.  Next, the displacement operator $\hat{V}(a,b)$ acts on the number operator $\hat{n}$ in the Heisenberg picture as
\begin{align}
    \hat{V}(a,b)^{\dagger}\hat{n}\hat{V}(a,b) & = \hat{V}(a,b)^{\dagger}\frac{\hat{q}^2+\hat{p}^2-1}{2}\hat{V}(a,b) \\
    & = \frac{(\hat{q}+a)^2+(\hat{p}+b)^2-1}{2} \\
    & = \hat{n} + \frac{a^2 + b^2}{2} + a\hat{q} + b\hat{p}.
\end{align}
Thus, for any state $\hat{\rho}\in\mathfrak{S}_{E}({\cal H}_Q)$, we have
\begin{align}
    \braket{\hat{V}(a,b)^{\dagger}\hat{n}\hat{V}(a,b)}_{\hat{\rho}} &= E + \frac{a^2 + b^2}{2} + a\braket{\hat{q}}_{\hat{\rho}} + b\braket{\hat{p}}_{\hat{\rho}} \\
    &\leq E + \frac{a^2 + b^2}{2} + \abs a \sqrt{\braket{\hat{q}^2}_{\hat{\rho}}} + \abs b \sqrt{\braket{\hat{p}^2}_{\hat{\rho}}} \\
    &\leq E + \frac{a^2 + b^2}{2} + (\abs a + \abs 
    b)\sqrt{2E+1} \label{eq:change_energy_displacement}
\end{align}
Thus, the displacements $\hat{V}(c,0)$ and $\hat{V}(0,c)$ satisfy the $g_{\rm sup}$-energy constraint as long as
\begin{equation}
    g_{\rm sup}(E) \geq E + c\sqrt{2E + 1} + \frac{c^2}{2}.
\end{equation}
Finally, the SUM gate transforms the number operator $\hat{n}_1$ on the controlled system as 
\begin{align}
    \hat{\text{SUM}}^{\dagger}\hat{n}_1\hat{\text{SUM}} &=\hat{\text{SUM}}^{\dagger}\frac{\hat{q}^2+\hat{p}^2-1}{2}\hat{\text{SUM}}\\  &= \frac{\hat{q}_1^2 + (\hat{p}_1 + \hat{p}_2)^2 -1}{2} \\
    & = \hat{n}_1 + \hat{p}_1\hat{p}_2 +\frac{\hat{p}_2^2}{2}.
\end{align}
We thus have, for any $\hat{\rho}\in\mathfrak{S}({\cal H}_{Q_1Q_2})$ with $\Tr _{Q_2}[\hat{\rho}]\in\mathfrak{S}_{E_1}({\cal H}_{Q_1})$ and $\Tr _{Q_1}[\hat{\rho}]\in\mathfrak{S}_{E_2}({\cal H}_{Q_2})$, that
\begin{align}
    \braket{\hat{\text{SUM}}^{\dagger}\hat{n}_1\hat{\text{SUM}}}_{\hat{\rho}} 
    &\leq E_1 + \sqrt{\braket{\hat{p}_1^2}_{\hat{\rho}}\braket{\hat{p}_2^2}_{\hat{\rho}}} + \frac{\braket{\hat{p}_2^2}_{\hat{\rho}}}{2} \\
    &\leq E_1 + E_2 + \sqrt{(2E_1+1)(2E_2+1)} + \frac{1}{2},
\end{align}
where the first inequality comes from the positivity of the covariance matrix.
The same inequality holds for $\hat{n}_2$, and thus the SUM gate satisfies $g_{\rm sup}$-energy constraint as long as
\begin{equation}
    g_{\rm sup}(E) \geq 4E + \frac{3}{2}.
\end{equation}
Thus, we conclude that each ideal gate operation given in Sec.~\ref{sec:physical_operations} actually satisfies the $g_{\rm sup}$-energy constraint defined above for an appropriate function $g_{\rm sup}$.  (In the case of the gate set given above, the function $g_{\rm sup}(E)=4E+c\sqrt{2E+1}+\frac{c^2}{2}$, for example, satisfies all conditions.)
Notice that $g_{\rm sup}(E) \geq E$ should always hold for any choice of a gate set since the wait does not change the energy.

For the $s$-gate to satisfy the energy constraint, the noise channel ${\cal N}^s$ also needs to obey the energy constraint.  As illustrated in Eq.~\eqref{eq:change_energy_displacement}, an up-to-$s$ random displacement channel obeys the $g_{\rm sup}$-energy constraint with an appropriate function $g_{\rm sup}$.  Thus, there exists a family of channels of the forms in Eq.~\eqref{eq:def_N_s} that obey the $g_{\rm sup}$-energy constraints.
From these, we find that all the $s$-gate gadgets can satisfy the $g_{\rm sup}$-energy constraint condition with an appropriate function $g_{\rm sup}$.

Now, we show that the energy of a quantum state during the computation is kept finite regardless of the depth $D$ of the circuit.  Since we use Knill-type EC gadgets, quantum information is kept teleported to a newly prepared state at every EC gadget.  More precisely, we have the following for the output state of the EC gadget.
\begin{proposition}[Energy reset of the EC gadget]\label{prop:energy_reset}
    let $M$ denote the classical register to keep the GKP Pauli measurement outcomes $\{0,1\}^2$ in the GKP EC gadget given in Eq.~\eqref{eq:knill_EC}, and let ${\cal P}$ denote the set of probability distributions over $\{0,1\}^2$.
    For the EC gadget defined in Eq.~\eqref{eq:knill_EC} with the input and output systems labeled by $Q_{\rm in}$ and $Q_{\rm out}$, respectively, let $\Phi^{\rm EC}:\mathfrak{S}({\cal H}_{Q_{\rm in}})\to {\cal P}_M\otimes \mathfrak{S}({\cal H}_{Q_{\rm out}})$ be a CPTP map of the EC gadget with all the preparation gadgets inside satisfying the $E_{\rm prep}$-energy constraint and all the gate gadgets inside satisfying the $g_{\rm sup}$-energy constraint.  Then, for any state $\hat{\rho}\in\mathfrak{S}({\cal H}_{Q_{\rm in}R})$ between the input of the EC gadget and a reference system, $\Phi^{\rm EC}\otimes \mathrm{Id}_R(\hat{\rho})$ is a classical-quantum-quantum state between the systems $M$, $Q_{\rm out}$, and $R$ with its reduced density operator on the system $Q_{\rm out}$ contained in $\mathfrak{S}_{g_{\rm sup}^4(E_{\rm prep})}({\cal H}_{Q_{\rm out}})$, where $g_{\rm sup}^{m}$ for $m\in\mathbb{N}$ is defined as
    \begin{equation}
        g_{\rm sup}^{m} \coloneqq \underbrace{g_{\rm sup}\circ g_{\rm sup} \circ \cdots\circ g_{\rm sup}}_m.\label{eq:undergoes_m}
    \end{equation}
\end{proposition}
\begin{proof}
    The output quantum system of $\Phi^{\rm EC}$ is generated by the $E_{\rm prep}$-energy-constrained preparation and goes under four $g_{\rm sup}$-energy-constrained gate gadgets as shown in Eq.~\eqref{eq:knill_EC}.  (The number of gates that the prepared state experiences is three, but the first SUM gate takes the maximum energy of the two input states when computing the energy increase.)  Thus, the reduced density operator on the system $Q_{\rm out}$ is contained in $\mathfrak{S}_{g_{\rm sup}^4(E_{\rm prep})}({\cal H}_{Q_{\rm out}})$.  Since the reduced density operator on $Q_{\rm out}$ is not affected by any other operations in $\Phi^{\rm EC}$, nor by the input state $\hat{\rho}$, the statement holds.
\end{proof}

\noindent Therefore, in an FT-GKP circuit with Knill EC gadgets, a quantum state prepared at some point in the circuit experiences only a finite number of gate gadgets before it goes into the measurement gadget.  Thus, we arrive at the following conclusion, which is one of the main reasons why we adopt Knill-type error correction in the fault-tolerant GKP EC gadget (although analogous arguments may hold for Steane-type, as well).
\begin{proposition} \label{prop:maximum_energy_ExRec}
    Let $\ell$ be the maximum number of gate gadgets (including the wait) that a quantum state prepared during the computation undergoes before it is measured.  Assume that all the preparation gadgets and all the gate gadgets in an FT-GKP circuit satisfy the $E_{\rm prep}$-energy constraint and the $g_{\rm sup}$-energy constraint, respectively. Then, a quantum state in the FT-GKP circuit reduced to an arbitrary single mode is always contained in $\mathfrak{S}_{g_{\rm sup}^{\ell}(E_{\rm prep})}$ (except inside a measurement gadget, which we need not specify).
    Furthermore, a reduced density operator of an input of a GKP EC gadget in the FT-GKP circuit is contained in $\mathfrak{S}_{g_{\rm sup}^{\ell-1}(E_{\rm prep})}$.
\end{proposition}
\begin{proof}
    The first half of the statement is trivial from the definition of $\ell$ and the energy constraint condition in Defs.~\ref{def:energy_constrained_prep} and \ref{def:energy_constrained_gate}.  The latter statement follows from the fact that an input state of a GKP EC gadget experiences exactly one $g_{\rm sup}$-energy-constrained gate as shown in Eq.~\eqref{eq:knill_EC},
\end{proof}

\subsection{Threshold theorem for CV fault-tolerant quantum computation}\label{sec:threshold_theorem}

What we will show in this section is that the noise model defined for the CV circuit in Sec.~\ref{sec:ftqc_noise_model} will be translated into the local Markovian noise model on the logical qubit circuit, which will be defined below. 
To describe this noise model, let us consider a quantum circuit consisting of qubits for now.
Let $\{C_i\}_{i\in{\cal I}}$ be the chronologically ordered set of locations (preparations, gates, measurements, and waits) within a $W$-qubit $D$-depth quantum circuit $C$. (Among the locations at the same time step in $C$, their order is arbitrary.)  Then, $|{\cal I}|\leq WD$ holds by definition.
Let ${\cal O}_i$ be a CPTP map that we want to implement at location $C_i$.  Let $\tilde{\cal O}_i$ be its noisy version that we can implement in the actual quantum circuit, and assume that it is independent of noise at the other locations.
Then, the noisy state $\hat{\rho}^{\rm noisy}$ implemented by the actual quantum circuit can be written as
\begin{align}
    \hat{\rho}^{\rm noisy} &= \tilde{\cal O}_{|{\cal I}|} \circ \cdots  \circ \tilde{\cal O}_{1},\\
    &= ({\cal O}_{|{\cal I}|} + {\cal F}_{|{\cal I}|})\circ\cdots\circ({\cal O}_{1} + {\cal F}_{1}),\label{eq:fault_paths}
\end{align}
where ${\cal F}_i\coloneqq \tilde{\cal O}_i - {\cal O}_i$ is regarded as a fault. There are several important things to note about this formula. First, notice that there is no input state because some of the $\tilde{\cal O}_i$ are expected to be state preparations. Second, the output $\hat{\rho}^{\rm noisy}$ may be a quantum state or, if all qubits are measured, a probability distribution. Since the latter can be written as a diagonal density operator, we can use one symbol for both cases (or for a combination thereof). Finally, note that ${\cal F}_i$ is not a quantum operation; it is merely a linear map.

When we expand all the parentheses in the above, each term (except for $\mathcal{O}_{|\mathcal{I}|}\circ\cdots\circ\mathcal{O}_{1}$) is called a \emph{fault path}. In each fault path, the faults $\{{\cal F}_i\}$ are applied at certain locations, and the ideal operations $\{{\cal O}_i\}$ are applied at all the remaining locations.
Let us define the \emph{noise strength} $\epsilon_i$ for a particular fault~${\cal F}_i$ at location~$C_i$ as
\begin{equation}
    \epsilon_i \coloneqq \|{\cal F}_i\|_{\diamond}= \|\tilde{\cal O}_i - {\cal O}_i\|_{\diamond}.
\end{equation}
The above describes the case of \emph{independent Markovian} noise, where the fault at a given location does not depend on the faults in other locations. \emph{Local Markovian} noise is more general than this. For local Markovian noise, the fault at a given location may depend on faults in other locations as long as the correlations between faults are restricted so that the following definition holds~\cite{Aliferis2013}.

\begin{definition}[Local Markovian noise model]\label{def:local_Markovian}
    A quantum circuit on qubits and environmental system is said to experience a local Markovian noise model if the noisy evolution can be expanded as a sum of fault paths, where faults are described as the difference between a noisy map (acting as a correlated map on qubits and environmental system) and an ideal map (acting as a tensor-product map on qubits and environmental system), and the norm of a fault path with faults in a set~$R$ of specific locations is upper-bounded by $\epsilon_{\rm qubit}^{|R|}$ for a constant upper bound $\epsilon_{\rm qubit}$ on the noise strength at any location. 
\end{definition}
This type of noise is only approximately Markovian (despite the name). In fact, the definition shows that each fault~${\cal F}_i$ can depend on the noisy operations in other locations as long as the norm of any fault path falls off exponentially with the number of faults in the path.
If the noise model is local Markovian in this sense, then we can bound the total variation distance between a noisy outcome and an ideal one.  Let $\hat{\rho}^{\rm ideal}$ be a density operator representing a classical probability distribution of the measurement outcome at the end of the circuit given by
\begin{align}
    \hat{\rho}^{\rm ideal} \coloneqq {\cal O}_{|{\cal I}|}\circ\cdots\circ {\cal O}_{1}, 
\end{align}
and let $\hat{\rho}^{\rm fault}$ be
\begin{equation}
    \hat{\rho}^{\rm fault} \coloneqq \hat{\rho}^{\rm noisy} - \hat{\rho}^{\rm ideal}.
\end{equation}
(Note that $\hat{\rho}^{\rm fault}$ is not a density operator.) Then, the trace distance (or the total variation distance for classical probabilities) between the noisy and the ideal outcome is given by the trace norm of~$\hat{\rho}^{\rm fault}$:
\begin{align}
    \frac{1}{2}\|\hat{\rho}^{\rm noisy} - \hat{\rho}^{\rm ideal}\|_1 = \frac{1}{2}\|\hat{\rho}^{\rm fault}\|_1.
\end{align}
The trace norm of $\hat{\rho}^{\rm fault}$ can be bounded from above when $\epsilon_{\rm qubit} \leq 1/|{\cal I}|$ by the inclusion-exclusion formula~\cite{Aliferis2013}. The basic idea is that there are $\binom{|{\cal I}|}{r}$ ways of choosing exactly $r$ faulty locations among the $|{\cal I}|$ possibilities.  Since an upper bound on the norm of a fault path in local Markovian noise model depends only on the number of faulty locations, we have
\begin{equation}
    \|\hat{\rho}^{\rm fault}\|_1 \leq \left\|\sum_{r=1}^{|{\cal I}|} [\text{fault path with }r\text{ faulty locations}]\right\|_1 \leq \sum_{r=1}^{|{\cal I}|}{|{\cal I}| \choose r}\epsilon_{\rm qubit}^r \leq (1+\epsilon_{\rm qubit})^{|{\cal I}|} - 1 \leq (e-1)|{\cal I}|\,\epsilon_{\rm qubit},
\end{equation}
where we used in the last inequality that $(1+x)^a-1 \leq \frac{(1+x_{\max})^a-1}{x_{\max}}x$ holds for $a\geq 1$ and $x\in[0,x_{\max}]$, and that $(1+1/x)^x$ for $x>0$ is monotonically increasing and approaches $e$.

\begin{figure}[t]
\begin{centering}
\begin{picture}(280,70)(0,-5)

\thicklines

\put(20,45){\oval(40,20)[l]}
\put(20,35){\line(0,1){20}}
\put(20,45){\line(1,0){20}}
\put(22,50){\makebox(0,0)[bl]{$\scriptstyle s_1$}}

\put(20,15){\oval(40,20)[l]}
\put(20,5){\line(0,1){20}}
\put(20,15){\line(1,0){20}}
\put(22,20){\makebox(0,0)[bl]{$\scriptstyle s_2$}}

\put(40,35){\framebox(30,20){\text{EC}}}
\put(70,45){\line(1,0){25}}
\put(72,50){\makebox(0,0)[bl]{$\scriptstyle s_3$}}

\put(-4,2){\dashbox{2}(86,26){}}
\put(-4,32){\dashbox{2}(86,26){}}

\put(40,5){\framebox(30,20){\text{EC}}}
\put(70,15){\line(1,0){26}}
\put(72,20){\makebox(0,0)[bl]{$\scriptstyle s_4$}}

\put(95,45){\circle*{5}}
\put(95,15){\circle{10}}
\put(95,10){\line(0,1){35}}
\put(95,45){\line(1,0){20}}
\put(95,15){\line(1,0){20}}
\put(96,50){\makebox(0,0)[bl]{$\scriptstyle s_5$}}

\put(115,35){\framebox(30,20){\text{EC}}}
\put(145,45){\line(1,0){20}}
\put(147,50){\makebox(0,0)[bl]{$\scriptstyle s_6$}}

\put(115,5){\framebox(30,20){\text{EC}}}
\put(145,15){\line(1,0){20}}
\put(147,20){\makebox(0,0)[bl]{$\scriptstyle s_7$}}

\put(36,-2){\dashbox{2}(121,64){}}

\put(175,15){\circle{20}}
\put(175,15){\makebox(0,0){$\overline{H}$}}
\put(185,15){\line(1,0){20}}
\put(186,20){\makebox(0,0)[bl]{$\scriptstyle s_8$}}

\put(165,45){\oval(40,20)[r]}
\put(165,35){\line(0,1){20}}
\put(186,50){\makebox(0,0)[bl]{$\scriptstyle s_9$}}

\put(205,5){\framebox(30,20){\text{EC}}}
\put(235,15){\line(1,0){20}}
\put(237,20){\makebox(0,0)[bl]{$\scriptstyle s_{10}$}}

\put(111,32){\dashbox{2}(85,26){}}
\put(111,2){\dashbox{2}(139,26){}}

\put(255,15){\oval(40,20)[r]}
\put(255,5){\line(0,1){20}}
\put(276,20){\makebox(0,0)[bl]{$\scriptstyle s_{11}$}}

\put(201,-2){\dashbox{2}(88,34){}}

\end{picture}
\caption{Fault-tolerant circuit with extended rectangles (ExRecs) surrounded by broken lines \cite{Gottesman2009}.}
\label{fig:ExRec_example}
\end{centering}
\end{figure}

We aim to show the level reduction from a CV quantum circuit with a noise model defined in Sec.~\ref{sec:ftqc_noise_model} to a qubit quantum circuit with a local Markovian noise model.  For this, we define the following.
\begin{definition}
    A GKP extended rectangle (GKP ExRec) in a fault-tolerant circuit $C'$ consists of an FT-GKP gadget $C'_i$ replacing the location~$C_i$ of the original qubit circuit $C$ plus all the FT-GKP EC gadgets between $C'_i$ and the gadgets adjacent to it in the circuit $C'$. The FT-GKP EC steps before $C'_i$ are called leading FT-GKP EC steps, and those after $C'_i$ are called trailing FT-GKP EC steps.
\end{definition}
The GKP ExRecs in the circuit are illustrated in Fig.~\ref{fig:ExRec_example}.
We want to reduce the question of whether a higher-level (qubit-level) circuit is perfect or erroneous to an equivalent question of whether its constituent GKP ExRecs are correct or not, where the correctness of a GKP ExRec is defined as show below.
Our conditions for correctness are formally the same as those in Ref.~\cite{Gottesman2009} while the meanings of gadgets and the ideal decoder are different. Importantly, we have introduced the appropriate equivalence class of noise with the gadgets, filter, and the ideal decoder for the GKP code so that we can repurpose the techniques developed for proving the threshold theorem with the qubit concatenated code in Ref.~\cite{Gottesman2009} for the rest of our analysis.  
\begin{definition}[Correctness for the ideal GKP decoder]\label{def:correctness_ideal_decoder}
    A GKP gate or wait ExRec is correct if 
    \begin{equation}
        \begin{picture}(180,20)
    \thicklines
    \put(10,5){\line(1,0){15}}
    \put(25,-5){\framebox(30,20){\text{EC}}}
    \put(55,5){\line(1,0){15}}
    \put(57,10){\makebox(0,0)[bl]{$\scriptstyle s_l$}}
    \put(80,5){\circle{20}}
    \put(80,5){\makebox(0,0){$\overline{U}$}}
    \put(90,5){\line(1,0){15}}
    \put(91,10){\makebox(0,0)[bl]{$\scriptstyle s$}}
    \put(105,-5){\framebox(30,20){\text{EC}}}
    \put(135,5){\line(1,0){15}}
    \put(137,10){\makebox(0,0)[bl]{$\scriptstyle s_t$}}
    \put(150,-5){\line(0,1){20}}
    \put(150,15){\line(1,-1){10}}
    \put(160,5){\line(-1,-1){10}}
    \thinlines
    \put(160,5){\line(1,0){10}}
    \end{picture}
    =
    \begin{picture}(120, 20)
        \thicklines
    \put(10,5){\line(1,0){15}}
    \put(25,-5){\framebox(30,20){\text{EC}}}
    \put(55,5){\line(1,0){15}}
    \put(57,10){\makebox(0,0)[bl]{$\scriptstyle s_l$}}
    \put(70,-5){\line(0,1){20}}
    \put(70,15){\line(1,-1){10}}
    \put(80,5){\line(-1,-1){10}}
    \thinlines
    \put(80,5){\line(1,0){10}}
    \put(100,5){\circle{20}}
    \put(100,5){\makebox(0,0){$U$}}
    \put(110,5){\line(1,0){15}}
    \end{picture}
    \label{eq:gate_correct}
    \end{equation}
    \\
    A GKP-preparation ExRec is correct if
    \begin{equation}
    \begin{picture}(110,20)
    \thicklines
    \put(20,5){\oval(40,20)[l]}
    \put(20,-5){\line(0,1){20}}
    \put(22,10){\makebox(0,0)[bl]{$\scriptstyle s$}}
    \put(20,5){\line(1,0){15}}
    \put(35,-5){\framebox(30,20){\text{EC}}}
    \put(65,5){\line(1,0){15}}
    \put(67,10){\makebox(0,0)[bl]{$\scriptstyle s_t$}}
    \put(80,-5){\line(0,1){20}}
    \put(80,15){\line(1,-1){10}}
    \put(90,5){\line(-1,-1){10}}
    \thinlines
    \put(90,5){\line(1,0){10}}
    \end{picture}
    =
    \begin{picture}(50,20)
    \thinlines
    \put(20,5){\oval(20,20)[l]}
    \put(20,-5){\line(0,1){20}}
    \put(20,5){\line(1,0){15}}
    \end{picture}
    \end{equation}
    \\
    A GKP-measurement ExRec is correct if
    \begin{equation}
        \begin{picture}(105,20)
    \thicklines
    \put(10,5){\line(1,0){15}}
    \put(25,-5){\framebox(30,20){\text{EC}}}
    \put(55,5){\line(1,0){15}}
    \put(57,10){\makebox(0,0)[bl]{$\scriptstyle s_l$}}
    \put(70,5){\oval(40,20)[r]}
    \put(70,-5){\line(0,1){20}}
    \put(91,10){\makebox(0,0)[bl]{$\scriptstyle s$}}
    \end{picture}
    =
       \begin{picture}(130,20)
    \thicklines
    \put(10,5){\line(1,0){15}}
    \put(25,-5){\framebox(30,20){\text{EC}}}
    \put(55,5){\line(1,0){15}}
    \put(57,10){\makebox(0,0)[bl]{$\scriptstyle s_l$}}
    \put(70,-5){\line(0,1){20}}
    \put(70,15){\line(1,-1){10}}
    \put(80,5){\line(-1,-1){10}}
    \thinlines
    \put(80,5){\line(1,0){10}}
    \put(90,5){\oval(20,20)[r]}
    \put(90,-5){\line(0,1){20}}
    \end{picture}
    \label{eq:measurement_correct}
    \end{equation}
\end{definition}
\medskip

We also define the goodness/badness of the GKP ExRec as has been done in Ref.~\cite{Gottesman2009}.
\begin{definition}[Good or bad GKP ExRec] \label{def:good_gkp_exrec}
    A GKP preparation ExRec is good if it consists of an $s$-preparation gadget followed by an $s_t$-EC gadget satisfying $s+s_t<\frac{c}{2}$.  A GKP single-qubit gate ExRec is good if it consists of a leading $s_l$-EC gadget, an $s$-gate, and a trailing $s_t$-EC gadget satisfying $s_l+s+s_t<\frac{c}{2}$.  A GKP two-qubit gate ExRec is good if it consists of leading $s_l$-EC gadgets, an $s$-gate, and trailing $s_t$-EC gadgets satisfying $2s_l+s+s_t<\frac{c}{2}$.  A GKP measurement ExRec is good if it consists of an $s_l$-EC gadget followed by an $s$-measurement gadget satisfying $s_l + s < \frac{c}{2}$.  A GKP ExRec is called bad if it is not good.
\end{definition}
Using the same argument as the qubit case in Ref.~\cite{Gottesman2009}, we see that the goodness of the GKP ExRec is a sufficient condition for its correctness, as shown in the following lemma.
\begin{lemma}[Good implies correct for the ideal GKP decoder]\label{lem:good_implies_correct}
    If a GKP ExRec is good, then it is correct for the ideal GKP decoder.
\end{lemma}
\begin{proof}
    Let us use a good GKP single-qubit gate ExRec as an example, and the other cases can be proved with the same reasoning.
    From Proposition~\ref{prop:s-gate} and \ref{prop:s-EC}, the following chain of equality holds.
    \begin{align*}
        &\begin{picture}(185,20)
    \thicklines
    \put(10,5){\line(1,0){15}}
    \put(25,-5){\framebox(30,20){\text{EC}}}
    \put(55,5){\line(1,0){15}}
    \put(57,10){\makebox(0,0)[bl]{$\scriptstyle s_l$}}
    \put(80,5){\circle{20}}
    \put(80,5){\makebox(0,0){$\overline{U}$}}
    \put(90,5){\line(1,0){20}}
    \put(91,10){\makebox(0,0)[bl]{$\scriptstyle s$}}
    \put(110,-5){\framebox(30,20){\text{EC}}}
    \put(140,5){\line(1,0){15}}
    \put(142,10){\makebox(0,0)[bl]{$\scriptstyle s_t$}}
    \put(155,-5){\line(0,1){20}}
    \put(155,15){\line(1,-1){10}}
    \put(165,5){\line(-1,-1){10}}
    \thinlines
    \put(165,5){\line(1,0){10}}
    \end{picture}\\
    &= 
    \begin{picture}(185,40)
    \thicklines
    \put(10,5){\line(1,0){15}}
    \put(25,-5){\framebox(30,20){\text{EC}}}
    \put(55,5){\line(1,0){15}}
    \put(57,10){\makebox(0,0)[bl]{$\scriptstyle s_l$}}
    \put(70,-5){\framebox(5,20)}
    \put(75,5){\line(1,0){15}}
    \put(77,10){\makebox(0,0)[bl]{$\scriptstyle s_l$}}
    \put(100,5){\circle{20}}
    \put(100,5){\makebox(0,0){$\overline{U}$}}
    \put(110,5){\line(1,0){20}}
    \put(111,10){\makebox(0,0)[bl]{$\scriptstyle s$}}
    \put(130,-5){\framebox(30,20){\text{EC}}}
    \put(160,5){\line(1,0){15}}
    \put(162,10){\makebox(0,0)[bl]{$\scriptstyle s_t$}}
    \put(175,-5){\line(0,1){20}}
    \put(175,15){\line(1,-1){10}}
    \put(185,5){\line(-1,-1){10}}
    \thinlines
    \put(185,5){\line(1,0){10}}
    \end{picture}\\
    &= 
    \begin{picture}(240,40)
    \thicklines
    \put(10,5){\line(1,0){15}}
    \put(25,-5){\framebox(30,20){\text{EC}}}
    \put(55,5){\line(1,0){15}}
    \put(57,10){\makebox(0,0)[bl]{$\scriptstyle s_l$}}
    \put(70,-5){\framebox(5,20)}
    \put(75,5){\line(1,0){15}}
    \put(77,10){\makebox(0,0)[bl]{$\scriptstyle s_l$}}
    \put(100,5){\circle{20}}
    \put(100,5){\makebox(0,0){$\overline{U}$}}
    \put(110,5){\line(1,0){15}}
    \put(111,10){\makebox(0,0)[bl]{$\scriptstyle s$}}
    \put(125,-5){\framebox(5,20)}
    \put(130,5){\line(1,0){30}}
    \put(132,10){\makebox(0,0)[bl]{$\scriptstyle s_l+s$}}
    \put(160,-5){\framebox(30,20){\text{EC}}}
    \put(190,5){\line(1,0){15}}
    \put(192,10){\makebox(0,0)[bl]{$\scriptstyle s_t$}}
    \put(205,-5){\line(0,1){20}}
    \put(205,15){\line(1,-1){10}}
    \put(215,5){\line(-1,-1){10}}
    \thinlines
    \put(215,5){\line(1,0){10}}
    \end{picture}\\
    &= 
    \begin{picture}(240,40)
    \thicklines
    \put(10,5){\line(1,0){15}}
    \put(25,-5){\framebox(30,20){\text{EC}}}
    \put(55,5){\line(1,0){15}}
    \put(57,10){\makebox(0,0)[bl]{$\scriptstyle s_l$}}
    \put(70,-5){\framebox(5,20)}
    \put(75,5){\line(1,0){15}}
    \put(77,10){\makebox(0,0)[bl]{$\scriptstyle s_l$}}
    \put(100,5){\circle{20}}
    \put(100,5){\makebox(0,0){$\overline{U}$}}
    \put(110,5){\line(1,0){15}}
    \put(111,10){\makebox(0,0)[bl]{$\scriptstyle s$}}
    \put(125,-5){\framebox(5,20)}
    \put(130,5){\line(1,0){25}}
    \put(132,10){\makebox(0,0)[bl]{$\scriptstyle s_l+s$}}
    \put(155,-5){\line(0,1){20}}
    \put(155,15){\line(1,-1){10}}
    \put(165,5){\line(-1,-1){10}}
    \thinlines
    \put(165,5){\line(1,0){10}}
    \end{picture}\\
    &=
    \begin{picture}(120, 40)
        \thicklines
    \put(10,5){\line(1,0){15}}
    \put(25,-5){\framebox(30,20){\text{EC}}}
    \put(55,5){\line(1,0){15}}
    \put(57,10){\makebox(0,0)[bl]{$\scriptstyle s_l$}}
    \put(70,-5){\framebox(5,20)}
    \put(75,5){\line(1,0){15}}
    \put(77,10){\makebox(0,0)[bl]{$\scriptstyle s_l$}}
    \put(90,-5){\line(0,1){20}}
    \put(90,15){\line(1,-1){10}}
    \put(100,5){\line(-1,-1){10}}
    \thinlines
    \put(100,5){\line(1,0){10}}
    \put(120,5){\circle{20}}
    \put(120,5){\makebox(0,0){$U$}}
    \put(130,5){\line(1,0){15}}
    \end{picture}\\
    &=
    \begin{picture}(120, 40)
        \thicklines
    \put(10,5){\line(1,0){15}}
    \put(25,-5){\framebox(30,20){\text{EC}}}
    \put(55,5){\line(1,0){15}}
    \put(57,10){\makebox(0,0)[bl]{$\scriptstyle s_l$}}
    \put(70,-5){\line(0,1){20}}
    \put(70,15){\line(1,-1){10}}
    \put(80,5){\line(-1,-1){10}}
    \thinlines
    \put(80,5){\line(1,0){10}}
    \put(100,5){\circle{20}}
    \put(100,5){\makebox(0,0){$U$}}
    \put(110,5){\line(1,0){15}}
    \end{picture}
    \end{align*}
    In the above, we used the fact that $s_l+s+s_t<\frac{c}{2}$, which is from the goodness of this ExRec.  Thus, we reach the correctness of the ExRec.
\end{proof}

It is thus obvious that if the FT circuit contains only good GKP ExRecs, the output distribution of the FT protocol $C'$ with all the measurement outcomes of EC gadgets ignored is the same as the output distribution of the initial qubit circuit $C$.  
Even though it is unlikely that all the GKP ExRecs are good, we can use the qubit concatenated code to suppress the error at the higher level of the concatenation so that the frequency of bad ExRecs at the concatenation level $k$ is arbitrarily small~\cite{Gottesman2009}.  To show this, however, we need to specify how a bad ExRec should be interpreted in a higher-level (qubit-level) circuit.
Unfortunately, the corresponding logical qubit picture for a bad GKP ExRec cannot be determined locally since it may depend on the left-over errors in the previous ExRecs.  We thus need to see larger contexts to determine the erroneous logical qubit operation. 
For this, we introduce the GKP $*$-decoder ${\cal D}^*_{\rm GKP}$ by generalizing the $*$-decoder in Ref.~\cite{Gottesman2009}.
The GKP $*$-decoder keeps the syndrome subsystem ${\cal H}_S$ instead of tracing it out as is done in the ideal decoder Eq.~\eqref{eq:ideal_decoder}, i.e.,
\begin{equation}
    {\cal D}^*_{\rm GKP}(\hat{\rho}) = \sum_{\mu,\nu\in\{0,1\}}\iint_{|z_1|,|z_2|<\frac{c}{2}}dz_1dz_2\iint_{|z'_1|,|z'_2|<\frac{c}{2}}dz'_1 dz'_2 \ket{\mu;z_1, z_2}\!\bra{\mu;z_1,z_2}_{LS} \hat{\rho}  \ket{\nu;z'_1, z'_2}\!\bra{\nu;z'_1,z'_2}_{LS}. \label{eq:star_decoder}
\end{equation}
One can check from Eq.~\eqref{eq:completeness_SSS} that ${\cal D}^*_{\rm GKP}$ is nothing but the identity channel, representing a state in the SSS Hilbert space.
Diagrammatically, the GKP $*$-decoder is denoted as follows.
\begin{equation*}
\begin{picture}(100,30)
\thicklines
        \put(0,20){\line(1,0){10}}
\put(10,10){\line(0,1){20}}
\put(10,30){\line(1,-1){10}}
\put(20,20){\line(-1,-1){10}}
\thinlines
\put(20,20){\line(1,0){10}}
\put(15,5){\line(0,1){10}}
\put(15,5){\line(1,0){15}}
\put(45,10){\makebox(55,16)[l]{\text{GKP $*$-decoder}}}
\end{picture}
\end{equation*}
In this diagram, the line at the top denotes the logical qubit $L$ and the line at the bottom denotes the syndrome subsystem $S$.  In this sense, the stabilizer subsystem decomposition~\cite{Mackenzie2022} is nothing but the $*$-decoder in our analysis.
For a bad GKP-gate ExRec, we thus have the following.
\begin{equation}
    \begin{picture}(180,20)
    \thicklines
    \put(10,5){\line(1,0){15}}
    \put(25,-5){\framebox(30,20){\text{EC}}}
    \put(55,5){\line(1,0){15}}
    \put(57,10){\makebox(0,0)[bl]{$\scriptstyle s_l$}}
    \put(80,5){\circle{20}}
    \put(80,5){\makebox(0,0){$\overline{U}$}}
    \put(90,5){\line(1,0){15}}
    \put(91,10){\makebox(0,0)[bl]{$\scriptstyle s$}}
    \put(105,-5){\framebox(30,20){\text{EC}}}
    \put(135,5){\line(1,0){15}}
    \put(137,10){\makebox(0,0)[bl]{$\scriptstyle s_t$}}
    \put(150,-5){\line(0,1){20}}
    \put(150,15){\line(1,-1){10}}
    \put(160,5){\line(-1,-1){10}}
    \thinlines
    \put(155,-10){\line(0,1){10}}
    \put(155,-10){\line(1,0){15}}
    \put(160,5){\line(1,0){10}}
    \end{picture}
    =
    \begin{picture}(100, 20)
        \thicklines
    \put(10,5){\line(1,0){15}}
    \put(25,-5){\line(0,1){20}}
    \put(25,15){\line(1,-1){10}}
    \put(35,5){\line(-1,-1){10}}
    \thinlines
    \put(30,-10){\line(0,1){10}}
    \put(30,-10){\line(1,0){15}}
    \put(35,5){\line(1,0){10}}
    \put(55,5){\oval(20,20)[t]}
    \put(55,-10){\oval(20,20)[b]}
    \put(45,-10){\line(0,1){15}}
    \put(65,-10){\line(0,1){15}}
    \put(55,-2.5){\makebox(0,0){${\cal U}'$}}
    \put(65,5){\line(1,0){10}}
    \put(65,-10){\line(1,0){10}}
    \end{picture}
    \label{eq:bad_gate_exrec}
\end{equation}
\\
The noisy interaction (channel) ${\cal U}'$ depends on the gate $\overline{U}$ and noise during the ExRec.
This time, the GKP $*$-decoder absorbs the entire GKP ExRec (including the leading EC gadgets, unlike the good GKP ExRecs where the decoder leaves the leading EC gadgets as shown in Lemma~\ref{lem:good_implies_correct}) to avoid the correlated errors between consecutive gates in $C$.
Similar relations hold also for bad GKP-measurement ExRecs and preparation ExRecs.
For a bad GKP-measurement ExRec, a qubit Pauli-$Z$/$X$ measurement is performed on the logical qubit $L$ while the syndrome subsystem $S$ is traced out just after the noise operator entangles two subsystems $L$ and $S$.
For a bad GKP-preparation ExRec, the operator entangles the logical qubit $L$ and the syndrome subsystem $S$ immediately after the preparations of a logical qubit state in ${\cal H}_L$ and an arbitrary state in ${\cal H}_S$, or equivalently, an entangled state in ${\cal H}_L\otimes {\cal H}_S$ is prepared.  Note that the form of ${\cal U}'$ implicitly assumes the independence and the Markovianity of the noise.  Otherwise, ${\cal U}'$ may act on several locations at the same time (if not independent) and may depend on an environmental state (if not Markovian).  The same applies to the following arguments.

Now, we need to redefine the correctness of the GKP ExRec for the GKP $*$-decoder as follows.
\begin{definition}[Correctness for the GKP $*$-decoder]\label{def:correctness_star_decoder}
    The GKP gate ExRec is correct if it satisfies the following.
    \begin{align}
        \begin{picture}(180,20)
    \thicklines
    \put(10,5){\line(1,0){15}}
    \put(25,-5){\framebox(30,20){\text{EC}}}
    \put(55,5){\line(1,0){15}}
    \put(57,10){\makebox(0,0)[bl]{$\scriptstyle s_l$}}
    \put(80,5){\circle{20}}
    \put(80,5){\makebox(0,0){$\overline{U}$}}
    \put(90,5){\line(1,0){15}}
    \put(91,10){\makebox(0,0)[bl]{$\scriptstyle s$}}
    \put(105,-5){\framebox(30,20){\text{EC}}}
    \put(135,5){\line(1,0){15}}
    \put(137,10){\makebox(0,0)[bl]{$\scriptstyle s_t$}}
    \put(150,-5){\line(0,1){20}}
    \put(150,15){\line(1,-1){10}}
    \put(160,5){\line(-1,-1){10}}
    \thinlines
    \put(160,5){\line(1,0){10}}
    \put(155,-10){\line(0,1){10}}
    \put(155,-10){\line(1,0){15}}
    \end{picture}
    &=
    \begin{picture}(120, 20)
        \thicklines
    \put(10,5){\line(1,0){15}}
    \put(25,-5){\framebox(30,20){\text{EC}}}
    \put(55,5){\line(1,0){15}}
    \put(57,10){\makebox(0,0)[bl]{$\scriptstyle s_l$}}
    \put(70,-5){\line(0,1){20}}
    \put(70,15){\line(1,-1){10}}
    \put(80,5){\line(-1,-1){10}}
    \thinlines
    \put(80,5){\line(1,0){10}}
    \put(100,5){\circle{20}}
    \put(100,5){\makebox(0,0){$U$}}
    \put(75,-15){\line(0,1){15}}
    \put(75,-15){\line(1,0){15}}
    \put(90,-20){\framebox(20,10)}
    \put(100,-15){\makebox(0,0){${\cal V}$}}
    \put(110,-15){\line(1,0){15}}
    \put(110,5){\line(1,0){15}}
    \end{picture}
    \label{eq:good_star_decoder}
    \\
    &\nonumber 
    \end{align}
    In the above, ${\cal V}$ is a CPTP map acting only on the syndrome subsystem $S$, which depends on the gate $\overline{U}$ and the noise during the ExRec\@.
    Based on this principle, analogous definitions extending Definitions~\ref{def:correctness_ideal_decoder} exist for a correct GKP-measurement ExRec and a correct GKP-preparation ExRec, as well.
\end{definition}
\noindent The important thing in the above definition is that the correct logical gate~$\hat{U}$ is applied and no correlations are introduced between the logical qubit and the syndrome subsystem.
For a good GKP ExRec, we have the following.
\begin{lemma}[Good implies correct for the $*$-decoder]
    If the GKP ExRec is good, then it is also correct for the $*$-decoder.
\end{lemma}
\begin{proof}
    Notice that the ideal GKP decoder corresponds to the GKP $*$-decoder with the syndrome subsystem $S$ ignored. Since the goodness of a GKP ExRec implies the correctness for the ideal GKP decoder, the logical subsystem $L$ must be decoupled from the syndrome subsystem $S$.  Then, it can be seen that Eq.~\eqref{eq:good_star_decoder} is the most general form of such decoupled dynamics.
\end{proof}

So far, we have given how good and bad ExRecs can be interpreted in the qubit picture.  However, since we forced the GKP $*$-decoder to erase all the gadgets within a bad GKP ExRec, the leading consecutive ExRec lacks the trailing GKP EC gadget.
Thus, we need the following definition.
\begin{definition}[Truncated GKP ExRec]
    A GKP ExRec missing one or more trailing ECs is called a truncated GKP ExRec (except for GKP measurement ExRecs).  A truncated GKP gate ExRec is correct for the GKP $*$-decoder (and thus for the ideal GKP decoder) if it satisfies the following.
    \begin{align}
        \begin{picture}(140,20)
    \thicklines
    \put(10,5){\line(1,0){15}}
    \put(25,-5){\framebox(30,20){\text{EC}}}
    \put(55,5){\line(1,0){15}}
    \put(57,10){\makebox(0,0)[bl]{$\scriptstyle s_l$}}
    \put(80,5){\circle{20}}
    \put(80,5){\makebox(0,0){$\overline{U}$}}
    \put(90,5){\line(1,0){20}}
    \put(91,10){\makebox(0,0)[bl]{$\scriptstyle s$}}
    \put(110,-5){\line(0,1){20}}
    \put(110,15){\line(1,-1){10}}
    \put(120,5){\line(-1,-1){10}}
    \thinlines
    \put(120,5){\line(1,0){10}}
    \put(115,-10){\line(0,1){10}}
    \put(115,-10){\line(1,0){15}}
    \end{picture}
    &=
    \begin{picture}(120, 20)
        \thicklines
    \put(10,5){\line(1,0){15}}
    \put(25,-5){\framebox(30,20){\text{EC}}}
    \put(55,5){\line(1,0){15}}
    \put(57,10){\makebox(0,0)[bl]{$\scriptstyle s_l$}}
    \put(70,-5){\line(0,1){20}}
    \put(70,15){\line(1,-1){10}}
    \put(80,5){\line(-1,-1){10}}
    \thinlines
    \put(80,5){\line(1,0){10}}
    \put(100,5){\circle{20}}
    \put(100,5){\makebox(0,0){$U$}}
    \put(75,-15){\line(0,1){15}}
    \put(75,-15){\line(1,0){15}}
    \put(90,-20){\framebox(20,10)}
    \put(100,-15){\makebox(0,0){${\cal V}$}}
    \put(110,-15){\line(1,0){15}}
    \put(110,5){\line(1,0){15}}
    \end{picture}
    \label{eq:truncated_star_decoder}
    \\
    &\nonumber 
    \end{align}
    We can similarly define the correctness for the truncated GKP preparation ExRec.
\end{definition}

In the same way as the full GKP ExRec, we define the goodness and badness of the truncated GKP ExRecs\@.
\begin{definition}
    A truncated GKP preparation ExRec is good if it consists of an $s$-preparation gadget with $s<\frac{c}{2}$.  A truncated GKP gate ExRec is good if it consists of a leading $s_l$-EC gadget and an $s$-gate with $s_l+s<\frac{c}{2}$.  A truncated GKP ExRec is bad if it is not good.
\end{definition}
One can notice that the good GKP ExRec is also good when truncated, which will be used later.  One can derive the following with the same procedures as the previous discussions (and thus we omit the proof).
\begin{lemma}
    A good truncated GKP ExRec is correct for the GKP $*$-decoder (and thus for the ideal GKP decoder).
\end{lemma}

We can now reduce our noisy CV quantum circuit consisting of multiple GKP ExRecs to a qubit-level quantum circuit.
\begin{definition}[Noisy quantum circuit defined through the GKP code]\label{def:noisy_reduction}
    Given an FT-GKP CV circuit $C'$ for implementing an ideal qubit circuit $C$, the circuit $\tilde{C}$ on qubits and environmental systems is defined by moving the GKP $*$-decoder from the end to the start of the circuit $C'$ as in Eqs.~\eqref{eq:bad_gate_exrec},~\eqref{eq:good_star_decoder}, and~\eqref{eq:truncated_star_decoder} and regard the syndrome subsystems as environmental systems.
\end{definition}

We finally come to the point of stating our level-reduction theorem.
Recall that we have defined our noise model using three parameters, i.e., $s$, $\epsilon$, and $E$ in Defs.~\ref{def:s_eps_Markovian_prep}--\ref{def:E_s_eps_Markovian_meas}.
The essence of our analysis is to take sufficiently small constant $s$ so that the FT-GKP CV circuit $C'$ should be able to simulate the original qubit circuit $C$ correctly in the ideal case (i.e., in the case of $\epsilon=0$), and then take into account the effect of actual nonzero $\epsilon$ in a non-ideal case by perturbatively evaluating the norm of each fault path for $\tilde{C}$.
In particular, we choose the parameters $s_{\rm p}$, $s_{\rm g}$, $s_{\rm m}$, and $s_{\rm e}$ to satisfy the following conditions:
\begin{align}
    2s_{\rm p} + 2s_{\rm g} + \max\{s_{\rm g} + s_{\rm m}, 2s_{\rm g}\} &= s_{\rm e}, \label{eq:EC_from_pgm} \\
    s_{\rm p} + s_{\rm e} &< \frac{c}{2}, \label{eq:prep+EC} \\
    3s_{\rm e} + s_{\rm g} &< \frac{c}{2}, \\
    s_{\rm e} + s_{\rm m} &< \frac{c}{2}. \label{eq:EC+meas} 
\end{align}
These conditions ensure that all the GKP preparation ExRecs, gate ExRecs (including the two-mode gate), and measurement ExRecs composed of $s_{\rm p}$-preparations, $s_{\rm g}$-gates, and $s_{\rm m}$-measurements in the circuit $C'$ are good.
Then, the evaluation of errors in the qubit circuit $\tilde{C}$ reduces to the evaluation of the deviation from this \quoted{idealized} situation in the CV circuit $C'$ due to the effect of nonzero $\epsilon$.

To evaluate this effect, let $E_{\max}^\ell(\epsilon)$ be the function of $\epsilon\in(0,1)$ defined as
\begin{align}
        E_{\max}^\ell(\epsilon) \coloneqq g_{\rm sup}\left(\frac{g_{\rm sup}^{\ell-1}(E_{\rm prep})}{\epsilon^2}\right),
        \label{eq:defining_eq_E_max}
\end{align}
where $E_{\rm prep}$ and $g_{\rm sup}$ are a positive constant and a positive, monotonically increasing, locally bounded function, respectively, which define the energy-constraint conditions in Defs.~\ref{def:energy_constrained_prep} and \ref{def:energy_constrained_gate}, and $\ell$ is the maximum number of gate gadgets that a quantum state prepared during the computation undergoes before it is measured as defined in Prop.~\ref{prop:maximum_energy_ExRec}. 
Notice that $E_{\max}^\ell(\epsilon)$ increases when $\epsilon$ decreases. 
Then, we prove the following theorem.

\begin{theorem}[Level reduction]\label{theo:level_reduction}
    Let $C'$ be an FT-GKP circuit on CV systems for implementing a circuit $C$ on qubits. Suppose that all preparation gadgets and gate gadgets in $C'$ satisfy an $E_{\rm prep}$-energy constraint and a $g_{\rm sup}$-energy constraint, respectively (Defs.~\ref{def:energy_constrained_prep} and~\ref{def:energy_constrained_gate}).  Also, for $\epsilon$ satisfying $0<\epsilon<1$ and $s_{\rm p}$, $s_{\rm g}$, and $s_{\rm m}$ satisfying Eqs.~\eqref{eq:EC_from_pgm}--\eqref{eq:EC+meas}, suppose that the CV circuit $C'$ experiences $(s_{\rm p},\epsilon)$-independent Markovian noise for state preparations, $(E_{\max}^{\ell}(\epsilon),s_{\rm g},\epsilon)$-independent Markovian noise for gates (including waits), and $(E_{\max}^{\ell}(\epsilon),s_{\rm m},\epsilon)$-independent Markovian noise for measurements (Defs.~\ref{def:s_eps_Markovian_prep}--\ref{def:E_s_eps_Markovian_meas}), where an $s_{\rm p}$-preparation and an $s_{\rm g}$-gate inside the definitions of these noises also satisfy an $E_{\rm prep}$-energy constraint and a $g_{\rm sup}$-energy constraint, respectively.
    For a nonempty subset $R\subseteq{\cal I}$ of the indices of the locations in $C$, let $\hat{\rho}_R$ be the sum of fault paths of the quantum computation with faults applied to at least one location in the $i^{\rm th}$ truncated GKP ExRec in $C'$ (i.e., corresponding to the location $C_i$ in $C$) for each $i\in R$, where faults are described in the sense of Eq.~\eqref{eq:fault_paths}.  For any such sum of fault paths $\hat{\rho}_R$ in $C'$, define a fault path in a circuit $\tilde{C}$ on qubits and environmental systems as follows: replace the location $C_i$ of $C$ for each $i\in R$ with the corresponding fault map acting both on the qubits and the environmental subsystems in $C'$, which can, for example, be described as the difference between ${\cal U}'$ in Eq.~\eqref{eq:bad_gate_exrec} and ${\cal U}\otimes{\cal V}$ in Eq.~\eqref{eq:good_star_decoder}, and keep the others (i.e., $C_i$ for $i\in{\cal I}\setminus R$) unchanged on the qubits while acting appropriate channels on the environmental systems as in Eq.~\eqref{eq:good_star_decoder}.  Then, the circuit $\tilde{C}$ experiences a local Markovian noise model (Def.~\ref{def:local_Markovian}).  Furthermore, an upper bound $\epsilon_\mathrm{qubit}$ on the noise strength at any location in the circuit $\tilde{C}$ is given by
    \begin{equation}
        \epsilon_\mathrm{qubit} = 10 \epsilon L_{\max},
    \end{equation}
    where $L_{\max}$ denotes the maximum number of GKP preparation, gate, and measurement gadgets in any truncated GKP ExRec in $C'$.
\end{theorem}
\begin{proof}
    Our proof is based on evaluating the effect of nonzero $\epsilon$, applying the techniques in the analysis of threshold theorem under independent Markovian noises on qubits~\cite{Aliferis2013} to the GKP code.
    Since $E^{\ell}_{\max}(\epsilon)$ is a function  of $\epsilon$, the noise model here has two types of parameters; i.e., one is $s_{\rm p}$, $s_{\rm g}$, $s_{\rm m}$, and $s_{\rm e}$, and the other is $\epsilon$.
    As is clear from the correctness conditions of the GKP ExRecs, the GKP code can tolerate and correct the effect of nonzero $s_{\rm p}$, $s_{\rm g}$, $s_{\rm m}$, and $s_{\rm e}$; in particular, with $s_{\rm p}$, $s_{\rm g}$, $s_{\rm m}$, and $s_{\rm e}$ satisfying Eqs.~\eqref{eq:EC_from_pgm}--\eqref{eq:EC+meas}, if $\epsilon$ were zero, then the FT-GKP circuit $C'$ for $C$ that consists only of $s_{\rm p}$-preparations, $s_{\rm g}$-gates, $s_{\rm m}$-measurements, and $s_{\rm e}$-ECs would output the same probability distribution as the ideal qubit circuit $C$ since all the GKP ExRecs become good (see Def.~\ref{def:good_gkp_exrec}). 
    By contrast, it turns out that nonzero $\epsilon$ may directly lead to an error on the logical qubit of the GKP code.
    In this sense, our analysis of the effect of $\epsilon$ on the GKP code is analogous to considering the distance-1 code on qubits in the level reduction under the independent Markovian noise model~\cite{Aliferis2013}; however, the difference between Ref.~\cite{Aliferis2013} and our analysis of GKP code arises since we also need to use bounds on the energy-constrained diamond norm introduced in Sec.~\ref{sec:energy_constrained_diamond_norm}.
    
    In our case, each fault in the fault path of $\tilde{C}$ can be regarded as the difference between an ideal preparation, gate, or measurement in $C$ and the corresponding erroneous ones, such as ${\cal U}'$ in Eq.~\eqref{eq:bad_gate_exrec}, which are given by moving the GKP $*$-decoder through the whole circuit as in Def.~\ref{def:noisy_reduction}.
    Thus, one fault in a fault path of the CV circuit $C'$, which is the difference between the actual noisy gate (resp.~preparation, measurement) and a corresponding $s_{\rm g}$-gate (resp.~$s_{\rm p}$-preparation, $s_{\rm m}$-measurement), may immediately cause a fault in $\tilde{C}$ (which is akin to the distance-1 code on qubits in the sense that no fault is tolerated).
    In this case, given a fault path, if the trailing EC gadget of a full GKP ExRec includes faults, then these faults are always taken into account as those in the leading EC gadget of the successive GKP ExRec due to the rule in Def.~\ref{def:noisy_reduction}, and thus the given GKP ExRec becomes a truncated GKP ExRec; therefore, in our analysis (similar to the analysis of the distance-1 code on qubits), it suffices to count faults in truncated GKP ExRecs.
    In other words, a fault in $\tilde{C}$ may occur if at least one of the CV preparations, gates, and measurements in the corresponding truncated GKP ExRec (i.e., the leading EC gadgets or the location $C'_i$ corresponding to $C_i$ in $C$) is not an $E_{\rm prep}$-energy-constrained $s_{\rm p}$-preparation, a $g_{\rm sup}$-energy-constrained $s_{\rm g}$-gate, and an $s_{\rm m}$-measurement, respectively.
    We will thus bound, for a nonempty subset $R\subseteq {\cal I}$, the norm of the sum $\hat{\rho}_R$ of fault paths in $C'$ in which faults apply to at least one location in the $i^{\rm th}$ truncated GKP ExRec in $C'$ for each $i\in R$, which then corresponds to the fault path in $\tilde{C}$ where faults apply to $\{\tilde{C}_i\}_{i\in R}$.

    Let $\{C'_j\}_{j\in{\cal I}'}$ be a chronologically ordered set of locations in $C'$. (Recall that each of these may be a preparation, gate, or measurement.)
    For brevity in what follows, we define a composite map realized at locations from $C'_a$ to $C'_b$ by
    \begin{align}
        {\cal O}_{b:a}
    \coloneqq
        \prod_{j=a}^b {\cal O}_j
    =
        {\cal O}_{b}
    \circ
        {\cal O}_{b-1}
    \circ
        \dotsm
    \circ
        {\cal O}_{a+1}
    \circ
        {\cal O}_{a}
    \end{align}
    with the composition product taken in the appropriate order, starting with~$a$ (on the right) and ending with~$b$ (on the left). We will use two version of this definition, ${\cal O}^{\rm noisy}_{b:a}$ and ${\cal O}^{\rm ideal}_{b:a}$, where an individual CPTP map ${\cal O}^{\rm noisy}_j$ (resp.~${\cal O}^{\rm ideal}_j$) with $j\in{\cal I}'$ corresponds to either an $E_{\rm prep}$-energy-constrained actual preparation (resp.~$s_{\rm p}$-preparation), a $g_{\rm sup}$-energy-constrained actual gate (resp.~$s_{\rm g}$-gate), or an actual measurement (resp.~$s_{\rm m}$-measurement) at the $j^{\text{th}}$ location in $C'$.  Let us further define the fault map ${\cal F}_j$ as
    \begin{equation}
        {\cal F}_j\coloneqq {\cal O}^{\rm noisy}_j - {\cal O}^{\rm ideal}_j. \label{eq:def_fault_map}
    \end{equation}
    From the definition of the $(s,\epsilon)$- and $(E,s,\epsilon)$-independent Markovian noise in Defs.~\ref{def:s_eps_Markovian_prep}--\ref{def:E_s_eps_Markovian_meas}, we have
    \begin{equation}
        \frac{1}{2}\|{\cal F}_j\|_{\diamond}^{E_{\max}^{\ell}(\epsilon)} \leq \epsilon,\label{eq:fault_diamond_norm}
    \end{equation}
    where we slightly abuse the notation when ${\cal O}_j^{\rm noisy}$ and ${\cal O}_j^{\rm ideal}$ are state-preparation maps by defining $\mathfrak{S}_{E}(\mathbb{C})\coloneqq\{1\}\eqqcolon\mathfrak{S}(\mathbb{C})$.
    Then, the sum $\hat{\rho}_R$ of fault paths defined above is given by
    \begin{equation}
        \hat{\rho}_R = \sum_{i\in R}\sum_{T_i\subseteq R'_i:|T_i|\geq 1}\bigl[\text{Fault }{\cal F}_j\text{ replaces } {\cal O}^{\rm ideal}_j \text{ in }{\cal O}^{\rm ideal}_{|{\cal I}'|:1} \text{ for all }j\in T_i\bigr],
    \end{equation}
    where $R'_i\subseteq{\cal I}'$ for $i\in R$ denotes the set of indices that correspond to the locations of $i^{\rm th}$ truncated GKP ExRec.

    Let $i_{\initial}$ be the initial (smallest) element of the ordered set $R$ and let $j_{\initial}$ and $j_{\final}$ be the initial and the final elements of $R'_{i_\initial}$, i.e., $R'_{i_\initial}=\{j_\initial,\ldots,j_\final\}$.  Then, by using the definition of ${\cal F}_j$ in Eq.~\eqref{eq:def_fault_map}, we have 
    \begin{align}
        \hat{\rho}_R 
        &={\cal O}'_{|{\cal I}'|:j_\final+1}\circ \left(\sum_{T\subseteq\{j_\initial,\ldots,j_\final\}:|T|\geq 1} \bigl[\text{Fault }{\cal F}_j\text{ replaces } {\cal O}^{\rm ideal}_j \text{ in }{\cal O}^{\rm ideal}_{j_\final:j_\initial} \text{ for all }j\in T \bigr] \right) \circ {\cal O}^{\rm ideal}_{j_\initial-1:1}
        \label{eq:quantity_to_bound_middle_as_sum}\\
        & = {\cal O}'_{|{\cal I}'|:j_\final+1}\circ \left({\cal O}^{\rm noisy}_{j_\final:j_\initial} - {\cal O}^{\rm ideal}_{j_\final:j_\initial} \right) \circ {\cal O}^{\rm ideal}_{j_\initial-1:1} , \label{eq:quantity_to_bound_middle}
    \end{align}
    where
    \begin{equation}
        {\cal O}'_{|{\cal I}'|:j_\final+1}\coloneqq \sum_{i\in R\setminus \{i_{\initial}\}}\sum_{T_i\subseteq R'_i:|T_i|\geq 1} \bigl[\text{Fault }{\cal F}_j \text{ replaces } {\cal O}_j^{\rm ideal} \text{ in }{\cal O}_{|{\cal I}|':j_\final+1}^{\rm ideal} \text{ for all }j\in T_i\bigr],
    \end{equation}
    and ${\cal O}_{j:j+1}\coloneqq\mathrm{Id}\eqqcolon{\cal O}_{j-1:j}$ for notational convenience (since, in both cases, the initial index is after the final index).
    By adding
    \begin{equation}
        0={\cal O}'_{|{\cal I}'|:j_\final+1} \circ \left(\sum_{j=j_\initial}^{j_\final-1} -{\cal O}^{\rm noisy}_{j_\final:j+1}\circ {\cal O}^{\rm ideal}_{j:j_\initial} + {\cal O}^{\rm noisy}_{j_\final:j+1}\circ {\cal O}^{\rm ideal}_{j:j_\initial}\right)\circ {\cal O}^{\rm ideal}_{j_\initial-1:1} 
    \end{equation}
    to Eq.~\eqref{eq:quantity_to_bound_middle}, we have\footnote{Another way to see this is to reconsider the difference ${{\cal O}^{\rm noisy}_{j_\final:j_\initial} - {\cal O}^{\rm ideal}_{j_\final:j_\initial}}$ in Eq.~\eqref{eq:quantity_to_bound_middle}, which is the sum in parentheses in Eq.~\eqref{eq:quantity_to_bound_middle_as_sum}. The first term, ${\cal O}^{\rm noisy}_{j_\final:j_\initial}$, is the sum over all~$2^{\abs{T}}$ fault paths, including the no-fault path. The only thing that subtracting the second term, ${\cal O}^{\rm ideal}_{j_\final:j_\initial}$, is doing is guaranteeing at least one fault. And this means there must be a location of the \emph{first} fault, which we call~$j$. Now, group the terms in this sum according to this location. By definition, (1)~there are no faults before this location (${\cal O}^{\rm ideal}_{j-1:j_\initial}$), (2)~${\cal F}_j$ has replaced the ideal operation at this location, and (3)~the rest of the locations may or may not have any faults, and we must sum over these possibilities (${\cal O}^{\rm noisy}_{j_\final:j+1}$). Composing these and summing over all first-fault locations~$j$ gives the result in Eq.~\eqref{eq:quantity_to_bound}.}
    \begin{equation}
        \hat{\rho}_R = {\cal O}'_{|{\cal I}'|:j_\final+1}\circ \left(\sum_{j=j_\initial}^{j_\final} {\cal O}^{\rm noisy}_{j_\final:j+1} \circ {\cal F}_j \circ {\cal O}^{\rm ideal}_{j-1:j_\initial} \right) \circ {\cal O}^{\rm ideal}_{j_\initial-1:1}.\label{eq:quantity_to_bound}
    \end{equation}
    Now, we upper-bound the trace norm of $\hat{\rho}_R$.  (Recall that $\{{\cal O}_j\}_{j\in{\cal I}'}$ includes a state-preparation map.)  Using $\|{\cal U}(\hat{\rho})\|_1\leq \|{\cal U}\|^E_{\diamond}\|\hat{\rho}\|_1$ for any linear map ${\cal U}$ and any state $\hat{\rho}\in\mathfrak{S}_E({\cal H}_Q)$, we have
    \begin{align}
        \|\hat{\rho}_R\|_1&=\left\|{\cal O}'_{|{\cal I}'|:j_\final+1}\circ \left(\sum_{j=j_\initial}^{j_\final} {\cal O}^{\rm noisy}_{j_\final:j+1} \circ {\cal F}_j \circ {\cal O}^{\rm ideal}_{j-1:j_\initial} \right) \circ {\cal O}^{\rm ideal}_{j_\initial-1:1}\right\|_1 \\
        &\leq \left\|{\cal O}'_{|{\cal I}'|:j_\final+1}\circ \left(\sum_{j=j_\initial}^{j_\final} {\cal O}^{\rm noisy}_{j_\final:j+1} \circ {\cal F}_j \circ {\cal O}^{\rm ideal}_{j-1:j_\initial} \right) \right\|_{\diamond}^{g_{\rm sup}^{\ell-1}(E_{\rm prep})}\|{\cal O}^{\rm ideal}_{j_\initial-1:1}\|_1 \label{eq:output_of_ideal_sequence} \\
        &\leq\sum_{j=j_\initial}^{j_\final} \left\|{\cal O}'_{|{\cal I}'|:j_\final+1}\circ \left( {\cal O}^{\rm noisy}_{j_\final:j+1} \circ {\cal F}_j \circ {\cal O}^{\rm ideal}_{j-1:j_\initial} \right) \right\|_{\diamond}^{g_{\rm sup}^{\ell-1}(E_{\rm prep})} \\
        &\leq \sum_{j=j_\initial}^{j_\final}\left\|{\cal O}'_{|{\cal I}'|:j_\final+1}\circ \left( {\cal O}^{\rm noisy}_{j_\final:j+1} \circ {\cal F}_j \circ {\cal O}^{\rm ideal}_{j-1:j_\initial} \right) \right\|_{\diamond}^{g_{\rm sup}^{\ell-1}(E_{\rm prep})/\epsilon^2}, \label{eq:putting_out_E-diamond}
    \end{align}
    where the first inequality follows from the fact that an input of a GKP EC gadget (i.e., an output of the leading truncated GKP ExRec) is contained in $g_{\rm sup}^{\ell-1}(E_{\rm prep})$ as shown in Prop.~\ref{prop:maximum_energy_ExRec}, the second inequality from the triangle inequality and $\|{\cal O}^{\rm ideal}_{j_\initial-1:1}\|_1=1$, and the last inequality from the monotonicity of the energy-constrained diamond norm under energy increase in Eq.~\eqref{eq:monotonicity}.  
    Notice that
    \begin{align}
        {\cal O}^{\rm noisy}_{j_\final:j+1}  \circ {\cal O}^{\rm ideal}_{j:j_\initial}
        : \mathfrak{S}({\cal H}_Q)\rightarrow \mathfrak{S}_{g_{\rm sup}^{\ell-1}(E_{\rm prep})}({\cal H}_Q) \label{eq:actual_reset}
    \end{align}
    holds for any $j\in\{j_\initial-1,\ldots,j_\final\}$ from Props.~\ref{prop:energy_reset} and~\ref{prop:maximum_energy_ExRec}.  Furthermore, we have that
    \begin{equation}
        {\cal O}^{\rm ideal}_{j-1:j_\initial}:\mathfrak{S}_{g_{\rm sup}^{\ell-1}(E_{\rm prep})/\epsilon^2}({\cal H}_Q) \rightarrow \mathfrak{S}_{E_{\max}^{\ell}(\epsilon)}({\cal H}_Q) \label{eq:max_may_be_achieved}
    \end{equation}
    holds from Eq.~\eqref{eq:defining_eq_E_max} and the fact that the input of a GKP EC gadget undergoes only one gate before being measured and that energy of states newly prepared at the GKP EC gadget is bounded from above by $g_{\rm sup}^{\ell-1}(E_{\rm prep})\leq E_{\max}^{\ell}(\epsilon)$ for any $j\in\{j_\initial,\ldots,j_\final\}$. We thus have
    \begin{equation}
        \left\|{\cal O}^{\rm noisy}_{j_\final:j+1} \circ {\cal F}_j \circ {\cal O}^{\rm ideal}_{j-1:j_\initial} \right\|_{\diamond}^{g_{\rm sup}^{\ell-1}(E_{\rm prep})/\epsilon^2} \leq \|{\cal O}^{\rm noisy}_{j_\final:j+1}\circ {\cal F}_j\|_{\diamond}^{E_{\max}^{\ell}(\epsilon)} \|{\cal O}^{\rm ideal}_{j-1:j_\initial}\|_{\diamond}^{g_{\rm sup}^{\ell-1}(E_{\rm prep})/\epsilon^2} \leq \|{\cal O}^{\rm noisy}_{j_\final:j+1}\|_{\diamond} \|{\cal F}_j\|_{\diamond}^{E_{\max}^{\ell}(\epsilon)} \leq 2\epsilon,  \label{eq:2epsilon_bound}
    \end{equation}
    where we used Eqs.~\eqref{eq:submultiplicativity_1} and~\eqref{eq:max_may_be_achieved} in the first inequality, Eq.~\eqref{eq:submultiplicativity_2} in the second inequality, and Eq.~\eqref{eq:fault_diamond_norm} in the last inequality.
    Note also that $\|{\cal O}\|_\diamond=1=\|{\cal O}\|_\diamond^{E}$ holds for any CPTP map ${\cal O}$.
    Then, from Eqs.~\eqref{eq:def_fault_map}, \eqref{eq:actual_reset}, and~\eqref{eq:2epsilon_bound}, we can apply Corol.~\ref{cor:composition} to the pair ${\cal O}^{\rm noisy}_{j_\final:j}\circ{\cal O}^{\rm ideal}_{j-1:j_\initial}$ and ${\cal O}^{\rm noisy}_{j_\final:j+1}\circ{\cal O}^{\rm ideal}_{j:j_\initial}$ of CPTP maps in
    Eq.~\eqref{eq:putting_out_E-diamond} for any $j\in\{j_\initial,\ldots,j_\final\}$ to have
    \begin{align}
        \|\hat{\rho}_R\|_1
        &\leq  10\epsilon(j_\final-j_\initial+1)\|{\cal O}'_{|{\cal I}'|:j_\final+1}\|_{\diamond}^{g_{\rm sup}^{\ell-1}(E_{\rm prep})/\epsilon^2} \label{eq:reason_for_multiple_fault_case} \\
        &\leq 10\epsilon L_{\max}\|{\cal O}'_{|{\cal I}'|:j_\final+1}\|_{\diamond}^{g_{\rm sup}^{\ell-1}(E_{\rm prep})/\epsilon^2} , \label{eq:end_no_consecutive}
    \end{align}
    where we used $j_\final-j_\initial+1=|R'_{i_\initial}|\leq L_{\max}$ from the definition of $L_{\max}$.
    Here, we notice from Eqs.~\eqref{eq:putting_out_E-diamond} and Eq.~\eqref{eq:end_no_consecutive} that the same line of argument to obtain an upper bound applies to $i_{\initial + 1}\in R$ and so on.  Thus, we inductively prove 
    \begin{equation}
        \|\hat{\rho}_R\|_1 \leq (10\epsilon L_{\max})^{|R|}.
    \end{equation}
    Therefore, the circuit $\tilde{C}$ experiences the local Markovian noise model with the noise strength upper-bounded by $10\epsilon L_{\max}\eqqcolon \epsilon_{\rm qubit}$.
\end{proof}

One may notice that the analysis above does not fully exploit the advantage of using ExRecs since all the ExRecs just become good if $\epsilon$ is zero.  It is also possible to consider a more general setting where parameters $s_{\rm p}$, $s_{\rm g}$, and $s_{\rm m}$ can be random variables~\cite{Gottesman2009}, and the fault paths that violate the conditions~\eqref{eq:EC_from_pgm}--\eqref{eq:EC+meas} in a given GKP ExRec are regarded as bad.  Extending the techniques developed in this paper to such a generalized analysis is left to future work, but our analysis clarifies that taking nonzero $s_{\rm p}$, $s_{\rm g}$, and $s_{\rm m}$ satisfying these conditions suffices to show that $\tilde{C}$ experiences the local Markovian noise model.

By combining the threshold theorem of a qubit concatenated code with the local Markovian noise in Ref.~\cite{Aliferis2013}, we reach the following conclusion.
\begin{corollary}[FT threshold under an $(s,\epsilon)$- and $(E,s,\epsilon)$-independent Markovian noise model]\label{cor:threshold_theorem}
Consider implementing an original circuit $C$ on qubits using an FT-GKP circuit $C'$ on CV systems.
Suppose all preparations satisfy the $E_{\rm prep}$-energy constraint and all gates satisfy the $g_{\rm sup}$-energy constraint (Defs.~\ref{def:energy_constrained_prep} and~\ref{def:energy_constrained_gate}). 
For $0<\epsilon<1$, consider a family of noise models on $C'$ parameterized by $(s_{\rm p},s_{\rm g}, s_{\rm m}; \epsilon)$ in which $(s_{\rm p},\epsilon)$-independent Markovian noise is applied to the GKP preparation gadgets, $(E_{\max}^{\ell}(\epsilon),s_{\rm g},\epsilon)$-independent Markovian noise is applied to the GKP gate gadgets (including waits), and $(E_{\max}^{\ell}(\epsilon),s_{\rm m},\epsilon)$-independent Markovian noise is applied to the GKP measurement gadgets (Defs.~\ref{def:s_eps_Markovian_prep}--\ref{def:E_s_eps_Markovian_meas}).  Then, there exists a fault-tolerant threshold $\epsilon_{\rm th}>0$ such that for any $\epsilon\leq\epsilon_{\rm th}$ and any choice of the parameters $(s_{\rm p},s_{\rm g}, s_{\rm m})$ satisfying Eqs.~\eqref{eq:EC_from_pgm}--\eqref{eq:EC+meas}, we can achieve fault-tolerant quantum computation using a (noisy) CV circuit $C'$ under the independent Markovian noise model parameterized by $(s_{\rm p},s_{\rm g}, s_{\rm m}; \epsilon)$.  In particular, there exists a threshold $s_\mathrm{th}>0$ such that if $\epsilon<\epsilon_{\rm th}$ and $s_{\rm p},s_{\rm g}, s_{\rm m}<s_{\rm th}$, then we can achieve the fault-tolerant quantum computation.
\end{corollary}

\begin{proof}
    Since the qubit FT circuit has a fault-tolerant threshold $\epsilon_{\rm qubit}^{\star}$ against local Markovian noise~\cite{Aliferis2013}, by applying it to $\tilde{C}$ in Theorem~\ref{theo:level_reduction}, the circuit $C'$ can be made fault tolerant for any choice $(s_{\rm p},s_{\rm g}, s_{\rm m}; \epsilon)$ of parameters as long as the triple $(s_{\rm p},s_{\rm g}, s_{\rm m})$ satisfies Eqs.~\eqref{eq:EC_from_pgm}--\eqref{eq:EC+meas} and $10\epsilon L_{\max} < \epsilon_{\rm qubit}^{\star}$, i.e., $\epsilon < \epsilon_{\rm th}$ with $\epsilon_{\rm th}\coloneqq \epsilon_{\rm qubit}^{\star}/(10 L_{\max})$.  The last statement directly follows by substituting $s_{\rm p}=s_{\rm g}=s_{\rm m}=s_{\rm th}\coloneqq c/38$ to Eqs.~\eqref{eq:EC_from_pgm}--\eqref{eq:EC+meas}, where $c$ is in Eq.~\eqref{eq:C}.
\end{proof}

\bibliography{reference}

\end{document}